\documentclass{article}
\usepackage[preprint]{neurips_2025}

\usepackage[utf8]{inputenc} 
\usepackage[T1]{fontenc}    
\usepackage{hyperref}       
\usepackage{url}            
\usepackage{booktabs}       
\usepackage{amsfonts}       
\usepackage{nicefrac}       
\usepackage{microtype}      
\usepackage{xcolor}         

\title{Learning-Augmented Online Bipartite Fractional Matching}

\author{%
  Davin Choo\thanks{Part of work done while the author was affiliated with the National University of Singapore, Singapore.}$\hspace{5pt}$\footnotemark[4] \\
  Harvard John A. Paulson School Of Engineering And Applied Sciences \\
  Harvard University \\
  Boston, Massachusetts, USA \\
  \texttt{davinchoo@seas.harvard.edu} \\
  \And
  Billy Jin\thanks{Work done while the author was at the University of Chicago Booth School of Business, USA. }$\hspace{5pt}$\footnotemark[4] \\
  Daniels School of Business \\
  Purdue University\\
  West Lafayette, Indiana, USA \\
  \texttt{jin608@purdue.edu} \\
  \And
  Yongho Shin\thanks{Part of work done while the author was affiliated with Yonsei University, South Korea.}$\hspace{5pt}$\thanks{Equal contribution.} \\
  Institute of Computer Science \\
  University of Wrocław \\
  Wrocław, Poland  \\
  \texttt{yongho@cs.uni.wroc.pl} \\
}

\usepackage{graphicx}
\graphicspath{ {../fig/} }

\usepackage{caption}
\usepackage{subcaption}

\usepackage{tikz}
\usetikzlibrary{decorations.pathreplacing}

\usepackage{amsmath, amssymb}

\usepackage{enumitem}

\usepackage[ruled, noend, linesnumbered]{algorithm2e}
\DontPrintSemicolon
\SetKwInput{KwIIn}{Initial input}
\SetKwInput{KwOIn}{Online input}
\SetCommentSty{textrm}

\definecolor{darkgreen}{rgb}{0,0.5,0}
\hypersetup{
    unicode=false,          
    colorlinks=true,        
    linkcolor=red,          
    citecolor=darkgreen,    
    filecolor=magenta,      
    urlcolor=blue           
}
\usepackage[capitalize, nameinlink]{cleveref}

\makeatletter
\AddToHook{cmd/appendix/before}{\def\cref@section@alias{appendix}\def\cref@subsection@alias{appendix}}
\makeatother

\usepackage{amsthm}
\usepackage{thm-restate}
\theoremstyle{plain}
\newtheorem{theorem}{Theorem}
\newtheorem{proposition}[theorem]{Proposition}
\newtheorem{lemma}[theorem]{Lemma}

\theoremstyle{definition}
\newtheorem{definition}[theorem]{Definition}

\newtheorem{observation}[theorem]{Observation}
\theoremstyle{remark}

\crefname{observation}{Observation}{Observations}
\Crefname{observation}{Observation}{Observations}

\usepackage{color-edits}

\newcommand{\cC}{\mathcal{C}}

\newcommand{\cG}{\mathcal{G}}

\newcommand{\cM}{\mathcal{M}}

\newcommand{\cR}{\mathcal{R}}

\newcommand{\E}{\mathbb{E}}

\newcommand{\R}{\mathbb{R}}

\newcommand{\wh}{\widehat}

\newcommand{\Balance}{\textsc{Balance}}
\newcommand{\Coinflip}{\textsc{CoinFlip}}
\newcommand{\Waterfill}{\textsc{Waterfilling}}

\newcommand{\LearnAugBal}{\textsc{LearningAugmentedBalance}}
\newcommand{\LAB}{\textsc{LAB}}
\newcommand{\PushAndWater}{\textsc{PushAndWaterfill}}
\newcommand{\PAW}{\textsc{PAW}}

\newcommand{\Aup}[1]{A^{(#1)}}
\newcommand{\Xup}[1]{X^{(#1)}}

\newcommand{\Z}{\mathbb{Z}}

\newcommand{\region}[1]{\mathcal{D}_\mathsf{#1}}
\newcommand{\trajcost}[1]{\mathsf{cost}(#1)}

\newcommand{\robadv}{\cR}
\newcommand{\conadv}{\cC}

\newcommand{\tgtalg}{\cM}
\newcommand{\lev}[1]{d^{(#1)}}

\newcommand{\modalg}{\overline{\cM}}
\newcommand{\modlev}[1]{\overline{d}^{(#1)}}
\newcommand{\modx}{\overline{x}}
 \newcommand{\modavg}{\overline{d}}

\newcommand{\tstar}{t^{\star}}
\newcommand{\pat}[1]{p^{(#1)}}
\newcommand{\qat}[1]{q^{(#1)}}
\newcommand{\tsat}{t_s}

\newcommand{\xbar}{\overline{x}}
\newcommand{\dbar}{\overline{d}}
\newcommand{\roblev}[1]{\ell^{(#1)}}

\newcommand{\ALG}{\mathsf{ALG}}
\newcommand{\OPT}{\mathsf{OPT}}
\newcommand{\ADVICE}{\mathsf{ADVICE}}
\newcommand{\DUAL}{\mathsf{DUAL}}

\addauthor{billy}{magenta}

\newcommand{\ALGF}{\mathsf{ALG}_{\mathsf{frac}}}
\newcommand{\ALGI}{\mathsf{ALG}_{\mathsf{int}}}
\newcommand{\I}{\mathbb{I}}
\newcommand{\Var}{\mathrm{Var}}

\begin{document}

\maketitle

\begin{abstract}
Online bipartite matching is a fundamental problem in online optimization, extensively studied both in its integral and fractional forms due to its theoretical significance and practical applications, such as online advertising and resource allocation.
Motivated by recent progress in learning-augmented algorithms, we study online bipartite fractional matching when the algorithm is given advice in the form of a suggested matching in each iteration.
We develop algorithms for both the vertex-weighted and unweighted variants that provably dominate the na\"{i}ve ``coin flip'' strategy of randomly choosing between the advice-following and advice-free algorithms.
Moreover, our algorithm for the vertex-weighted setting extends to the AdWords problem under the small bids assumption, yielding a significant improvement over the seminal work of Mahdian, Nazerzadeh, and Saberi (EC 2007, TALG 2012).
Complementing our positive results, we establish a hardness bound on the robustness-consistency tradeoff that is attainable by any algorithm.
We empirically validate our algorithms through experiments on synthetic and real-world data.
\end{abstract}

\section{Introduction}
\label{sec:intro}

Online bipartite matching is a fundamental problem in online optimization with significant applications in areas such as online advertising \cite{mehta2007adwords, feldman2009onlinead}, resource allocation \cite{devanur2019near}, and ride-sharing platforms \cite{dickerson2021allocation,feng2024two}. 
In its classical formulation \cite{karp1990optimal, aggarwal2011online}, the input is a bipartite graph where one side of (possibly weighted) \emph{offline} vertices is known in advance, while the other side of \emph{online} vertices arrives sequentially one at a time.
When an online vertex $v$ arrives, its incident edges are revealed, and the algorithm irrevocably decides whether to match $v$ and, if so, to which currently unmatched neighbor.
The objective is to maximize the total weight of the matched offline vertices.
Algorithms for online bipartite matching are often evaluated by their \emph{competitive ratio}: An algorithm is \emph{$\rho$-competitive} if it always outputs a matching whose (expected) total weight is at least $\rho$ times the weight of the best matching in hindsight.
In a seminal paper, \cite{karp1990optimal} proposed the \textsc{Ranking} algorithm and showed it is $(1-\nicefrac{1}{e})$-competitive for the unweighted setting.
This competitive ratio is best-possible, and was later extended to the vertex-weighted case by \cite{aggarwal2011online}.

Online bipartite matching has also been studied in the fractional setting, where edges can be fractionally chosen, provided that the total fractional value on the edges incident to any vertex does not exceed one \cite{wang2015two, huang2019tight, huang2020fully2, hosseini2024class}.
Fractional matching is important both theoretically and practically.
It naturally models settings where  online arrivals are divisible or offline vertices have large capacities \cite{kalyanasundaram2000optimal, mehta2007adwords, buchbinder2007online, feldman2009onlinead, mahdian2012online, devanur2016whole, feng2024batching}, and it forms the basis for designing integral algorithms using rounding techniques \cite{feldman2016online, buchbinder2023lossless, naor2025online}.
For fractional vertex-weighted online bipartite matching, the \textsc{Balance} algorithm of \cite{buchbinder2007online} gets a competitive ratio of $(1-\nicefrac{1}{e})$, which is best-possible and matches the ratio in the integral case. 

The main challenge in online bipartite matching is that irrevocable decisions must be made without knowledge of future arrivals.
Uncertainty in the arrival sequence is typically modeled either \emph{adversarially} or \emph{stochastically}.
The adversarial model assumes no structure and measures worst-case performance, but can be overly pessimistic.
On the other hand, the stochastic model assumes arrivals are drawn from a known distribution \cite{feldman2009online}, but such distributions are often estimated and may be inaccurate.
These models thus represent two extremes, each with practical limitations.
A middle ground is offered by \emph{algorithms with predictions}, or \emph{learning-augmented algorithms} \cite{mahdian2012online, lykouris2021competitive}, which incorporates advice -- derived from data, forecasts, or experts -- of unknown quality.
The performance is typically measured in terms of its \emph{robustness} (guaranteed performance regardless of advice quality) and \emph{consistency} (performance when advice is accurate) \cite{lykouris2021competitive, kumar2018improving}.\footnote{A third property, \emph{smoothness}, requires graceful degradation with advice quality \cite{elenter2024overcoming}. See \cref{sec:prelim}.} 
In online bipartite matching, an algorithm is $r$-robust if its competitive ratio is at least $r$, and $c$-consistent if it achieves at least a $c$-fraction of the total weight from following the advice (see \cref{def:rc}).
A natural baseline is the \Coinflip{} algorithm, which randomly chooses between robustness- and consistency-optimal strategies.
For matching, its tradeoff curve is the line segment between $(1 - \nicefrac{1}{e}, 1 - \nicefrac{1}{e})$ and $(0,1)$ in the vertex-weighted case, or $(\nicefrac{1}{2}, 1)$ in the unweighted case \cite{jin2022online}.

This paper investigates the robustness-consistency tradeoff of online bipartite matching under the learning-augmented framework, building on prior work including \cite{mahdian2007allocating, mahdian2012online, aamand2022optimal, jin2022online, spaeh2023online, choo2024online}.
Particularly relevant are the works of Mahdian et al. \cite{mahdian2007allocating,mahdian2012online} and Spaeh and Ene \cite{spaeh2023online}.
Mahdian et al.\ studied the AdWords problem (introduced in \cite{mehta2007adwords}), with advice in the form of a recommendation assigning each online impression to a specific offline advertiser.
They proposed a learning-augmented algorithm under the \emph{small bids} assumption that outperforms the na\"{i}ve \Coinflip{} strategy, but only over part of the robustness range.
Meanwhile, \cite{spaeh2023online} generalized this result to Display Ads and the generalized assignment problem \cite{feldman2009onlinead}.
However, as shown in \cref{fig:intro:fig}, neither of these algorithms dominate \Coinflip{} across the full robustness spectrum.
This raises a natural question: \emph{Does there exist a learning-augmented algorithm for online bipartite matching that dominates \Coinflip{} across the entire range of robustness?}

\begin{figure}[htb]
\centering
\begin{subfigure}[t]{0.45\textwidth}
    \centering
    \includegraphics[width=\textwidth] {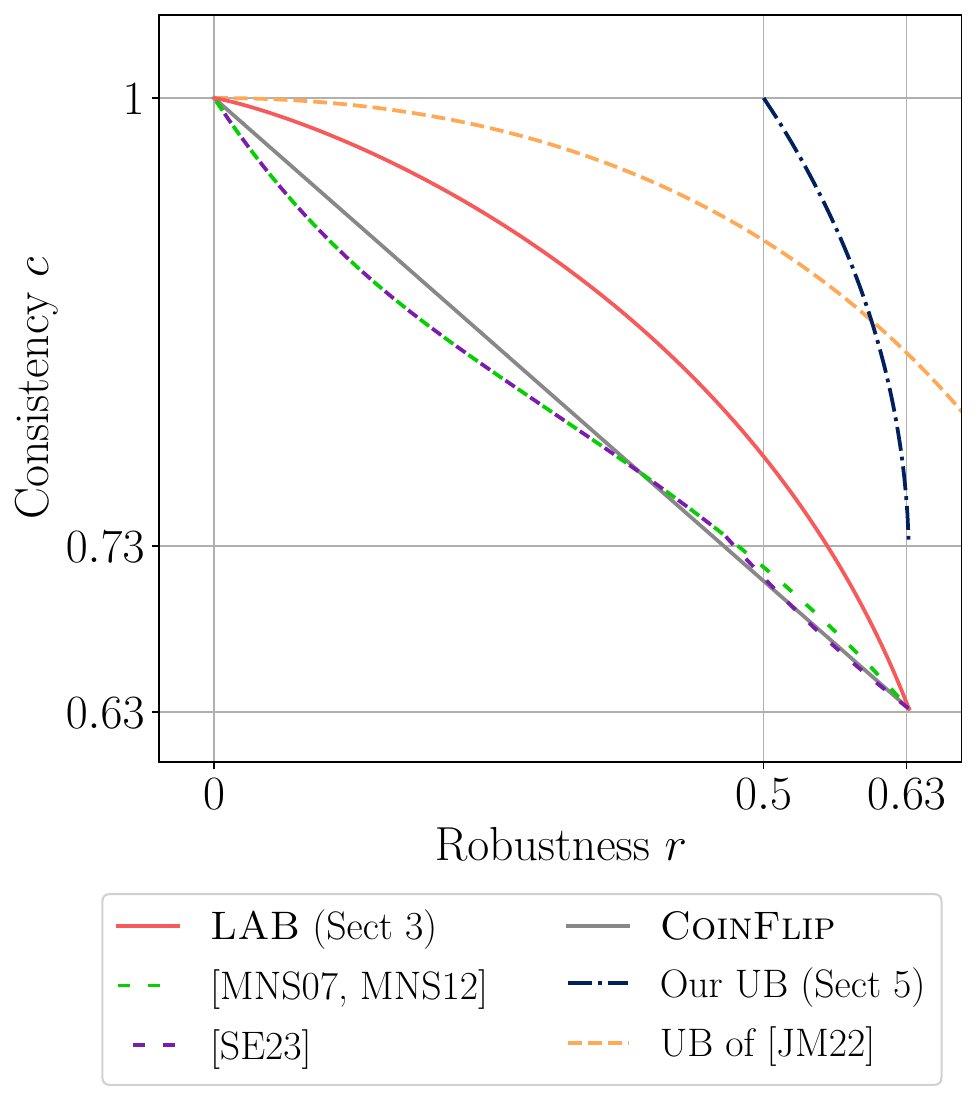}
    \caption{Vertex-weighted setting and AdWords under small bids assumption}
\end{subfigure}
\hfill
\begin{subfigure}[t]{0.45\textwidth}
    \centering
    \includegraphics[width=\textwidth]{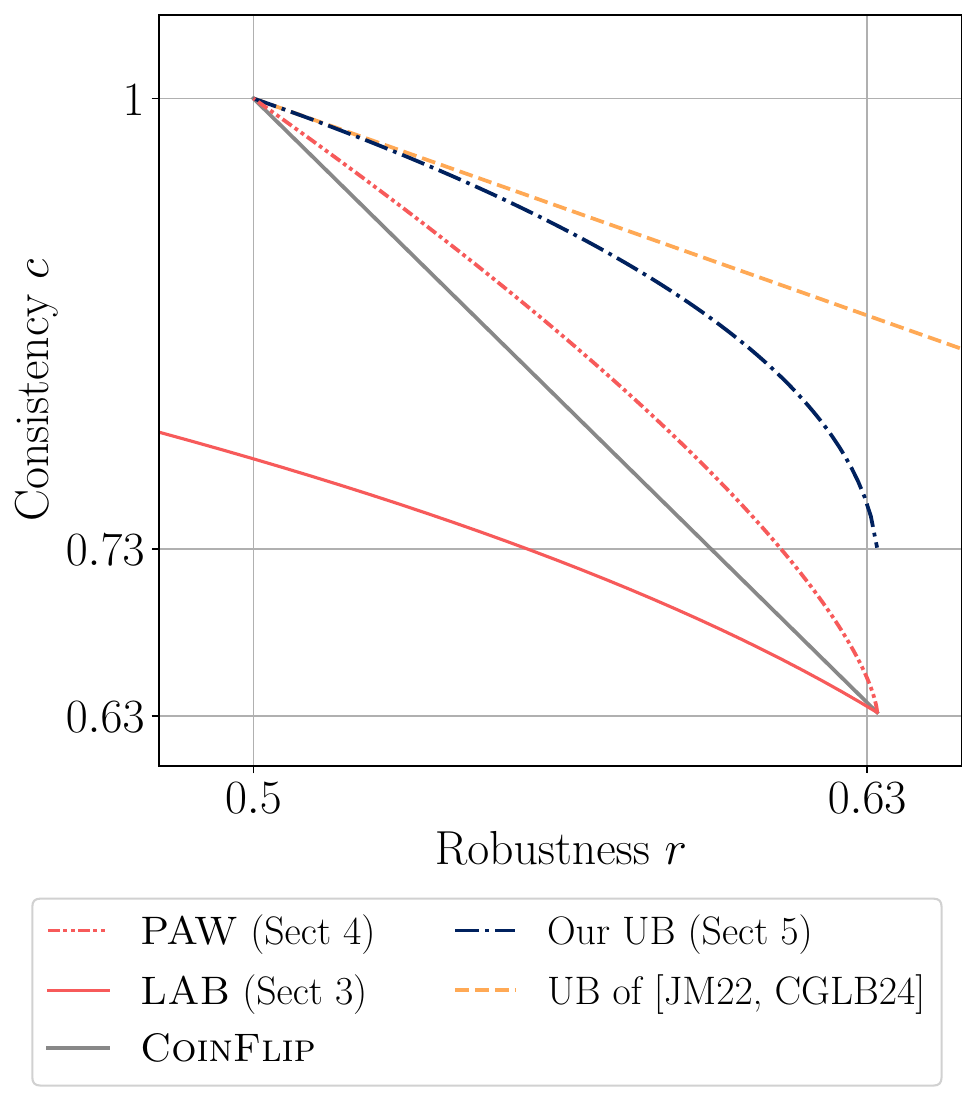}
    \caption{Unweighted setting}
\end{subfigure}
\caption{Robustness-consistency tradeoffs of previous works and our results.}
\label{fig:intro:fig}
\end{figure}

\subsection{Our contributions}

We answer the above question affirmatively by presenting learning-augmented algorithms for both vertex-weighted and unweighted online bipartite fractional matching whose robustness-consistency tradeoffs Pareto-dominate that of \Coinflip{} across the \emph{entire} range of robustness (see \cref{fig:intro:fig}).

Motivated by \cite{mahdian2007allocating, mahdian2012online, spaeh2023online}, we take the advice to be a feasible fractional matching that is revealed in an online fashion: upon arrival of each online vertex $v$, the algorithm is given as advice fractional matching values for each neighboring edge of~$v$.
Moreover, as in \cite{mahdian2007allocating, mahdian2012online, spaeh2023online}, our algorithms are parameterized by a tradeoff parameter $\lambda \in [0, 1]$ that represents how closely we follow the advice.
At the extremes, our algorithms blindly follow the advice when $\lambda = 1$ and revert to \textsc{Balance} when $\lambda = 0$.

For the vertex-weighted setting, we present an algorithm \LearnAugBal{} (\LAB{}) with the following guarantees:
\begin{restatable}{theorem}{vwalgmain}
\label{thm:vertex_weighted}
For any tradeoff parameter $\lambda \in [0,1]$, \LearnAugBal{} is an $r(\lambda)$-robust and $c(\lambda)$-consistent algorithm for vertex-weighted online bipartite fractional matching, where
\begin{align*}
r(\lambda) 
:= 1 - e^{\lambda-1} -\left( e^{\lambda-1} -\lambda \right) \ln (1-\lambda e^{1 - \lambda} ) - \lambda(1-\lambda)
\quad \text{and} \quad
c(\lambda)
:= 1 + \lambda - e^{\lambda - 1}.
\end{align*}
\end{restatable}
This algorithm is based on \Balance{} where the penalty function is modified to be \emph{advice-dependent}.
To analyze this algorithm, we adopt the standard primal-dual analysis of online bipartite matching and first prove its performance when the advice is integral.
We then prove that the robustness and consistency are minimized when the advice is integral, yielding the same guarantees for the general fractional advice case.

We further show that \LAB{} extends to the AdWords problem under the small bids assumption, yielding a significant improvement over \cite{mahdian2007allocating, mahdian2012online}.

\begin{restatable}{theorem}{adwordsmain}
\label{thm:adw:main}
Consider the small bids assumption where the maximum bid-to-budget ratio is bounded by some sufficiently small $\varepsilon > 0$.
For any tradeoff parameter $\lambda \in [0, 1]$, there exists an $r(\lambda) \cdot (1 - 3 \sqrt{\varepsilon \ln (1/\varepsilon)} )$-robust and 
$c(\lambda) \cdot (1 - 3 \sqrt{\varepsilon \ln (1/\varepsilon)})$-consistent algorithm for AdWords with advice, where $r(\lambda)$ and $c(\lambda)$ are the same as in \cref{thm:vertex_weighted}.
\end{restatable}

To achieve this result, we first extend \LAB{} to the \emph{fractional} AdWords setting while preserving its robustness and consistency, and then employ a reduction from \cite{feng2024batching} to reduce the integral AdWords problem to the fractional problem with small loss under the small bids assumption.

Observe in \cref{fig:intro:fig} that the robustness-consistency tradeoff of \LAB{} lies below the linear tradeoff of \Coinflip{} in the \emph{unweighted} setting: the top-left endpoint of the tradeoff for \LAB{} is $(r,c) = (0, 1)$, whereas in the unweighted setting \Coinflip{} can be implemented to be $1/2$-robust even when $c = 1$.
This happens because any maximal matching in an unweighted graph is automatically $\frac12$-robust.
To beat \Coinflip{} in the unweighted setting, a tighter analysis of \LAB{} would be required but this proved difficult using our current analysis framework for \LAB{}, even when the advice is integral.
Instead, we present another algorithm called \PushAndWater{} (\PAW{}) for the unweighted setting with integral advice that circumvents the aforementioned challenge in the analysis.

\begin{restatable}{theorem}{uwalgmain}
\label{thm:uw:main}
For any tradeoff parameter $\lambda \in [0,1]$, \PushAndWater{} is $r(\lambda)$-robust and $c(\lambda)$-consistent for unweighted online bipartite fractional matching with integral advice, where 
\begin{equation*}
    \textstyle r(\lambda) := 1 - \left(1 - \lambda + \nicefrac{\lambda^2}{2} \right) e^{\lambda - 1}
    \;\;\text{and}\;\;
    c(\lambda) := 1 - \left( 1 - \lambda \right) e^{\lambda - 1}.
\end{equation*}
\end{restatable}

\PAW{} is based on the unweighted version of \Balance{}, also known as \Waterfill{}, with one additional step at each iteration where it first increases the fractional value of the currently advised edge until the ``level'' of the advised offline vertex reaches the tradeoff parameter $\lambda$.
We analyze \PAW{} using primal-dual but with a different construction of dual variables from \LAB{}.

We complement our algorithmic results by presenting an upper bound on the robustness-consistency tradeoff of any learning-augmented algorithm for the unweighted setting with integral advice in \cref{sec:ub}, improving upon the previous upper bound results \cite{jin2022online, choo2024online} (see \cref{fig:intro:fig}).
Note that this result implies the same impossibility for more general problems including the vertex-weighted setting and the AdWords problem.
To obtain our hardness result, we construct two adaptive adversaries --- one for robustness and the other for consistency.
The construction of these adversaries is inspired by the standard upper-triangular worst-case instances  \cite{karp1990optimal}, while we modify this construction to make the two adversaries have the same behavior until the first half of the online vertices are revealed.
Due to this modification, the two adversaries are indistinguishable until the halfway point of the execution while inheriting the difficulty from the standard worst-case instances.
We then identify a set of conditions characterizing the behavior of Pareto-optimal algorithms on our hardness instance and solve a factor-revealing LP to upper bound the best possible consistency subject to the constraint on the robustness to be $r$, for each $r \in [\nicefrac{1}{2}, 1-\nicefrac{1}{e}]$.

Lastly, we implemented and evaluated our proposed algorithms \textsc{LAB} and \textsc{PAW} in \cref{sec:exp} against advice-free baselines on synthetic and real-world graph instances, for varying advice quality parameterized by a noise parameter $\gamma$, where larger $\gamma$ indicates poorer advice quality.
As predicted by our analysis, the attained competitive ratios of both \textsc{LAB} and \textsc{PAW} begin at 1 under perfect advice and smoothly degrades as the $\gamma$ increases.
Unsurprisingly, for sufficiently large $\gamma$, the worst case optimal advice-free algorithm \textsc{Balance} outperforms both \textsc{LAB} and \textsc{PAW}.

\subsection{Related work}

\paragraph{Learning-augmented algorithms for online matching.}
In addition to the works of \cite{mahdian2007allocating, mahdian2012online} and \cite{spaeh2023online}, several other papers study learning-augmented algorithms for online bipartite matching.
\cite{antoniadis2020secretary} studied the edge-weighted version under the random arrival model \cite{korula2009algorithms, kesselheim2013optimal}, where the advice estimates the edge-weight that each online vertex is assigned.
\cite{aamand2022optimal} considered a model where the advice predicts the degree of each offline vertex. They analyzed the performance of a greedy algorithm called \textsc{MinPredictedDegree} that uses this advice, within the random graph model of \cite{chung2003spectra}.
\cite{jin2022online} considered the two-stage model of \cite{feng2024two} with the advice modeled as a feasible matching in the first stage.
They developed optimal learning-augmented algorithms for the unweighted, vertex-weighted, edge-weighted, and AdWords variants.
\cite{choo2024online} showed that, in the unweighted setting under adversarial arrival, no algorithm that is 1-consistent can achieve robustness better than $1/2$. However, in the random arrival model, they designed a 1-consistent algorithm that achieves $(\beta - o(1))$-robustness using advice in the form of a histogram of arrival types, where $\beta$ is the competitive ratio of the best advice-free algorithm in the same setting.

\paragraph{Learning-augmented algorithms more broadly.} 
Since the seminal work of \cite{lykouris2021competitive}, there has been a surge of interest in incorporating unreliable advice into algorithm design and analyzing performance as a function of advice quality across various areas of computer science.
This framework has been especially successful in online optimization, where the core challenge lies in handling uncertainty about future inputs. In this context, advice can serve as a useful proxy for the unknown future. Beyond online bipartite matching, a wide range of online optimization problems have been studied under the learning-augmented framework.
Examples include caching and paging \cite{lykouris2021competitive, jiang2022online, im2022parsimonious, bansal2022learning}, ski rental \cite{kumar2018improving, wang2020online, shin2023improved, zhao24learning}, covering problems \cite{bamas2020primal, grigorescu2022learning}, scheduling \cite{lattanzi2020online, azar2021flow, im2023non, benomar24non}, metric or graph problems \cite{azar2022online, antoniadis2023online, sadek24algorithms}, causal graph learning \cite{choo2023active}, and distribution learning \cite{bhattacharyya2025learning,bhattacharyya2025product}.
For an overview of this growing area, we refer the reader to the survey by \cite{mitzenmacher2022algorithms}\footnote{See also \url{https://algorithms-with-predictions.github.io/}.}.

\paragraph{Online matching.} 
The study of online matching began with the seminal work of \cite{karp1990optimal}, who introduced the randomized \textsc{Ranking} algorithm and proved it to be $(1 - 1/e)$-competitive for online bipartite integral matching, the best possible in this setting.
Due to its foundational importance, the analysis of \textsc{Ranking} has been revisited and extended in numerous subsequent works, including \cite{goel2008online, birnbaum2008line, devanur2013randomized, eden2021economics}. The online matching problem has since been studied under a variety of extensions and settings, such as the edge-weighted case \cite{fahrbach2022edge, shin21making, gao2022improved, blanc2022multiway}, ad allocation \cite{mehta2007adwords, buchbinder2007online, goel2008online, feldman2009onlinead, huang2024adwords}, random and stochastic arrivals \cite{feldman2009online, korula2009algorithms, karande2011online, kesselheim2013optimal, jin2021improved, huang2021online, huang2022power}, and two-sided or general arrival models \cite{gamlath2019online, huang2020fully, huang2020fully2}. For a comprehensive overview of the field, we refer interested readers to the surveys by \cite{mehta2013online} and \cite{huang2024online}.

\paragraph{Paper outline.}
We begin with preliminaries such as definitions and notation used in the paper in \cref{sec:prelim}.
We then present our algorithmic results in the following two sections.
In \cref{sec:vwalg}, we present the \LAB{} algorithm for the vertex-weighted setting with fractional advice and then show that this algorithm extends to AdWords with small bids.
\cref{sec:uwalg} presents the \PAW{} algorithm for the unweighted setting with integral advice.
We then provide in \cref{sec:ub} our upper bound result on the robustness-consistency tradeoff for the unweighted setting with integral advice.
The experimental results are  in \cref{sec:exp}, followed by concluding remarks in \cref{sec:conclusion}.

\section{Preliminaries}
\label{sec:prelim}

\textbf{Online bipartite matching.}
In the \emph{vertex-weighted online bipartite fractional matching} problem, we have a bipartite graph $ G = (U \cup V, E) $ and a weight $w_u \geq 0$ for each $u \in U$.
If $w_u = 1$ for every $u \in U$, then the problem is called \emph{unweighted}.
The vertices in $U$ are the \emph{offline} vertices, and their weights are known to the algorithm from the very beginning.
On the other hand, the vertices in $V$ are the \emph{online} vertices, and arrive one by one. Whenever $v \in V$ arrives, its neighborhood $N(v) := \{ u \in U \mid (u, v) \in E \}$ is revealed. Since the online vertices arrive sequentially, we use the notation $t \prec v$ to mean that $t$ arrives earlier than $v$. 
Similarly, for each offline vertex $u \in U$, we also use $N(u) := \{v \in V \mid (u, v) \in E \}$ to denote the neighborhood of $u$.

We use the analogy of \emph{waterfilling} to describe the behavior of the algorithm.
When $v \in V$ arrives and its neighborhood $N(v)$ is revealed, the algorithm decides at that moment the amount $x_{u, v}$ of water to send from $v$ to each $u \in N(v)$ subject to the constraints that:
\begin{itemize}
    \item the total amount of water supplied from $v$ does not exceed 1, i.e., $\sum_{u \in N(v)} x_{u, v} \leq 1$;
    \item each offline vertex $u \in U$ can hold at most 1 unit of water, i.e., $\sum_{t \in N(u) : t \preceq v} x_{u, t} \leq 1$.
\end{itemize}
This decision is irrevocable, meaning that, $\{ x_{u, v} \}_{u \in N(v)}$ cannot be modified in the subsequent iterations.
Let $x \in \R^{E}$ be the final solution of the algorithm.
Note that $x$ is a fractional matching in the hindsight graph $G$.
The weight of this solution is defined to be $\sum_{(u, v) \in E} w_u x_{u, v}$.
The objective of this problem is to maximize the weight of the solution.

\textbf{Advice.}
Each online vertex $v \in V$ arrives with a \emph{suggested allocation} $\{a_{u, v}\}_{u \in N(v)}$, where we assume $a = \{a_{u,v}: (u,v) \in E\} \in \mathbb{R}^E$ is a feasible fractional matching in the hindsight graph $G$. 

\paragraph{AdWords.}
In the AdWords problem, the offline vertices $U$ are called \emph{advertisers}, and the online vertices $V$ are called \emph{impressions}.  
Each advertiser $u \in U$ starts with a budget $B_u \geq 0$.  
When an impression $v \in V$ arrives, each advertiser $u \in U$ submits a bid $b_{u,v} \geq 0$ for this impression.  
The algorithm then irrevocably assigns the impression to one of the advertisers.\footnote{By introducing an auxiliary advertiser who always bids 0 on every impression, we can assume without loss of generality that the algorithm assigns every impression.}  
If impression $v$ is assigned to advertiser $u$, the algorithm earns a revenue of $b_{u,v}$, provided that $u$ can afford the bid from its remaining budget.  
The objective is to maximize the total revenue earned by the algorithm.

We remark that vertex-weighted online bipartite matching is the special case of AdWords where $b_{u, v} = B_u$ if $(u,v) \in E$, and $b_{u,v} = 0$ if $(u,v) \not\in E$. 

Integral AdWords is commonly studied under the \emph{small bids} assumption~\cite{mehta2007adwords, buchbinder2007online, mahdian2012online}.
Under this assumption, the bid-to-budget ratio is assumed to be bounded by some small $\varepsilon > 0$.  In other words, for all $u \in U$ and $v \in V$, we have  $b_{u, v} \leq \varepsilon B_u$.

\textbf{Performance measures.}
Denote the value of the final output of an algorithm by $\ALG$, the value of an optimal solution in the hindsight instance by $\OPT$, and the value obtained by the advice by $\ADVICE$. We can then formally define the \emph{robustness} and \emph{consistency} of a learning-augmented algorithm.

\begin{definition}[Robustness and Consistency]
\label{def:rc}
For some $r \in [0, 1]$, we say an algorithm is \emph{$r$-robust} if $\E[\ALG] \geq r \cdot \OPT$ for any instance of the problem.
On the other hand, for some $c \in [0, 1]$, we say an algorithm is \emph{$c$-consistent} if $\E[\ALG] \geq c \cdot \ADVICE$ for any instance of the problem.
\end{definition}

Notice that, when we define the \emph{error} of the advice to be $\eta := \ADVICE / \OPT \in [0, 1]$, the consistency implies the  \emph{smoothness} of the algorithm since we have $\E[\ALG] \geq c \eta \cdot \OPT$.

\textbf{Primal-dual analysis.}
To prove the robustness and consistency of our algorithms, we adopt the standard primal-dual analysis for online bipartite matching \cite{devanur2013randomized}.
Observe that, for vertex-weighted bipartite matching, the primal and dual LPs are formulated as follows:

\begin{minipage}{\textwidth}
\centering
\begin{minipage}[t]{0.49 \textwidth}
    \begin{align*}
    \max\; & \textstyle \sum_{(u, v) \in E} w_u x_{u,v} \\
    \text{s.t.}\; & \textstyle \sum_{v \in N(u)} x_{u, v} \leq 1, & \forall u \in U, \\
    & \textstyle \sum_{u \in N(v)} x_{u, v} \leq 1, & \forall v \in V, \\
    & x_{u, v} \geq 0, & \forall (u, v) \in E; \\
    \end{align*}
\end{minipage}
\begin{minipage}[t]{0.49 \textwidth}
    \begin{align*}
    \min \; & \textstyle \sum_{u \in U} \alpha_u + \sum_{v \in V} \beta_v \\
    \text{s.t.} \; & \alpha_u + \beta_v \geq w_u, & \forall (u, v) \in E, \\
    & \alpha_u \geq 0, & \forall u \in U, \\
    & \beta_v \geq 0, & \forall v \in V.
    \end{align*}
\end{minipage}
\end{minipage}

The following lemma is the cornerstone of the primal-dual analysis.

\begin{lemma}[see, e.g., \cite{devanur2013randomized, fahrbach2022edge}] \label{lem:pre:pda}
Let $x \in \R^E_+$ be a feasible fractional matching output by an algorithm.
For some $\rho \in [0, 1]$, if there exists $(\alpha, \beta) \in \R^U_+ \times \R^V_+$ satisfying
\begin{itemize}
    \item (reverse weak duality) $\sum_{(u, v) \in E} w_u x_{u, v} \geq \sum_{u \in U} \alpha_u + \sum_{v \in V} \beta_v$ and
    \item (approximate dual feasibility) $\alpha_u + \beta_v \geq \rho \cdot w_u \;$ for every $(u, v) \in E$,
\end{itemize}
we have $\ALG \geq \rho \cdot \OPT$.
\end{lemma}

\begin{proof}
    Observe that $\left( \nicefrac{\alpha}{\rho}, \nicefrac{\beta}{\rho} \right) \in \R^U \times \R^V$ is feasible to the dual LP. We thus have
    \[
        \ALG
        = \sum_{(u, v) \in E} w_u x_{u, v}
        \geq \sum_{u \in U} \alpha_u + \sum_{v \in V} \beta_v
        = \rho \cdot \left[ \sum_{u \in U} \nicefrac{\alpha_u}{\rho} + \sum_{v \in V} \nicefrac{\beta_v}{\rho} \right]
        \geq \rho \cdot \OPT,
    \]
    where the last inequality comes from the weak duality of LP.
\end{proof}

Similarly, the following are an LP relaxation of AdWords and its dual LP:

\begin{minipage}{\textwidth}
    \centering
    \begin{minipage}[t]{0.49 \textwidth}
        \begin{align*}
            \max\; & \textstyle \sum_{u \in U} \sum_{v \in V} b_{u, v} x_{u,v} \\
            \text{s.t.}\; & \textstyle \frac{1}{B_u} \sum_{v \in V} b_{u,v} x_{u, v} \leq 1, & \forall u \in U, \\
            & \textstyle  \sum_{u \in U} x_{u, v} \leq 1, & \forall v \in V, \\
            & x_{u, v} \geq 0, & \forall u \in U, v \in V; \\
        \end{align*}
    \end{minipage}
    \begin{minipage}[t]{0.49 \textwidth}
        \begin{align*}
            \min \; & \textstyle \sum_{u \in U} \alpha_u + \sum_{v \in V} \beta_v \\
            \text{s.t.} \; & \textstyle  \frac{b_{u, v}}{B_u} \alpha_u + \beta_v \geq b_{u, v}, &  \forall u \in U, v \in V, \\
            & \alpha_u \geq 0, & \forall u \in U, \\
            & \beta_v \geq 0, & \forall v \in V.
        \end{align*}
    \end{minipage}
\end{minipage}

The below lemma is an adaptation of \cref{lem:pre:pda} to AdWords.
The proof is omitted.

\begin{lemma} \label{lem:pre:adw:pda}
    Let $x \in \R^{U \times V}$ be a feasible solution to the primal LP relaxation of AdWords.
    For some $\rho \in [0, 1]$, if there exists $(\alpha, \beta) \in \R^U \times \R^V$ satisfying
    \begin{itemize}
        \item (reverse weak duality) $ \sum_{u \in U} \sum_{v \in V} b_{u, v} x_{u, v} \geq \sum_{u \in U} \alpha_u + \sum_{v \in V} \beta_v$ and
        \item (approximate dual feasibility) $\frac{b_{u, v}}{B_u} \alpha_u + \beta_v \geq \rho \cdot b_{u, v} \;$ for every $u \in U$ and $v \in V$,
    \end{itemize}
    we have $\sum_{u \in U} \sum_{v \in V} b_{u, v} x_{u, v} \geq \rho \cdot \OPT$.
\end{lemma}

\paragraph{Lambert $W$ function.}
In the definition of \LAB{} in \cref{sec:vwalg}, we use the (principal branch of) Lambert $W$ function. Define $W : \left[ - \frac{1}{e}, \infty \right) \to [-1, \infty)$ to be the inverse function of $y e^y$ on $[-1, \infty)$, i.e., for any $z \in \left[ - \frac{1}{e}, \infty \right)$, we have
\begin{equation} \label{eq:pre:Wprop0}
    W(z) \cdot e^{W(z)} = z.
\end{equation}
It is known that $W$ is increasing on $\left[ - \frac{1}{e}, \infty \right)$.
Its derivative can be written as follows:
\begin{equation} \label{eq:pre:Wprop1}
    W'(x) = \frac{W(x)}{x(1 + W(x))}.
\end{equation}
Lastly, the next equality can easily be derived from \cref{eq:pre:Wprop0}:
\begin{equation} \label{eq:pre:Wprop2}
    \frac{W(z)}{z} = e^{-W(z)}.
\end{equation}
\section{Vertex-weighted matching with advice}
\label{sec:vwalg}

We now present our algorithm \LearnAugBal{} (\LAB{}) for vertex-weighted online bipartite matching with advice and provide a proof sketch showing that it achieves the robustness-consistency tradeoff stated in \cref{thm:vertex_weighted}.
Detailed pseudocode is given in \cref{sec:appendix-pseudocodes} and a full analysis is provided in the supplementary material.

\textbf{Algorithm description.}
Given a tradeoff parameter $\lambda \in [0, 1]$, we define $f_0 : [0, 1] \to [0, 1]$ and $f_1 : [0, 1] \to [0, 1]$ as follows, where $W$ is the Lambert $W$ function:
\begin{equation}
f_0(z) := \min \{ e^{z+\lambda-1}, 1 \},
\quad \text{and} \quad
f_1(z) :=
\begin{cases}
\frac{e^{\lambda-1} - \lambda}{1-z}, & \text{ if } z \in [0, \lambda e^{1-\lambda}), \\
\frac{-\lambda}{W(-\lambda e^{1-\lambda - z})}, & \text{ if } z \in [\lambda e^{1-\lambda}, 1), \\
1, & \text{ if } z = 1,
\end{cases}
\label{eq:vw:penaltyint} 
\end{equation}

Based on these functions, we define $f : [0, 1]^2 \to [0, 1]$ such that
\begin{equation} \label{eq:vw:penaltyfrac}
f(A, X) :=
\begin{cases}
f_1(X), & \text{ if } A > X, \\
\max\{ f_0(X-A), f_1(X) \}, & \text{ if } A \leq X.
\end{cases}
\end{equation}

For clarity, let us describe \LAB{} as a continuous process.
Upon the arrival of each online vertex $v \in V$ along with the advice $\{a_{u, v}\}_{u \in N(v)}$, define $A_u := \sum_{t \in N(u) : t \preceq v} a_{u, t}$ as the total advice-allocated amount to each offline vertex $u \in N(v)$, up to and including $v$.
\LAB{} then continuously pushes an infinitesimal unit of flow from $v$ to the neighbor $u \in N(v)$ maximizing $w_u (1 - f(A_u, X_u))$, where $X_u$ is the total amount 
allocated to $u$ by the algorithm right before it starts pushing this infinitesimal unit of flow.
This continues until $v$ is fully matched (i.e.\ one unit of flow is pushed) or all neighbors are saturated.

\paragraph{Intuition behind the algorithm.}
First, we give intuition for the algorithm.
For an online vertex $v$ and an offline neighbor $u \in N(v)$, the amount allocated from $v$ to $u$ should depend on three factors.
Firstly, a higher $w_u$ should lead to larger $x_{u,v}$.
Secondly, the more $u$ is filled, the less desirable it is to allocate to it further, preserving capacity for future vertices.
Thirdly, vertices favored by the advice should receive more allocation.

The classical \Balance{} algorithm handles the first two factors by choosing the offline vertex with the highest potential value $w_u (1 - g(X_u))$ via a convex increasing penalty function $g(z) = e^{z - 1}$.
To incorporate the third factor, \LAB{} introduces an advice-aware penalty function $f(A, X)$ that also depends on the total advice allocation $A$; see \cref{fig:vw:fig}.
This function is increasing in $X$ (penalizing already-filled vertices) and decreasing in $A$ (lower penalty for vertices recommended by the advice), thereby encouraging alignment with the advice.

The penalty function $f$ used by our algorithm is defined in \cref{eq:vw:penaltyfrac} based on the functions $f_0$ and $f_1$ from \cref{eq:vw:penaltyint}.
While $f_0$ and $f_1$ are derived from the primal-dual analysis, and their exact forms are not crucial for intuition, the structure of $f$ admits a natural interpretation.
Intuitively, if an offline vertex $u$ has received less allocation than the advice suggests (i.e., $A_u > X_u$), then the penalty function treats $u$ as if it were already saturated under the advice. Conversely, if $u$ has been filled beyond the advised amount (i.e., $A_u \leq X_u$), then the penalty effectively treats the excess allocation $X_u - A_u$ as if it were added despite the advice indicating $u$ should be unmatched.

\begin{figure}
    \centering
    \begin{subfigure}{0.32\textwidth}
        \centering
        \includegraphics[width=\textwidth]  {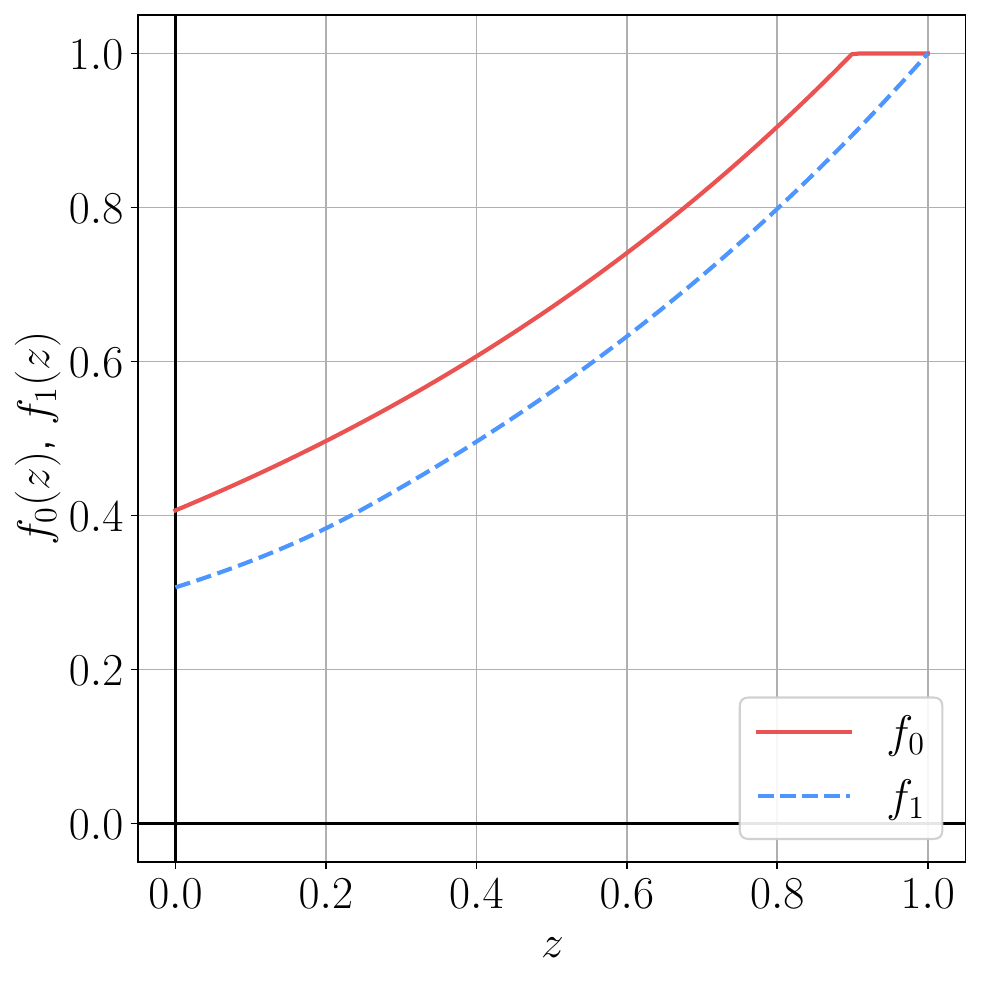}
        \caption{$f_0$, $f_1$ with $\lambda = 0.1$}
    \end{subfigure}
    \begin{subfigure}{0.32\textwidth}
        \centering
        \includegraphics[width=\textwidth]{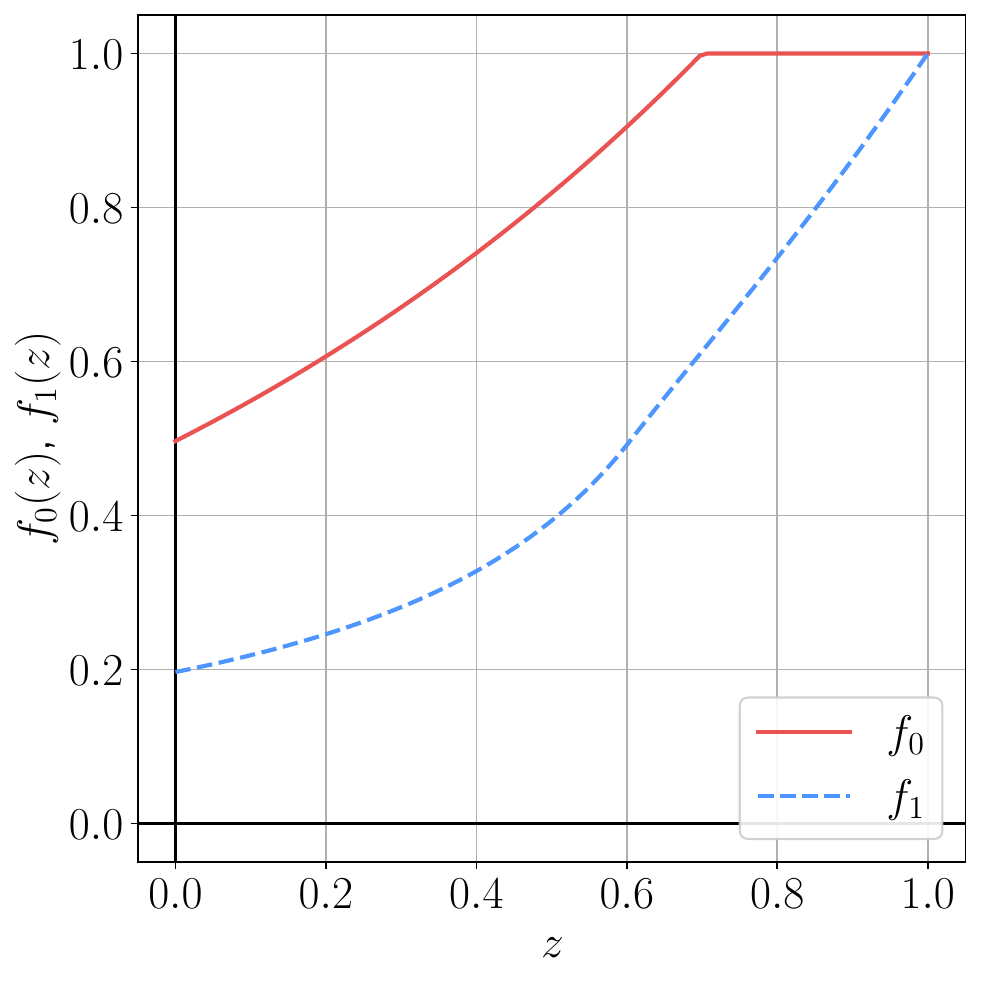}
        \caption{$f_0$, $f_1$ with $\lambda = 0.3$}
    \end{subfigure}
    \begin{subfigure}{0.32\textwidth}
        \centering
        \includegraphics[width=\textwidth]{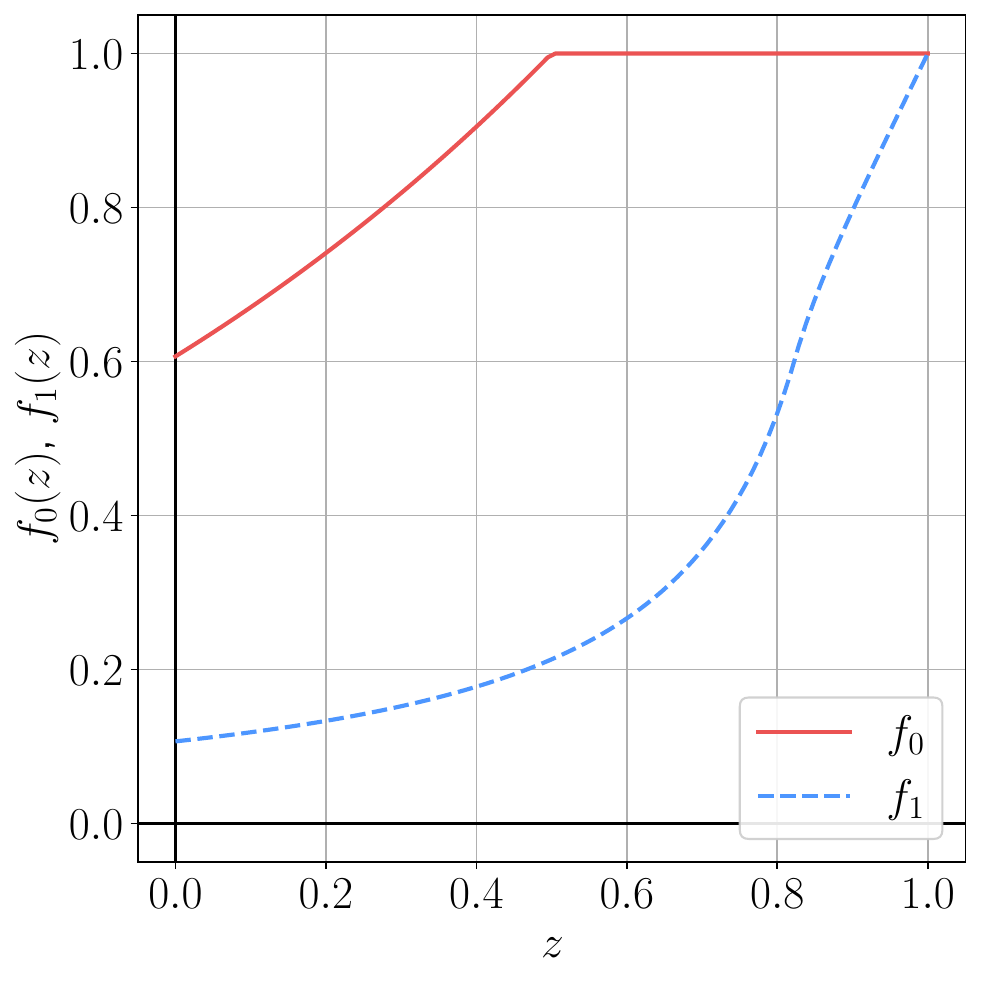}
        \caption{$f_0$, $f_1$ with $\lambda = 0.5$}
    \end{subfigure}
    \\
    \vspace{10pt}
    \begin{subfigure}{0.32\textwidth}
        \centering
        \includegraphics[width=\textwidth]{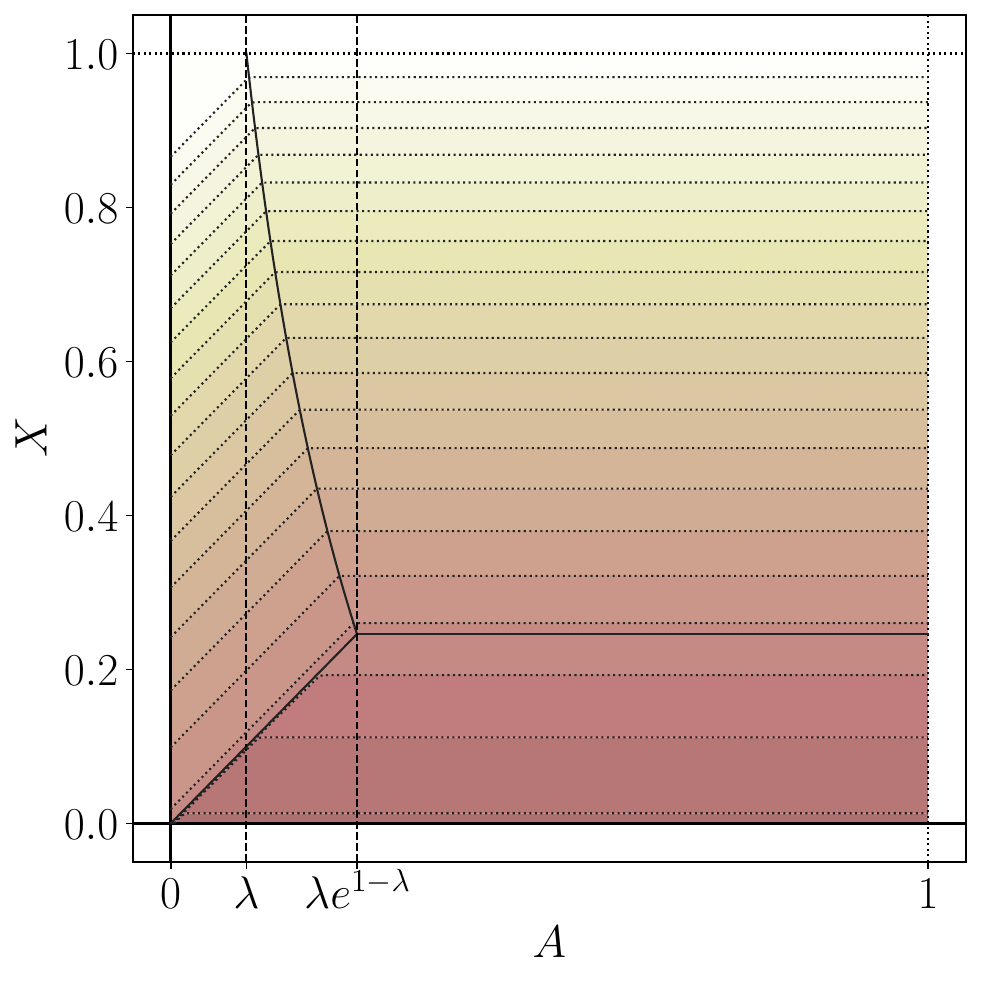}
        \caption{$f$ with $\lambda = 0.1$}
    \end{subfigure}
    \begin{subfigure}{0.32\textwidth}
        \centering
        \includegraphics[width=\textwidth]{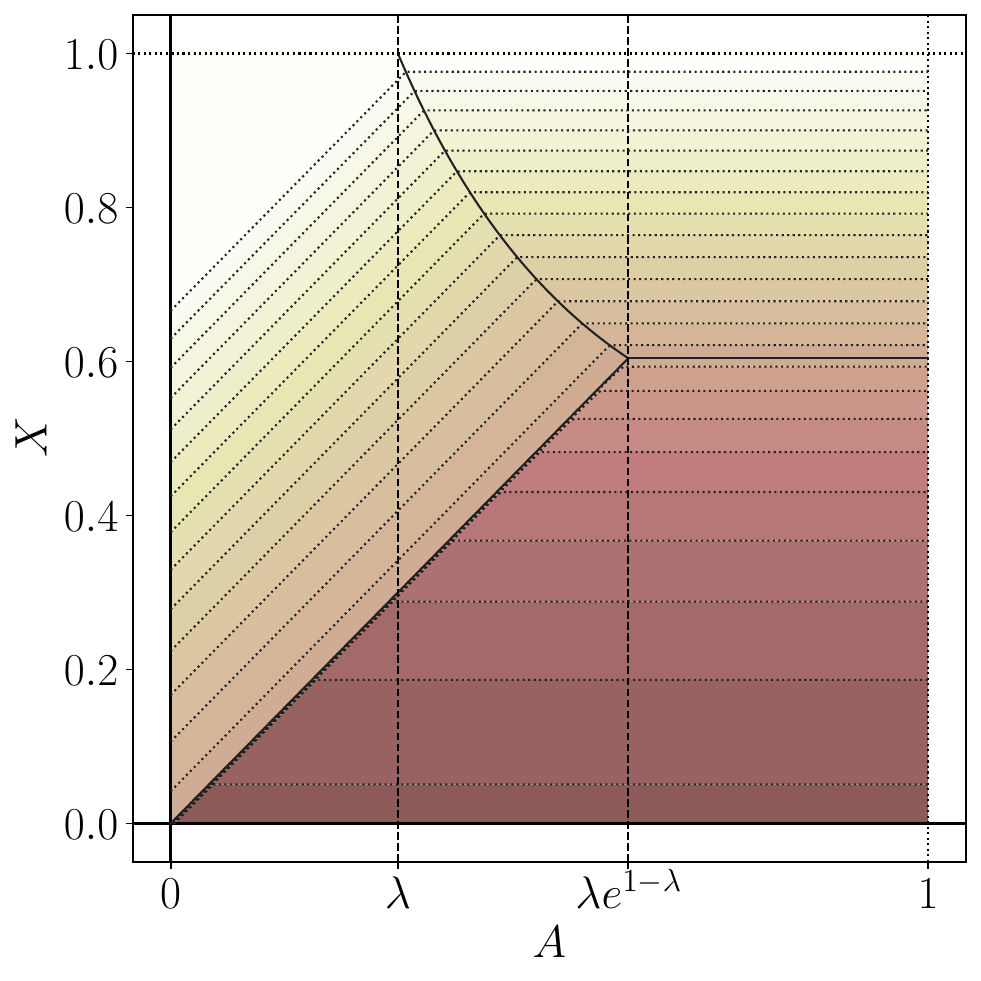}
        \caption{$f$ with $\lambda = 0.3$}
    \end{subfigure}
    \begin{subfigure}{0.32\textwidth}
        \centering
        \includegraphics[width=\textwidth]{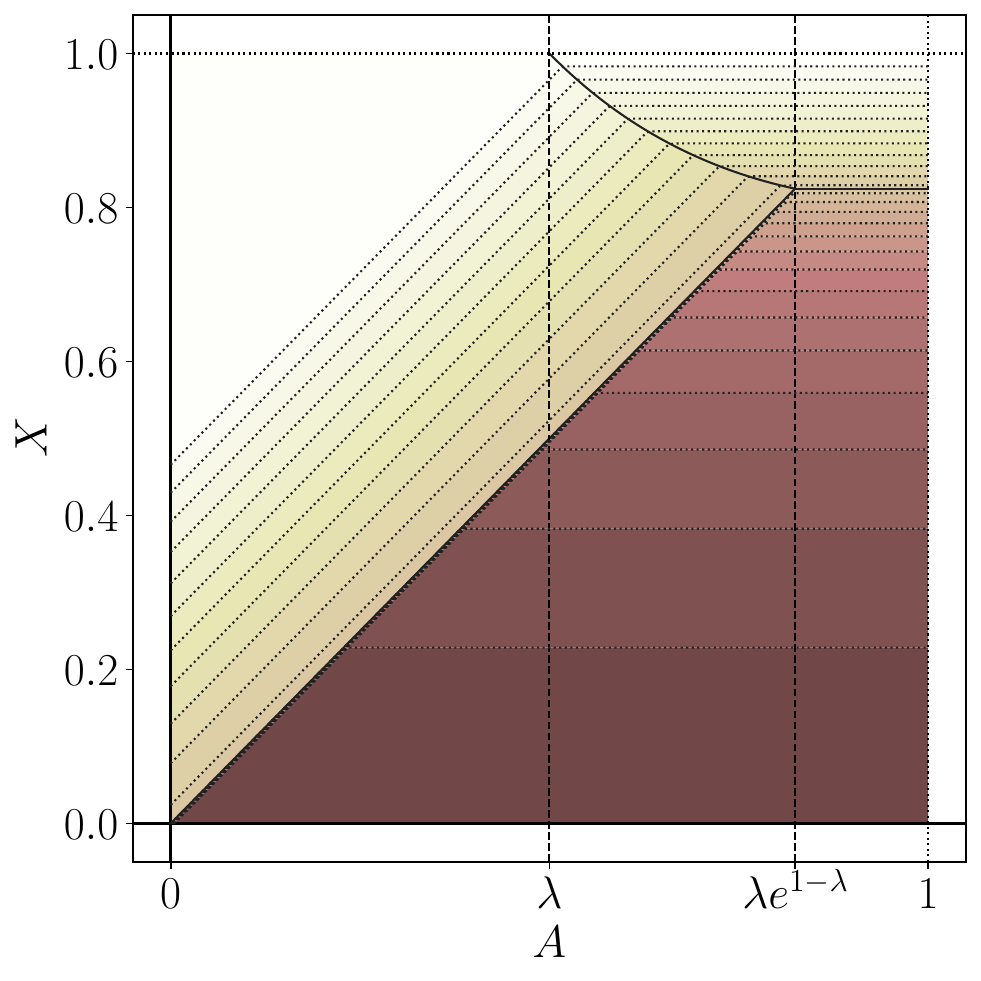}
        \caption{$f$ with $\lambda = 0.5$}
    \end{subfigure}
    \caption{$f_0$, $f_1$, and $f$ with $\lambda \in \{0.1, 0.3, 0.5\}$. (a)-(c) depict the function values of $f_0$ and $f_1$ with respect to $z \in [0, 1]$. (d)-(f) depict the contour plots with respect to $A \in [0, 1]$ and $X \in [0, 1]$: the brighter the color is, the closer to 1 the function value is.}
    \label{fig:vw:fig}
\end{figure}

The algorithm also takes as input a parameter $\lambda \in [0,1]$, which determines how much it trusts the advice; this is directly reflected in the choice of $f$.
When $\lambda = 0$, the penalty reduces to $f(A,X) = e^{X-1}$, and the algorithm recovers the classical \Balance{} algorithm with a $(1 - 1/e)$-competitive guarantee.
At the other extreme, when $\lambda = 1$, the penalty becomes
\[
    f(A,X) = \begin{cases}
        0, &\text{if $A > X$,} \\
        1, &\text{if $A \leq X$,}
    \end{cases}
\]
causing the algorithm to follow the advice exactly.
As $\lambda$ varies from $0$ to $1$, the algorithm achieves the robustness-consistency tradeoff described in \cref{thm:vertex_weighted}.
We illustrate the behavior of $f_0$, $f_1$, and $f$ across different values of $\lambda$ in \cref{fig:vw:fig}.

We now restate and prove \cref{thm:vertex_weighted} in the remaining of this section:
\vwalgmain*

Our analysis proceeds as follows.
In \cref{sec:vwalg:pd}, we perform a primal-dual analysis to bound robustness and consistency of \LAB{} using expressions involving $f$.
In \cref{sec:vwalg:int}, we analyze these expressions under integral advice to establish the guarantees in \cref{thm:vertex_weighted}.
In \cref{sec:vwalg:frac}, we extend the results to fractional advice.
Finally, in \cref{sec:vwalg:adw}, we show how \LAB{} extends to the AdWords problem under the small bids assumption.

\subsection{Primal-dual analysis}
\label{sec:vwalg:pd}
We prove \cref{thm:vertex_weighted} using an online primal-dual analysis.
Recall the primal and dual LPs for vertex-weighted bipartite matching in \cref{sec:prelim}.
We will maintain dual variables $(\alpha, \beta) \in \R^U \times \R^V$ online, such that the objective value of the dual is equal to the total weight obtained by the algorithm.

\paragraph{Properties of penalty functions.}
Before presenting the construction of our dual variables, we first argue that the penalty functions used in the algorithm indeed satisfy the properties that we previously discussed.
\begin{lemma}
    The functions $f_0$ and $f_1$ from \cref{eq:vw:penaltyint} and $f$ from \cref{eq:vw:penaltyfrac} satisfy the following properties:
    \begin{enumerate}
        \item \label{prop:vw:func1} $f_0$ and $f_1$ are increasing with $f_0(1) = f_1(1) = 1$;
        \item \label{prop:vw:func2} for any $z \in [0, 1]$, $f_0(z) \geq f_1(z)$;
        \item \label{prop:vw:func3} for any $A \in [0, 1]$, $f(A, X)$ is increasing on $X \in [0, 1]$ with $f(A, 1) = 1$;
        \item \label{prop:vw:func4} for any $X \in [0, 1]$, $f(A, X)$ is decreasing on $A \in [0, 1]$;
        \item \label{prop:vw:func5} for any $X \in [0, 1]$, $f(0, X) = f_0(X)$ and $f(1, X) = f_1(X)$.
    \end{enumerate}
\end{lemma}
\begin{proof}
    Let us first prove Property \ref{prop:vw:func1}.
    It is trivial to see that $f_0$ is an increasing function with $f_0(1) = 1$.
    We can also see that $f_1$ is an increasing function since $z \mapsto -\lambda e^{1 - \lambda - z}$ is increasing on $z \in [\lambda e^{1-\lambda}, 1]$, and $t \mapsto W(t)$ is also increasing on $t \geq -\lambda e^{1-\lambda-\lambda e^{1-\lambda}} \geq -1/e$.
    We have $f_1(1) = 1$ by definition.

    For Property \ref{prop:vw:func2}, we prove a stronger statement that $f_1(z) \leq e^{z-1}$ for any $z \in [0, 1]$;
    note that $f_0(z) = e^{z + \lambda - 1} \geq e^{z - 1}$ since $\lambda \in [0, 1]$.
    Indeed, for $z \in [0, \lambda e^{1 - \lambda})$, observe that $f_1(z) \leq e^{z - 1}$ is implied by
    \[
        \frac{1-z}{e^{1 - z}}
        > \frac{1 - \lambda e^{1-\lambda}}{e^{1 - \lambda e^{1 - \lambda}}}
        \geq \frac{1 - \lambda e^{1 - \lambda}}{e^{1 - \lambda}}
        = e^{\lambda - 1} - \lambda,
    \]
    where the first inequality comes from that $z \mapsto \frac{1-z}{e^{1-z}}$ is decreasing on $z \geq 0$, and the second from that $e^{1 - \lambda e^{1 - \lambda}} \leq e^{1 - \lambda}$ for $\lambda \in [0, 1]$.
    Meanwhile, for $\lambda \in [0, 1]$ and $z \in [\lambda e^{1 - \lambda}, 1)$, we have
    \[
        -\lambda e^{1-z} \cdot e^{-\lambda e^{1-z}}
        > -\lambda e^{1-z} \cdot e^{-\lambda}
        \geq - \lambda e^{1-\lambda} \cdot e^{- \lambda e^{1 - \lambda}}
        \geq -\frac{1}{e},
    \]
    where the first inequality follows from that $ -\lambda e^{1-z} < -\lambda$, and the second from that $z \geq \lambda e^{1-\lambda}$.
    We thus have
    \[
         -\lambda e^{1-z}
         > W(-\lambda e^{1 - \lambda - z})
         \geq -\lambda e^{1-\lambda}
         \geq -1,
    \]
    where the first inequality implies $f_1(z) \leq e^{z-1}$ while the last two inequalities show that $f_1$ is well-defined.

    The rest of the properties can be easily shown by the above two properties.
    For Property \ref{prop:vw:func3}, it is easy to see that $f(A, X)$ is increasing on $X$ by the definition of $f$ as well as Property \ref{prop:vw:func1}.
    Note also that $f(A, 1) = 1$ again due to Property \ref{prop:vw:func1}.
    For Property \ref{prop:vw:func4}, $f_0(X-A)$ is decreasing on $A \in [0, X]$ by Property \ref{prop:vw:func1} .
    The proof then follows from the definition of $f$.
    Lastly, for Property \ref{prop:vw:func5}, notice that $f(0, X) = \max\{f_0(X), f_1(X)\} = f_0(X)$ for any $X \in [0, 1]$ due to Property \ref{prop:vw:func2}.
    It is trivial by definition to see $f(1, X) = f_1(X)$ for any $X \in [0, 1]$.
\end{proof}

\paragraph{Construction of dual variables.}
We now describe the construction of the dual variables.
First, initialize all dual variables to 0.
Consider now an iteration when an online vertex $v \in V$ arrives.
For each offline vertex $u \in U$, let $\Aup{v}_u$ denote the total amount allocated to $u$ by the advice up to and including $v$.
For each neighbor $u \in N(v)$ of $v$, let $x_{u, v}$ denote the amount that $u$ is filled by the algorithm in this iteration, and let $\Xup{v}_u$ be the total amount allocated to $u$ by the algorithm by the end of this iteration.
We set
\begin{itemize}
    \item $\alpha_u \gets \alpha_u + x_{u, v} \cdot w_u f(\Aup{v}_u,\Xup{v}_u)$ for every $u \in N(v)$, and
    \item $\beta_v \gets \max_{u \in N(v)} \left\{ w_u (1-f(\Aup{v}_u, \Xup{v}_u)) \right\}$.  
\end{itemize}
Note that the $\alpha_u$ variables are potentially increased in each iteration of the algorithm. On the other hand, each $\beta_v$ variable is only updated once throughout the execution of the algorithm -- in the iteration when $v$ arrives. 

\begin{observation} \label{obs:vw:betav}
    For each $v \in V$, if $\sum_{u \in N(v)} x_{u, v} < 1$, we then have $\beta_v = 0$.
    Moreover, for every $u \in N(v)$ with $x_{u, v} > 0$, $w_u (1 - f(\Aup{v}_u, \Xup{v}_u))$ is constant. 
\end{observation}
\begin{proof}
    For the first statement, notice that the algorithm pushes water from $v$ until $v$ is fully matched or its neighbors $N(v)$ are saturated.
    Therefore, $\sum_{u \in N(v)} x_{u, v} < 1$ implies that $\Xup{v}_u = 1$ for all $u \in N(v)$, and hence, $\beta_v = 0$ since $f(\cdot, 1) = 1$.
    
    The second statement is because the algorithm always allocates an infinitesimal unit to a neighbor with highest potential, so all neighbors $u \in N(v)$ with $x_{u,v} > 0$ must have the same potential at the end of the iteration.
\end{proof}

\begin{lemma}
    \label{lem:pd_equal_vertex_weighted}
    The value of the algorithm is equal to the objective value of $(\alpha, \beta)$ in the dual LP. 
\end{lemma}
\begin{proof}
    Let $\ALG$ and $\DUAL$ denote the total weight obtained by the algorithm and the objective value of $(\alpha, \beta)$ at any iteration.
    At the very beginning, we have $\ALG = \DUAL = 0$.
    We will show that $\Delta \ALG = \Delta \DUAL$ in each iteration of the algorithm.
    Consider some iteration when an online vertex $v$ arrives.
    In that iteration, we have
    $
        \Delta \ALG = \sum_{u \in N(v)} w_u x_{u, v}.
    $
    Let us now calculate $\Delta \DUAL$.
    For clarity, let $A_u := \Aup{v}_u$ and $X_u := \Xup{v}_u$.
    We then have
    \begin{align*}
        \Delta \DUAL
        & = \sum_{u \in U} \Delta \alpha_u + \beta_v \\
        & \stackrel{(a)}{=} \sum_{u \in N(v)} x_{u, v} \cdot w_u f(A_u, X_u) + \sum_{u \in N(v)} x_{u, v} \cdot \beta_v \\
        & = \sum_{u \in N(v)} x_{u, v} \cdot \left(w_u f(A_u, X_u) + \beta_v\right)\\
        & \stackrel{(b)}{=} \sum_{u \in N(v)} x_{u, v} \cdot \left(w_u f(A_u, X_u) + w_u(1 - f(A_u, X_u))\right) \\
        & = \sum_{u \in N(v)} w_u x_{u, v} \\
        & = \Delta \ALG,
    \end{align*}
    where both (a) and (b) follow from \cref{obs:vw:betav}.
\end{proof}

We now analyze the robustness of the algorithm.
\begin{lemma}
    \label{lem:vw_r}
    The algorithm is $r$-robust for any $r$ satisfying that, for any $(u, v) \in E$, 
    \begin{equation}
    \label{eq:vw_r}
        r \leq  \int_0^{\Xup{v}_u} f(\Aup{v}_u,z) \,dz + (1 - f(\Aup{v}_u, \Xup{v}_u)).
    \end{equation}
\end{lemma}

\begin{proof}
Due to \cref{lem:pre:pda,lem:pd_equal_vertex_weighted}, it suffices to prove that that dual is approximately feasible, i.e.,
\[
\alpha_u + \beta_v \geq r \cdot w_u  \text{for all $(u, v) \in E$}.
\]
Note that, for any $(u, v) \in E$, we have
\begin{align*}
    \alpha_u + \beta_v &
    \stackrel{(a)}{\geq} \sum_{t \preceq v}x_{u, t} \cdot  w_u f(\Aup{t}_u, \Xup{t}_u) + \max_{u' \in N(v)} \left\{  w_{u'}(1 - f(\Aup{v}_{u'}, \Xup{v}_{u'})) \right\} \\
    &\geq \sum_{t \preceq v} x_{u, t} \cdot w_u f(\Aup{t}_u, \Xup{t}_u) +  w_{u} (1 - f(\Aup{v}_u, \Xup{v}_u)) \\
    &\stackrel{(b)}{\geq} \sum_{t \preceq v} x_{u, t} \cdot w_u f(\Aup{v}_u, \Xup{t}_u) +  w_{u}(1 - f(\Aup{v}_u, \Xup{v}_u))  \\
    &\stackrel{(c)}{\geq} w_u \cdot \left[ \int_0^{\Xup{v}_u} f(\Aup{v}_u, z) \,dz + (1 - f(\Aup{v}_u, \Xup{v}_u)) \right],
\end{align*}
where (a) is because $\alpha_u$ does not decrease throughout the execution, (b) is because $f(A,X)$ is decreasing in $A$, and (c) is because  $f(A,X)$ is increasing in $X$. 
Therefore, the algorithm is $r$-robust whenever $r$ satisfies that, for any $(u, v) \in E$, 
\[
    r \leq \int_0^{\Xup{v}_u} f(\Aup{v}_u, z) \,dz + (1 - f(\Aup{v}_u, \Xup{v}_u)).
\]
\end{proof}

Next, we analyze the consistency of the algorithm. 

\begin{lemma}
    \label{lem:vw_c}
    The algorithm is $c$-consistent, for any value of $c$ satisfying that, for every $u \in U$,
    \begin{equation}
    \label{eq:vw_c}
        \sum_{t \in N(u)} \left[ x_{u, t} \cdot f(\Aup{t}_u, \Xup{t}_u) + a_{u, t} \cdot (1-f(\Aup{t}_u, \Xup{t}_u)) \right]
        \geq c\cdot A_u,
    \end{equation}
    where $A_u := \sum_{t \in N(u)} a_{u, t}$ denotes the total amount that $u$ is eventually filled by the advice.
    \end{lemma}

    \begin{proof}
    Our goal here is to prove that
    $
        \ALG \geq c \cdot \ADVICE, 
    $
    where $\ALG$ is the value earned by the algorithm, and $\ADVICE$ is the value earned by the advice.
    On the one hand, we have  
    \[
        \ADVICE = \sum_{u \in U} w_u A_u.
    \]
    On the other hand, due to \cref{lem:pd_equal_vertex_weighted}, we have
    \begin{align*}
        \ALG 
        & = \sum_{u \in U} \alpha_u + \sum_{t \in V} \beta_t \\
        & \geq \sum_{u \in U}\alpha_u + \sum_{t \in V} \left( \beta_t \cdot \sum_{u \in N(t)} a_{u, t} \right) \\
        & = \sum_{u \in U}\alpha_u + \sum_{u \in U} \sum_{t \in N(u)} a_{u, t}  \beta_t\\
        & = \sum_{u \in U} \left(\alpha_u + \sum_{t \in N(u)} a_{u, t} \beta_t \right),
    \end{align*}
    where the inequality is due to the feasibility of the advice $\{a_e\}_{e \in E}$.
    Therefore, to show $\ALG \geq c \cdot \ADVICE$, it suffices to show
    \begin{equation}\label{eq:consistency}
        \alpha_u + \sum_{t \in N(u)} a_{u, t} \beta_t \geq c \cdot w_u A_u  \text{for all $u \in U$.}
    \end{equation}
    By construction, observe that the left-hand side of \cref{eq:consistency} is bounded by
    \begin{align*}
        \alpha_u + \sum_{t \in N(u)} a_{u, t} \beta_t 
        &\geq \sum_{t \in N(u)} x_{u, t} \cdot w_u f(\Aup{t}_u,\Xup{t}_u) + \sum_{t \in N(u)}  a_{u, t} \cdot w_u (1-f(\Aup{t}_u, \Xup{t}_u)) \\
        &= w_u \cdot \sum_{t \in N(u)} \left[ x_{u, t} \cdot f(\Aup{t}_u, \Xup{t}_u) + a_{u, t} \cdot (1-f(\Aup{t}_u, \Xup{t}_u)) \right].
    \end{align*}
    Therefore, \cref{eq:consistency} holds for any value of $c$ satisfying that, for all $u \in U$,
    \[
        \sum_{t \in N(u)} \left[ x_{u, t} \cdot f(\Aup{t}_u, \Xup{t}_u) + a_{u, t} \cdot (1-f(\Aup{t}_u, \Xup{t}_u)) \right]
        \geq c\cdot A_u.
    \]
\end{proof}

\subsection{Integral advice} \label{sec:vwalg:int}
We first prove \cref{thm:vertex_weighted} when the advice is integral. This serves as a warm-up, and several ideas from this setting will carry over to the case of fractional advice. When the advice is integral, the expressions for robustness and consistency in \cref{lem:vw_r,lem:vw_c} simplify considerably. 

\begin{lemma}
\label{lem:vw_rc_int}
    When the advice is integral, the algorithm is $r$-robust and $c$-consistent where
    \begin{align*}
    r & = \min_{X \in [0,1]} \min\left\{ \int_0^X f_0(z) \,dz + (1 - f_0(X)), \, \int_0^X f_1(z) \,dz + (1 - f_1(X)) \right\}  \text{and} \\
    c & = \min_{X \in [0,1]} \min_{Y \in [0, X]} \left\{  \int_0^{Y} f_0(z) \,dz + (X - Y) \cdot f_1(X) + (1 - f_1(X))\right\}.
    \end{align*}
\end{lemma}
\begin{proof}
    Recall that $f(0, X) = f_0(X)$ and $f(1, X) = f_1(X)$ for any $X \in [0, 1]$.
    Note also that, for any $(u, v) \in E$, $\Aup{v}_u$ is either 0 or 1.
    The expression for robustness $r$ then follows directly from \cref{lem:vw_r}.

    Let us now turn to consistency.
    Consider any offline vertex $u \in U$.
    If $u$ is exposed in the advice (i.e., $A_u = 0$), \cref{eq:vw_c} in \cref{lem:vw_c} trivially holds.
    On the other hand, if $u$ is matched by the advice with an online vertex $v \in N(u)$ (i.e., $a_{u, v} = 1$), we can observe that $\Aup{t}_u = 0$ for any $t \prec v$, and $\Aup{v}_u = 1$.
    Therefore, \cref{eq:vw_c} reduces to
    \[
        \sum_{t \in N(u) : t \prec v} x_{u, t} \cdot f_0(\Xup{t}_u) + x_{u, v} \cdot f_1(\Xup{v}_u) + (1 - f_1(\Xup{v}_u) \geq c,
    \]
    which is implied by 
    \[
        \int_0^{\Xup{v}_u - x_{u, v}} f_0(z) \,dz + x_{u, v} \cdot f_1(\Xup{v}_u) + (1 - f_1(\Xup{v}_u)) \geq c
    \]
    due to the fact that $f_0$ is increasing.
    Hence, the algorithm is $c$-consistent for 
    \[
        c = \min_{X \in [0,1]} \min_{Y \in [0, X]} \left\{  \int_0^{Y} f_0(z) \,dz + (X - Y) \cdot f_1(X) + (1 - f_1(X))\right\}.
    \]
\end{proof}

\newcommand{\tl}{t(\lambda)}
\newcommand{\fzero}{f(0,\cdot)}
\newcommand{\fone}{f(1,\cdot)}
\newcommand{\C}{\mathcal{C}}

We now prove \cref{thm:vertex_weighted} when the advice is integral.
Recall that we defined $f_0$ and $f_1$ in \cref{eq:vw:penaltyint} as follows:
\[
    f_0(z) := \min \{ e^{z+\lambda-1}, 1 \}, \text{ and }
    f_1(z) := \begin{cases}
    \frac{e^{\lambda-1} - \lambda}{1-z}, & \text{ if } z \in [0, \lambda e^{1-\lambda}), \\
    \frac{-\lambda}{W(-\lambda e^{1-\lambda - z})}, & \text{ if } z \in [\lambda e^{1-\lambda}, 1), \\
    1, & \text{ if } z = 1.
    \end{cases}
\]

The below lemma determines the robustness of \LAB{}.
\begin{lemma}
    \label{lem:vw_r_int}
    We have
    \begin{align*}
        r(\lambda)
        & := 1 - e^{\lambda-1} -\left( e^{\lambda-1} -\lambda \right) \ln (1-\lambda e^{1 - \lambda} ) - \lambda(1-\lambda) \\
        & = \min_{X \in [0,1]} \min\left\{ \int_0^X f_0(z) \,dz + (1 - f_0(X)), \, \int_0^X f_1(z) \,dz + (1 - f_1(X)) \right\}.
    \end{align*}
\end{lemma}
\begin{proof}
    We first consider the term defined by $f_0$.
    Observe that, for any $X \in [0, 1 - \lambda]$,
    \[
        \int_0^X f_0(z) \,dz + (1 - f_0(X)) = 1- e^{\lambda - 1}.
    \]
    Moreover, for any $X \in (1 - \lambda , 1]$, we can also see that
    \[
        \int_0^X f_0(z) \,dz + (1 - f_0(X)) \geq \int_0^{1 - \lambda} f_0(z) \,dz = 1 - e^{\lambda - 1}.
    \]

    We now turn to the other term defined by $f_1$. Let \[
        I_1(X) := \int_0^X f_1(z) \,dz + (1 - f_1(X)).
    \]
    For $X \in (0, 1)$, one can calculate the derivative of $I_1(X)$ as follows (see \cref{eq:pre:Wprop1}):
    \begin{align*}
        \frac{d}{dX}\left[ I_1(X) \right]
        & = f_1(X)-\frac{d}{dX}f_1(X) \\
        & = \begin{cases}
            \displaystyle -\frac{(e^{\lambda-1}-\lambda) X} {(1-X)^2}, & \text{for } X \in (0, \lambda e^{1 - \lambda}), \\[10pt]
            \displaystyle -\frac{\lambda}{1+          W \left( -\lambda e^{\,1-\lambda-X} \right)}, & \text{for } X \in [\lambda e^{1 - \lambda}, 1).
    \end{cases}
    \end{align*}
    Observe that $e^{\lambda - 1} - \lambda \geq 0$ for any $\lambda \in [0, 1]$.
    Note also that, for any $\lambda < 1$ and $X \geq \lambda e^{1 - \lambda}$, we have $1 + W(- \lambda e^{1 - \lambda - X}) > 0$ since $W$ is an increasing function with $W(-\lambda e^{1 - \lambda - \lambda e^{1- \lambda}}) = -\lambda e^{1 - \lambda} > -1$.
    This implies that $I_1(X)$ is decreasing in on $(0, 1)$, and hence, its minimum is attained at $X = 1$.

    We now compute \(I_{1}(1)\). Let $w(z) := W (-\lambda e^{\,1-\lambda-z} )$. We then have
    \[
        I_{1}(1)
        = \int_0^1 f_1(z) \,dz
        = \underbrace{\int_{0}^{\lambda e^{1 - \lambda }}\frac{e^{\lambda-1}-\lambda}{1-z} \,dz}_{(A)} + \underbrace{\int_{\lambda e^{1 - \lambda }}^{1} \left( -\frac{\lambda}{w(z)} \right) \,dz}_{(B)}.
    \]
    The first part is 
    \[
        (A)
        = \int_{0}^{\lambda e^{1 - \lambda}} \frac{e^{\lambda-1}-\lambda}{1-z} \,dz
        = -\left( e^{\lambda-1}-\lambda \right) \ln \left( 1-\lambda e^{1 - \lambda} \right).
    \]
    For the second part, since \( w e^{w} = -\lambda e^{1-\lambda-z} \), we have \(dz=-\tfrac{1+w}{w}\,dw\).
    Observe also that $w(1) = W(-\lambda e^{-\lambda}) = -\lambda$ and $w(\lambda e^{1 - \lambda}) = W(-\lambda e^{1 - \lambda - \lambda e^{1 - \lambda}}) = -\lambda\,e^{1-\lambda}.$
    We thus have
    \begin{align*}
        (B)
        & = \int_{\lambda e^{1 - \lambda}}^{1} \left( -\frac{\lambda}{w(z)} \right) \,dz \\
        & = \lambda \int_{-\lambda e^{1 - \lambda}}^{-\lambda} \left( \frac{1}{w^{2}} + \frac{1}{w} \right) dw \\
        & = \lambda \left[ -\frac{1}{w} + \ln (-w) \right]_{-\lambda e^{1 - \lambda}}^{- \lambda} \\
        & = 1 - e^{\lambda - 1} - \lambda(1 - \lambda).
    \end{align*}
    We can therefore conclude that 
    \[
        I_1(1)
        = (A) + (B)
        = 1 - e^{\lambda-1} -\left( e^{\lambda-1} -\lambda \right) \ln (1-\lambda e^{1 - \lambda} ) - \lambda(1-\lambda).
    \]
    
    It remains to show that $I_1(1) \leq 1 - e^{\lambda-1}$ for all $\lambda \in [0,1]$.
    This is equivalent to showing that $G(\lambda) \geq 0$ for all $\lambda \in [0,1]$, where 
    \[
        G(\lambda) := \lambda (1-\lambda) + (e^{\lambda-1} - \lambda)\ln(1-\lambda e^{1-\lambda}).
    \]
    To prove this, we need the following inequality:

    \begin{proposition} \label{lem:vw:int:logineq}
        For any $t \in (0, 1)$, we have
        $
            \ln (1 - t) \geq \frac{-t}{\sqrt{1-t}}.
        $
    \end{proposition}
    \begin{proof}
        Let us substitute $s := \sqrt{1-t}$ for simplicity.
        It then suffices to show that, for any $s \in (0, 1)$,
        \[
            \ln s^2  \geq \frac{s^2 - 1}{s} \iff s - 2 \ln s - \frac{1}{s} \leq 0.
        \]
        Let $h(s) := s - 2 \ln s - \frac{1}{s}$.
        Observe that $h(1) = 0$.
        Note also that $h'(s) = 1 - \frac{2}{s} + \frac{1}{s^2} = \frac{(s - 1)^2}{s^2} \geq 0$  for any $s \in (0, 1)$.
        This implies that $h$ is increasing on $(0, 1)$ with $h(1) = 0$, completing the proof.
    \end{proof}
    
    We are now ready to prove $G(\lambda) \geq 0$ for all $\lambda \in [0, 1]$. 
    Note that $G(0) = G(1) = 0$.
    For $\lambda \in (0, 1)$, it suffices to show that
    \begin{equation} \label{eq:vw:int:rob1}
        1 - \lambda \geq \left( 1 - \frac{1}{\lambda e^{1 - \lambda}} \right) \ln (1 - \lambda e^{1 - \lambda}).
    \end{equation}
    Since $1 - \frac{1}{\lambda e^{1 - \lambda}} < 0$ for $\lambda \in (0, 1)$, due to \cref{lem:vw:int:logineq} with $t := \lambda e^{1 - \lambda} \in (0, 1)$, we have
    \[
        \left( 1 - \frac{1}{\lambda e^{1 - \lambda}} \right) \ln (1 - \lambda e^{1 - \lambda}) \leq \left( 1 - \frac{1}{\lambda e^{1 - \lambda}} \right) \cdot \frac{-\lambda e^{1- \lambda}}{\sqrt{1 - \lambda e^{1 - \lambda}}} = \sqrt{1 - \lambda e^{1 - \lambda}}.
    \]

    We claim that
    \[
        1 - \lambda \geq \sqrt{1 - \lambda e^{1 - \lambda}}.
    \]
    Note that this claim immediately implies \cref{eq:vw:int:rob1}.
    For simplicity, let $t := 1 - \lambda$.
    It is then equivalent to showing that, for any $t \in (0, 1)$,
    \[
        g(t) := t^2 + (1 - t) e^t \geq 1.
    \]
    Since we have $g'(t) = t(2 - e^t)$, we can infer that $\inf_{t \in (0, 1)} g(t) = \min\{g(0), g(1)\} = 1$.
\end{proof}

The next lemma determines the consistency of \LAB{}.
\begin{lemma}
    \label{lem:vw_c_int}
    We have
    \begin{align*}
        c(\lambda)
        & := 1 + \lambda - e^{\lambda - 1} \\
        & = \min_{X \in [0,1]} \min_{Y \in [0, X]} \left\{  \int_0^{Y} f_0(z) \,dz + (X - Y) \cdot f_1(X) + (1 - f_1(X))\right\}.
    \end{align*}
\end{lemma}
\begin{proof}
For simplicity of presentation, we define
\[
    J(X, Y) := \int_0^Y f_0(z) \,dz + (X - Y) f_1(X) + (1 - f_1(X)).
\]
Notice that, by definition,
\begin{equation} \label{eq:vw:int:c}
    c(\lambda) = \min_{X \in [0, 1]} \min_{Y \in [0, X]} J(X, Y).
\end{equation}
Let us first calculate the inner minimizer $Y^\star(X)$ for each $X$ fixed.
Observe that
\begin{equation*}
\label{eq:cons_derivative}
    \frac{\partial}{\partial Y} J(X, Y) = f_0(Y) - f_1(X).
\end{equation*}
If $X < \lambda e^{1 - \lambda}$, we have
$
    f_1(X) = \frac{e^{\lambda - 1} - \lambda}{1 - X} < e^{\lambda - 1} \leq f_0(Y) \text{ for any $Y \in [0, X]$},
$
implying that $J(X, Y)$ is increasing on $Y \in [0, X]$. Therefore, the minimum is attained at $Y^\star(X) := 0$.
On the other hand, if $X \geq \lambda e^{1 - \lambda}$, since $f_1(X) \in [e^{\lambda - 1}, 1]$ and $f_0(Y) = e^{\lambda - 1 + Y} \in [e^{\lambda - 1}, 1]$ on $Y \in [0, 1 - \lambda]$, we can see that $J(X,Y)$ is minimized at $Y^\star(X)$ such that $f_0(Y^\star(X)) = f_1(X)$.
We therefore have
\[
    Y^\star(X) = \begin{cases}
        0, & \text{ if $X < \lambda e^{1 - \lambda}$, } \\
        1 - \lambda + \ln f_1(X), & \text{ if $X \geq \lambda e^{1 - \lambda}$.}
    \end{cases}
\]
Note that $Y^\star(X) \in [0, 1 - \lambda]$.

Let us now compute $c(\lambda)$ from \cref{eq:vw:int:c}. 
We claim $J(X, Y^\star(X)) = 1 + \lambda - e^{\lambda - 1}$ for any $X \in [0, 1]$; note that this claim immediately implies the lemma.
If $X < \lambda e^{1 - \lambda}$, we know $Y^\star (X) = 0$.
We therefore have
\[
    J(X,0)
    =X \, \frac{e^{\lambda - 1} - \lambda }{1-X} + \left( 1 -\frac{e^{\lambda - 1} - \lambda }{1-X} \right)
    =1 + \lambda - e^{\lambda - 1}.
\]

We now consider the other case when $X \geq \lambda e^{1 - \lambda}$.
Since $Y^\star(X) = 1 - \lambda + \ln f_1(X) \leq 1 - \lambda$, we have
\begin{equation*} \label{eq:vw:int:intf0}
    \int_0^{Y^\star(X)} f_0(z) \,dz = e^{\lambda - 1} (e^{Y^\star(X)} - 1) = f_1(X) - e^{\lambda - 1}.
\end{equation*}
Moreover, we can also derive that
\begin{align*}
    X - Y^\star(X)
    & = X - 1 + \lambda - \ln f_1(X) \\
    & = \ln \left( \frac{W(-\lambda e^{1 - \lambda - X})}{-\lambda e^{1 - \lambda - X}} \right) \\
    & = - W(- \lambda e^{1 - \lambda - X}),
\end{align*}
where the second equality is due to \cref{eq:pre:Wprop2}, implying that
\[
    (X - Y^\star(X)) f_1(X)
    = - W(- \lambda e^{1 - \lambda - X}) \cdot \frac{- \lambda}{W(- \lambda e^{1 - \lambda - X})}
    = \lambda.
\]
From these equations, we finally have
\begin{align*}
    J(X, Y^\star(X))
    & = \int_0^{Y^\star(X)} f_0(z) \,dz + (X - Y^\star(X)) f_1(X) + (1 - f_1(X)) \\
    & = f_1(X) - e^{\lambda - 1} + \lambda + 1 - f_1(X) \\
    & = 1 + \lambda - e^{\lambda - 1}.
\end{align*}
\end{proof}

\subsection{Fractional advice} \label{sec:vwalg:frac}
In this subsection, we show that even when the advice can be fractional, \LAB{} achieves the same robustness and consistency ratios as in the integral case. Recall the definition of $f(A, X)$ :
\begin{equation*} 
    f(A, X) := \begin{cases}
        f_1(X), & \text{ if } A > X, \\
        \max\{ f_0(X-A), f_1(X) \}, & \text{ if } A \leq X.
    \end{cases}
\end{equation*}
where $f_1(\cdot)$ and $f_0(\cdot)$ are the penalty functions used in the integral case.

We first consider robustness. Together with \cref{lem:vw_r}, the below lemma shows that the algorithm is $r(\lambda)$-robust. Its proof easily follows from the integral case.
\begin{lemma} \label{lem:vw:frac:r}
    For any $A \in [0, 1]$ and $X \in [0, 1]$, we have
    \[
        \int_0^X f(A, z) dz + (1 - f(A, X)) \geq r(\lambda).
    \]
\end{lemma}
\begin{proof}
    From the definition of $f(A,X)$, we know that either $f(A,X) = f_1(X)$ or $f(A,X) = f_0(X-A)$. First, suppose $f(A, X) = f_1(X)$. Then
    \[
        \int_0^X f(A, z) dz + (1 - f(A, X)) \geq  \int_0^X f_1(z) dz + (1 - f_1(X)) \geq r(\lambda),
    \]
    where the first inequality comes from the fact that $f(A, z) \geq f_1(z)$ for every $z \in [0, 1]$ and the second from \cref{lem:vw_r_int} from the integral case.
    On the other hand, if $f(A,X) \neq f_1(X)$ and hence $f(A, X) = f_0(X-A)$, this implies that $X \geq A$. Hence, we  obtain
    \begin{align*}
        \int_0^X f(A, z) dz + (1 - f(A, X))
        & \geq \int_A^X f(A, z) dz + (1 - f(A, X)) \\
        & \geq \int_A^X f_0(z-A) dz + (1 - f_0(X-A)) \\
        & \geq r(\lambda),
    \end{align*}
    where the second inequality follows from that $f(A, z) \geq f_0(z-A)$ for every $z \in [A, 1]$ and the last inequality by \cref{lem:vw_r_int} from the integral case.
\end{proof}

\begin{figure}
    \centering
    \begin{subfigure}{0.32\textwidth}
        \centering
        \includegraphics[width=\textwidth]{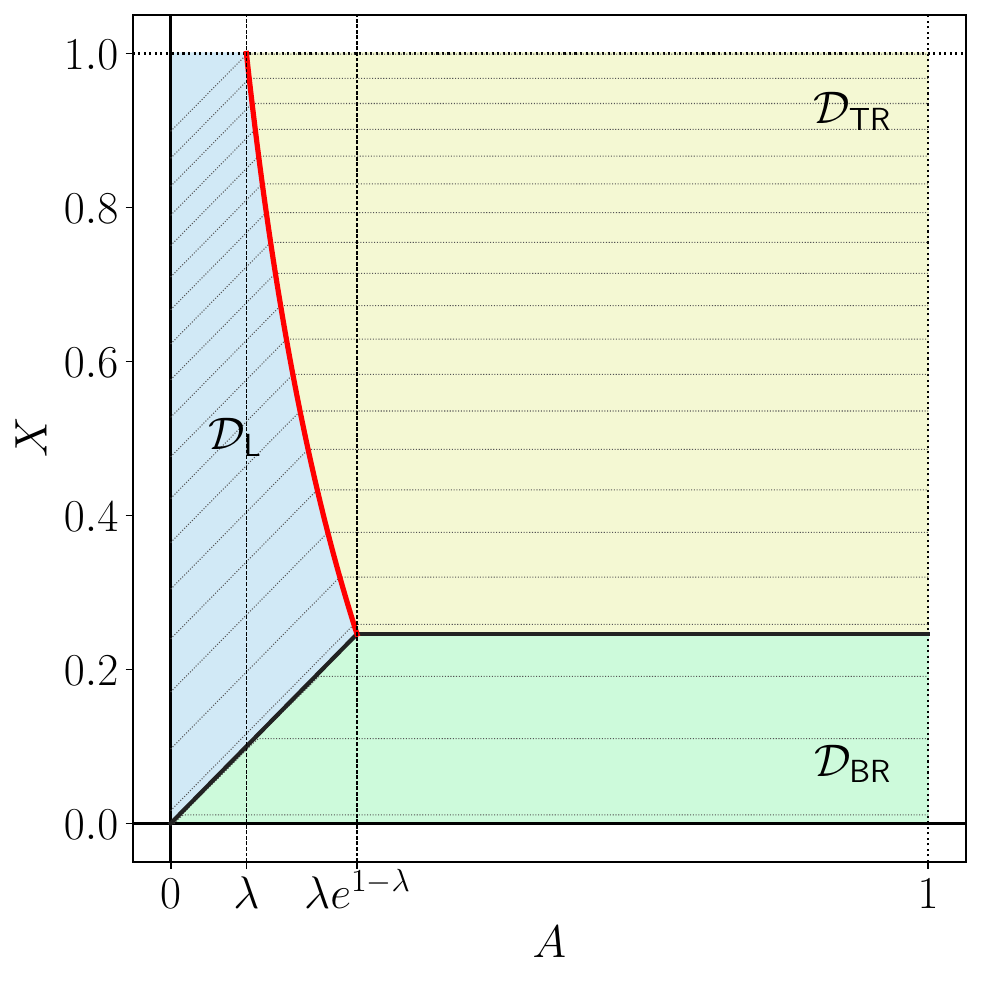}
        \caption{$\lambda = 0.1$}
    \end{subfigure}
    \begin{subfigure}{0.32\textwidth}
        \centering
        \includegraphics[width=\textwidth]{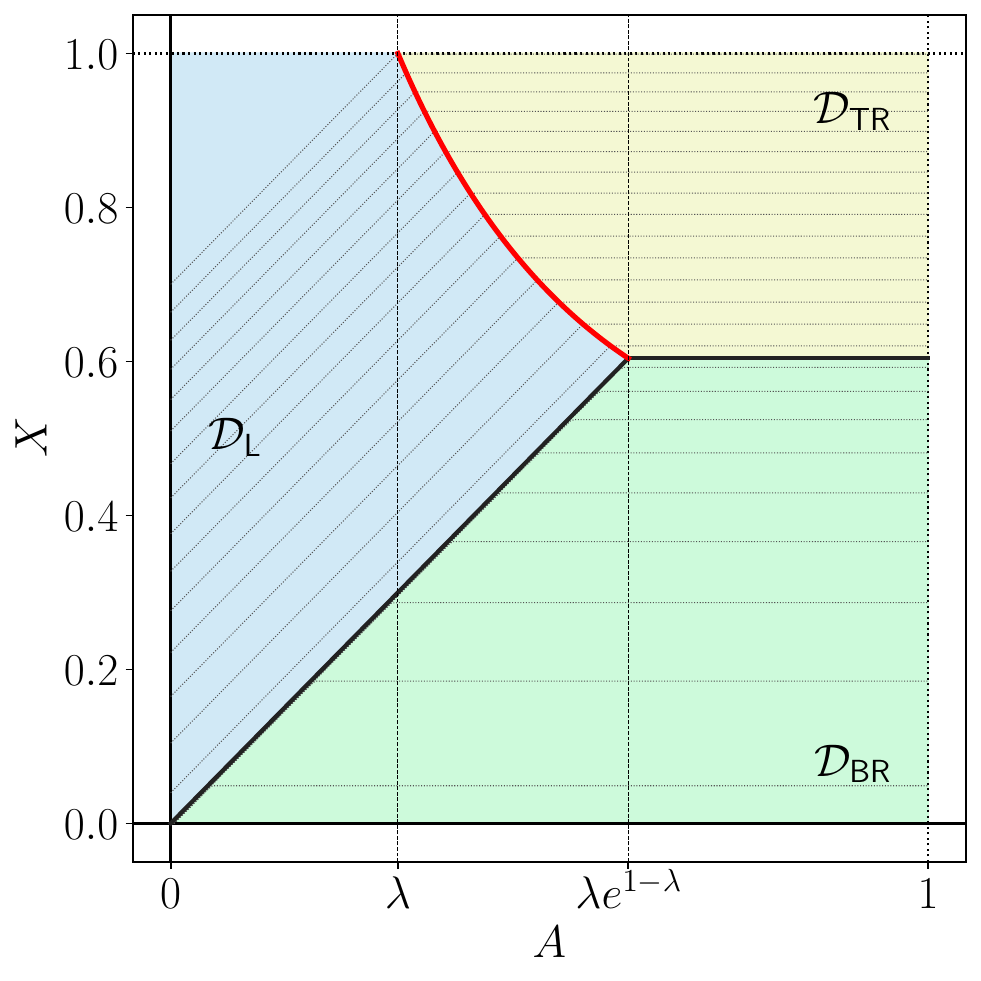}
        \caption{$\lambda = 0.3$}
    \end{subfigure}
    \begin{subfigure}{0.32\textwidth}
        \centering
        \includegraphics[width=\textwidth]{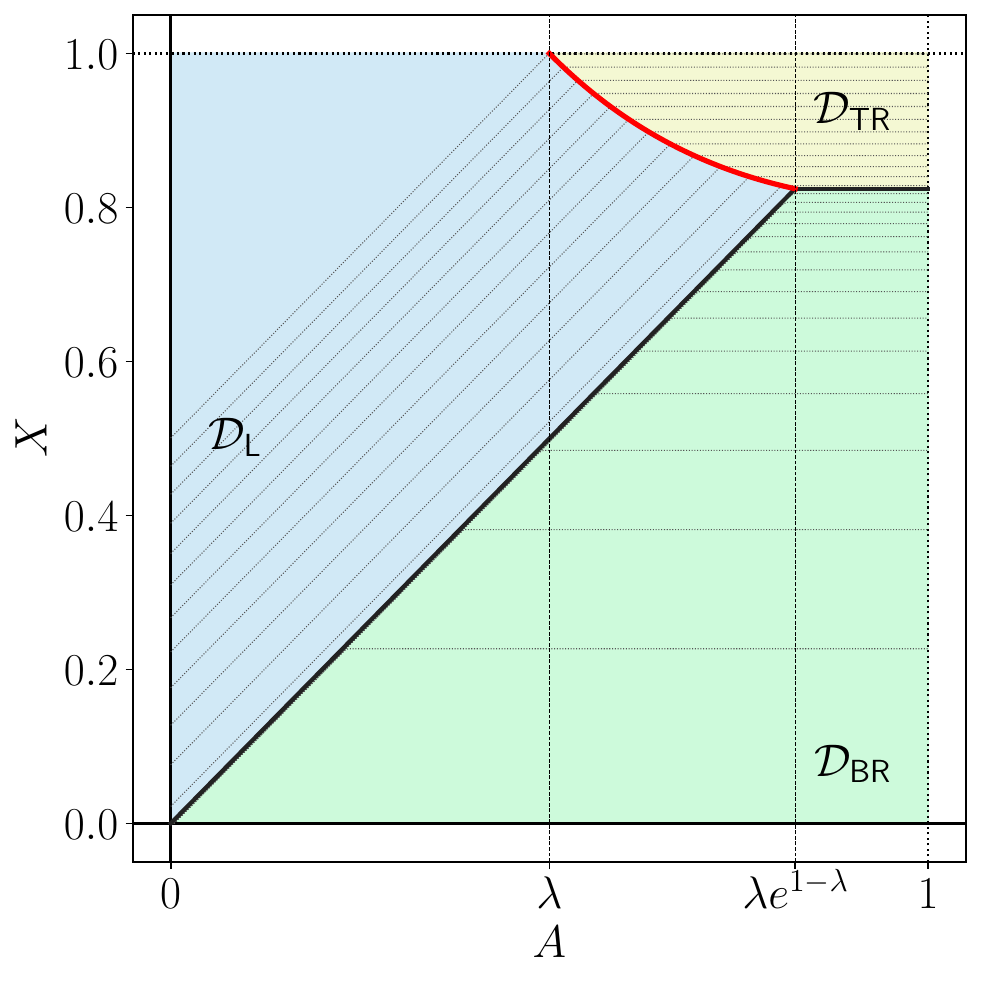}
        \caption{$\lambda = 0.5$}
    \end{subfigure}
    \caption{Contour plots of $f(A, X)$ for $\lambda \in \{0.1, 0.3, 0.5\}$. We partition $[0,1]^2$ into three regions: $\region{L}$ (Left, in blue), $\region{BR}$ (Bottom Right, in green), $\region{TR}$ (Top Right, in yellow). The red curve represents $X = A - \ln A + 1 - \lambda + \ln \lambda$ on $A \in[\lambda, \lambda e^{1-\lambda}]$ that separates $\region{L}$ and $\region{TR}$.}
    \label{fig:vw:fr:contour}
\end{figure}

Let us now focus on the consistency. Recall the contour plot of $f$ --- see \cref{fig:vw:fr:contour}.
We partition $[0, 1]^2$ into three regions as follows:
\begin{itemize}
    \item $\region{L} := \{ (A, X) \in [0, 1]^2 \mid A \leq X < A - \ln A + (1-\lambda) + \ln \lambda \}$; 
    \item $\region{BR} := \{ (A, X) \in [0, 1]^2 \mid X < A \text{ and } X < \lambda e^{1-\lambda }\}$;  and
    \item $\region{TR} := \{ (A, X) \in [0, 1]^2 \mid X \geq \lambda e^{1-\lambda} \text{ and } X \geq A - \ln A + (1-\lambda) + \ln \lambda \}$. 
\end{itemize}
Observe that $(0, 0) \in \region{L}$. 
These three regions partition $[0,1]^2$ into parts where $f$ has a simple closed-form.
\begin{lemma}
    The definition of $f$ in \cref{eq:vw:penaltyfrac} is equivalent to
    \[
        f(A, X) := \begin{cases}
            f_0(X-A), & \text{if } (A,X) \in \region{L}, \\
            f_1(X) = \frac{e^{\lambda-1} - \lambda}{1-X}
            & \text{if } (A,X) \in \region{BR}, \\
            f_1(X) = \frac{-\lambda}{W(-\lambda e^{1-\lambda - X})}& \text{if } (A,X) \in \region{TR}.
        \end{cases}
    \]
\end{lemma}
\begin{proof}
To begin, note that the boundaries of the regions intersect at the single point $(\lambda e^{1-\lambda}, \lambda e^{1-\lambda})$.

If $X < A$, we have $f(A, X) = f_1(X)$ by the definition of $f$. 
Also, if $X < \lambda e^{1-\lambda}$, then $f_1(X) = \frac{e^{\lambda-1} - \lambda}{1-X}$, and if $X \geq \lambda e^{1-\lambda}$ then $f_1(X) = \frac{-\lambda}{W(-\lambda e^{1-\lambda - X})}$. Therefore, we conclude that
\[
f(A, X) = 
\begin{cases}
    f_1(X) = \frac{e^{\lambda - 1} - \lambda}{1 - X}, & \text{if } (A, X) \in \region{BR}, \\
    f_1(X) = \frac{-\lambda}{W(-\lambda e^{1 - \lambda - X})}, & \text{if } (A, X) \in \region{TR} \cap \{(A, X) : X < A\}.
\end{cases}
\]
    
    On the other hand, suppose $X \geq A$.
    To finish the proof, it suffices to show that 
    \[
f_0(X - A) \geq f_1(X)
~ \text{if and only if}~~
X \leq A - \ln A + (1 - \lambda) + \ln \lambda.
\]
First, suppose $X < \lambda e^{1-\lambda}$. Then $(A,X) \in \region{L}$, and the analysis in \cref{lem:vw_c_int} shows that $f_0(X-A) \geq f_1(X)$ in this case. On the other hand, suppose $X \geq \lambda e^{1-\lambda}$. Again, by the analysis in \cref{lem:vw_c_int}, we have
\begin{equation}
    \label{eq:rl}
    f_0(X - A) \geq f_1(X)
    \quad \Longleftrightarrow \quad
    X - A \leq Y^\star(X) = 1 - \lambda + \ln f_1(X).
\end{equation}
    Here, $f_1(X) = -\lambda / w$ where $w := W(-\lambda e^{1-\lambda - X})$. Since $we^w = -\lambda e^{1 - \lambda - X}$, we have
\[
\ln\left(-\frac{\lambda}{w}\right) = w - 1 + \lambda + X.
\]
    Substituting this back into \cref{eq:rl}, we see that \cref{eq:rl} is equivalent to $A \geq -w$. This is equivalent to
    $$A - \ln A \geq -w - \ln(-w) = -\ln \lambda + X + \lambda - 1,$$
    which is equivalent to the desired inequality
\[
X \leq A - \ln A + (1 - \lambda) + \ln \lambda.
\]
\end{proof}

Define a \emph{trajectory} $\pi$ to be a sequence of $k+1 \geq 2$ pairs
\[
\pi := ((A_0, X_0), (A_1, X_1), \cdots, (A_k, X_k))
\]
such that $0 = A_0 \leq \cdots \leq A_k \leq 1$ and $0 = X_0 \leq \cdots \leq X_k \leq 1$.
The \emph{cost} of $\pi$ is defined as
\[
\trajcost{\pi}
: = \sum_{i=1}^{k} \left[ x_i f(A_i, X_i) + a_i (1 - f(A_i, X_i)) \right],
\]
where $a_i := A_i - A_{i - 1} \geq 0$ and $x_i := X_i - X_{i - 1} \geq 0$ for every $i \in \{1, \ldots, k\}$.

By \cref{lem:vw_c}, in order to show that the consistency  remains the same as in the integral case, it suffices to prove the following lemma.
\begin{lemma} \label{lem:vwalg:frac:main}
    For any trajectory $\pi = ((A_0, X_0), \ldots, (A_k, X_k))$, $\trajcost{\pi} \geq c(\lambda) \cdot A_k$.
\end{lemma}
To this end, we first identify a class of trajectories that minimizes the cost.
We say a trajectory $\pi = ((A_0, X_0), \ldots, (A_k, X_k))$ is \emph{irreducible} if one of the following is satisfied:
\begin{enumerate}
    \item \label{cond:vwalg:frac:irre1} $(A_k, X_k) \in \region{BR}$ and $k = 1$; or
    \item \label{cond:vwalg:frac:irre2} $(A_k, X_k) \in \region{L} \cup \region{TR}$, $A_0 = \cdots = A_{k-1} = 0$, $X_{k - 1} \leq 1-\lambda$, and $f(0, X_{k - 1}) = f(A_k, X_k)$.
\end{enumerate}
Note that an irreducible trajectory indeed satisfies \cref{lem:vwalg:frac:main}.
\begin{lemma} \label{lem:vw:fr:irr}
    For any irreducible trajectory $\pi = ((A_0, X_0), \ldots, (A_k, X_k))$, $\trajcost{\pi} \geq c(\lambda) \cdot A_k$.
\end{lemma}
\begin{proof}
    Suppose first that $\pi$ satisfies Condition \ref{cond:vwalg:frac:irre1}.
    Since $(A_1, X_1) \in \region{BR}$, we have $X_1 < \lambda e^{1-\lambda}$ and hence $f(A_1, X_1) = f_1(X_1) = \frac{e^{\lambda-1} - \lambda}{1 - X_1}  = \frac{1 - c(\lambda)}{1 - X_1} < f(0, 0)$.
    We thus have
    \begin{align*}
        \trajcost{\pi} 
        & = X_1 f_1(X_1) + A_1 (1 - f_1(X_1)) \\
        &=  X_1 f_1(X_1) + (1-f_1(X_1)) + (A_1-1) (1 - f_1(X_1)) \\
        & \geq c(\lambda) + (A_1 - 1) \cdot (1 - f_1(X_1)) \\
        & = c(\lambda) + (A_1 - 1) \cdot \left(1 - \frac{1 - c(\lambda)}{1 - X_1} \right) \\
        & \geq c(\lambda) + (A_1 - 1) \cdot c(\lambda) \\
        & = c(\lambda) \cdot A_1,
    \end{align*}
    where the first inequality follows from \cref{lem:vw_c_int} in the integral case (with $Y = 0$) and the second inequality from the fact that $A_k \leq 1$ and $X_1 \geq 0$.

    Let us now consider the  case that $\pi$ satisfies Condition \ref{cond:vwalg:frac:irre2}.  If $(A_k, X_k) \in \region{L}$, then $x_k = A_k$ since
    $
    f(0, X_{k-1}) = f(A_k, X_k)
    $
    and the contour lines in $\region{L}$ have slope 1. This implies
    \begin{align*}
    \trajcost{\pi}
             &= \sum_{i = 1}^{k-1} x_i f(0, X_i) + x_k f(A_k, X_k) + A_k (1 - f(A_k, X_k)) \\
             &\geq A_k f(A_k, X_k) + A_k (1 - f(A_k, X_k)) \\
             &= A_k \\
             & \geq c(\lambda) \cdot A_k.
    \end{align*}
    On the other hand, if $(A_k, X_k) \in \region{TR}$, then we have
    $
    f(A_k, X_k) = f(1, X_k) = f_1(X_k),
    $
    which implies
    \begin{align*}
        \trajcost{\pi}
        & = \sum_{i = 1}^{k-1} x_i f_0(X_i) + x_k f_1(X_k) + A_k (1 - f_1(X_k)) \\
        &\stackrel{(a)}{\geq} \int_0^{X_{k-1}} f_0(z) dz + x_k f_1(X_{k}) + A_k (1 - f_1(X_{k})) \\
        &= \int_0^{X_{k-1}} f_0(z) dz + x_k f_1(X_{k}) + (1 - f_1(X_{k})) + (A_k-1) (1 - f_1(X_{k})) \\
        &\stackrel{(b)}{\geq} c(\lambda) + (A_k - 1) \cdot (1 - f_1(X_{k})) \\
        &\stackrel{(c)}{\geq} c(\lambda) + (A_k - 1) \cdot (1 - e^{\lambda - 1}) \\
        &\stackrel{(d)}{\geq} c(\lambda) + (A_k - 1) \cdot c(\lambda) \\
        & = c(\lambda) \cdot A_k,
    \end{align*}
    where (a) comes from the fact that $f_0$ is increasing, (b) is from \cref{lem:vw_c_int} in the integral case (with $Y = X_{k-1}$ and $X = X_k$), (c) is because $X_k \geq \lambda e^{1-\lambda}$ and $f_1$ is an increasing function with $f_1(\lambda e^{1 - \lambda}) = e^{\lambda - 1}$, and (d) is because $c(\lambda) = 1 + \lambda - e^{\lambda - 1} \geq 1 - e^{\lambda - 1}$.
\end{proof}

Let us now show that the cost of a trajectory is minimized when it is irreducible. In particular, we will prove the following lemma.
\begin{lemma} \label{lem:vw:fr:reduce}
    For any trajectory $\pi$, there exists an irreducible trajectory $\pi'$ such that $\trajcost{\pi} \geq \trajcost{\pi'}$.
\end{lemma}
To prove this lemma, we fix an arbitrary trajectory $\pi$ and apply a sequence of local modifications without ever increasing the cost of the trajectory.
The following is a key technical lemma for the local modification.
\begin{lemma} \label{lem:remove-middle-point}
    Fix a trajectory $\pi = ((A_0, X_0), \ldots, (A_k, X_k))$ with $k \geq 2$. For some $i \in \{1, \ldots, k-1\}$, let $\pi'$ be a new trajectory obtained by removing $(A_i, X_i)$ from $\pi$. Then, we have $\trajcost{\pi} \geq \trajcost{\pi'}$ if and only if one of the following holds:
    \begin{enumerate} [label = (\Roman*)]
        \item \label{cond:vw:rmp3} $a_i \leq x_i$ and $f(A_{i}, X_{i}) \geq f(A_{i+1}, X_{i+1})$;
        \item \label{cond:vw:rmp4} $a_i \geq x_i$ and $f(A_{i}, X_{i}) \leq f(A_{i+1}, X_{i+1})$,
    \end{enumerate}
    where $a_i = A_i - A_{i - 1}$ and $x_i = X_i - X_{i - 1}$.
\end{lemma}
\begin{proof}
Since $\pi'$ is obtained by removing $(A_i, X_i)$ from $\pi$, we have
\begin{align*}
    & \trajcost{\pi} - \trajcost{\pi'} \\
    & = \bigg[ (X_i - X_{i-1}) f(A_i, X_i) + (A_i - A_{i - 1}) (1 - f(A_i, X_i)) \\
    & \phantom{==} + (X_{i+1} - X_{i}) f(A_{i+1}, X_{i+1}) + (A_{i+1} - A_{i}) (1 - f(A_{i+1}, X_{i+1})) \bigg] \\
    & \phantom{=} - \bigg[ (X_{i+1} - X_{i-1}) f(A_{i+1}, X_{i+1}) + (A_{i+1} - A_{i - 1}) (1 - f(A_{i+1}, X_{i+1})) \bigg] \\
    & = (x_i - a_i) \cdot (f(A_i, X_i) - f(A_{i + 1}, X_{i+1}) ).
\end{align*}
Therefore, $\trajcost{\pi} - \trajcost{\pi'} \geq 0$ if and only if one of the conditions in the lemma holds.
\end{proof}

The following lemmas further identify pairs that can be removed from $\pi$ without increasing the cost.
\begin{lemma} \label{lem:vw:fr:2to1}
    Fix a trajectory $\pi = ((A_0, X_0), \ldots, (A_k, X_k))$ with $k \geq 2$. For some $i \in \{1, \ldots, k-1\}$, let $\pi'$ be a new trajectory obtained by removing $(A_i, X_i)$ from $\pi$. If $(A_{i-1}, X_{i - 1}) \in \region{L}$ and $(A_i, X_i) \in \region{BR}$ , we have $\trajcost{\pi} \geq \trajcost{\pi'}$.
\end{lemma}
\begin{proof}
    Since $(A_i, X_i) \in \region{BR}$ and $(A_{i - 1}, X_{i - 1}) \in \region{L}$, we have $a_i > x_i$.
    Moreover, since $(A_i, X_i) \in \region{BR}$, we also have $f(A_i, X_i) \leq f(A_{i+1}, X_{i + 1})$.
    These together imply that $i$ satisfies Condition \ref{cond:vw:rmp4} of \cref{lem:remove-middle-point}.
\end{proof}

\begin{lemma} \label{lem:vw:fr:no1btw2}
    Fix a trajectory $\pi = ((A_0, X_0), \ldots, (A_k, X_k))$ with $k \geq 2$.
    Suppose there exists $s$ and $t$ such that $0 \leq s < t \leq k$, $(A_s, X_s) \in \region{L}$, $(A_t, X_t) \in \region{L} \cup \region{TR}$, and $(A_{s+1}, X_{s+1}), \ldots, (A_{t-1}, X_{t - 1}) \in \region{BR}$.
    Let $\pi'$ be the trajectory obtained by removing $(A_{s+1}, X_{s+1}), \ldots, (A_{t-1}, X_{t - 1})$ from $\pi$.
    We have $\trajcost{\pi} \geq \trajcost{\pi'}$.
\end{lemma}
\begin{proof}
    Note that \cref{lem:vw:fr:2to1} holds with $\pi$ and $i := s+1$.
    Repeated application of \cref{lem:vw:fr:2to1} leads to the proof of this lemma.
\end{proof}

\begin{lemma} \label{lem:vw:fr:1to2}
    Fix a trajectory $\pi = ((A_0, X_0), \ldots, (A_k, X_k))$ with $k \geq 2$. For some $i \in \{1, \ldots, k-1\}$, let $\pi'$ be a new trajectory obtained by removing $(A_i, X_i)$ from $\pi$. If $(A_i, X_i) \in \region{L}$, $a_i \leq x_i$, $(A_{i+1}, X_{i+1}) \in \region{BR} \cup \region{L}$, and $a_{i+1} \geq x_{i+1}$, we have $\trajcost{\pi} \geq \trajcost{\pi'}$.
\end{lemma}
\begin{proof}
    Note that we have $f(A_i, X_i) \geq f(A_{i + 1}, X_{i + 1})$. Indeed, if $(A_{i + 1}, X_{i + 1}) \in \region{BR}$, it trivially follows; otherwise if $(A_{i + 1}, X_{i + 1}) \in \region{L}$, it is implied by the condition that $a_{i+1} \geq x_{i+1}$.
    Therefore, Condition \ref{cond:vw:rmp3} of \cref{lem:remove-middle-point} is satisfied with this $i$ on $\pi$.
\end{proof}

\begin{lemma} \label{lem:vw:fr:no2btw1}
    Fix a trajectory $\pi = ((A_0, X_0), \ldots, (A_k, X_k))$ with $k \geq 2$.
    Suppose there exists $s$ and $t$ such that $0 \leq s < t \leq k$, $(A_s, X_s), (A_t, X_t) \in \region{BR}$, and $(A_{s+1}, X_{s+1}), \ldots, (A_{t-1}, X_{t - 1}) \in \region{L}$.
    Let $\pi'$ be the trajectory obtained by removing $(A_{s+1}, X_{s+1}), \ldots, (A_{t-1}, X_{t - 1})$ from $\pi$.
    We have $\trajcost{\pi} \geq \trajcost{\pi'}$.
\end{lemma}
\begin{proof}
    Since $(A_s, X_s) \in \region{BR}$ and $(A_{s+1}, X_{s+1}) \in \region{L}$, we have $a_{s+1} < x_{s+1}$.
    Similarly, we have $a_t > x_t$ due to the condition that $(A_{t-1}, X_{t-1}) \in \region{L}$ and $(A_t, X_t) \in \region{BR}$.
    Therefore, there must exist $i \in \{s+1, \ldots, t-1\}$ such that $a_i \leq x_i$ and $a_{i+1} \geq x_{i+1}$.
    Due to \cref{lem:vw:fr:1to2} with $\pi$ and this $i$, removing this $(A_i, X_i)$ from $\pi$ does not increase the cost of the trajectory.
    This lemma then follows by repeatedly applying this process until all of $(A_{s+1}, X_{s+1}), \ldots, (A_{t-1}, X_{t-1})$ are removed from $\pi$.
\end{proof}

\begin{lemma} \label{lem:vw:fr:12to3}
    Fix a trajectory $\pi = ((A_0, X_0), \ldots, (A_k, X_k))$ with $k \geq 2$. For some $i \in \{1, \ldots, k-1\}$, let $\pi'$ be a new trajectory obtained by removing $(A_i, X_i)$ from $\pi$. If $(A_{i-1}, X_{i-1}) \in \region{BR} \cup \region{L}$, $(A_i, X_i) \in \region{TR}$, and $a_i \geq x_i$, we have $\trajcost{\pi} \geq \trajcost{\pi'}$.
\end{lemma}
\begin{proof}
    Since $(A_{i+1}, X_{i+1}) \in \region{TR}$, we have $f(A_i, X_i) \leq f(A_{i+1}, X_{i+1})$, immediately completing the proof of this lemma due to Condition \ref{cond:vw:rmp4} of \cref{lem:remove-middle-point}.
\end{proof}

Equipped with the above lemmas, we can now prove \cref{lem:vw:fr:reduce}.
\begin{proof}[Proof of \cref{lem:vw:fr:reduce}]
    We break the proof into two cases depending on the region where the last pair of $\pi$ is contained.

\begin{figure}
    \centering
    \begin{subfigure}{0.32\textwidth}
        \centering
        \includegraphics[width=\textwidth]{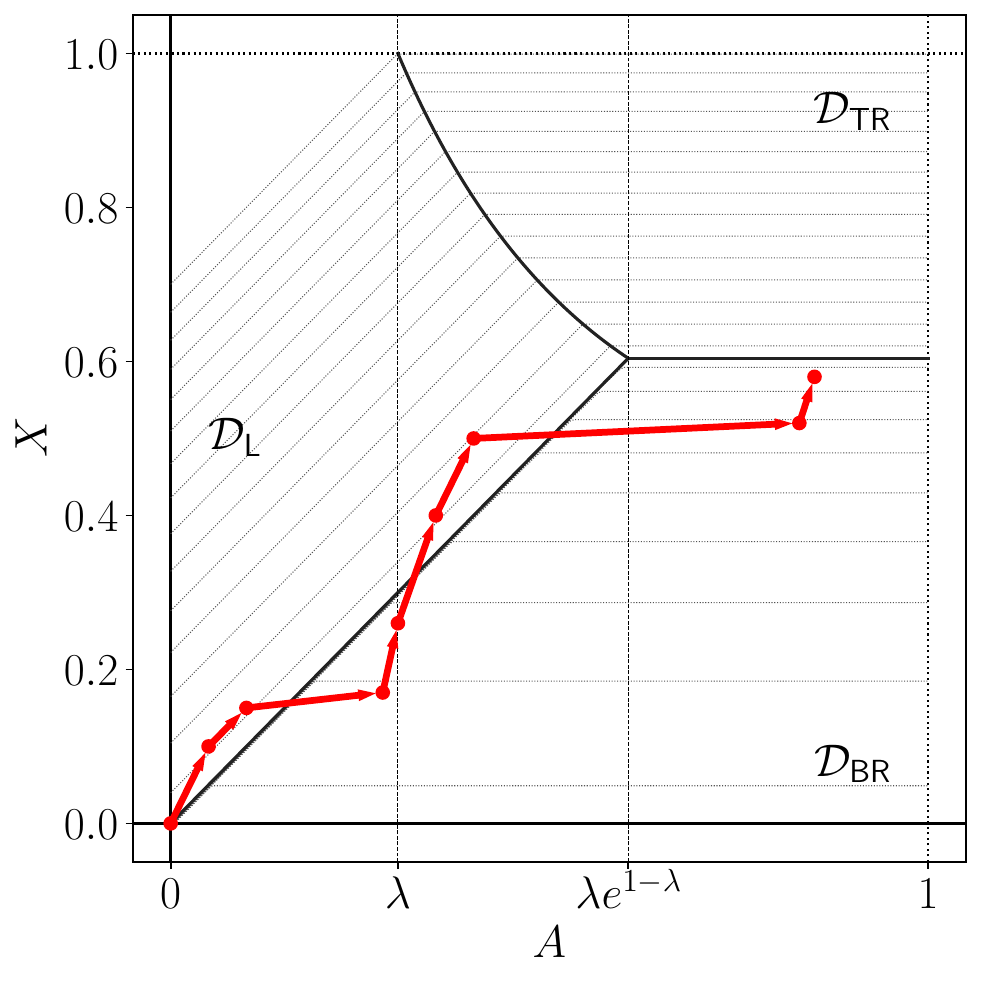}
        \caption{Initial trajectory}
    \end{subfigure}
    \begin{subfigure}{0.32\textwidth}
        \centering
        \includegraphics[width=\textwidth]{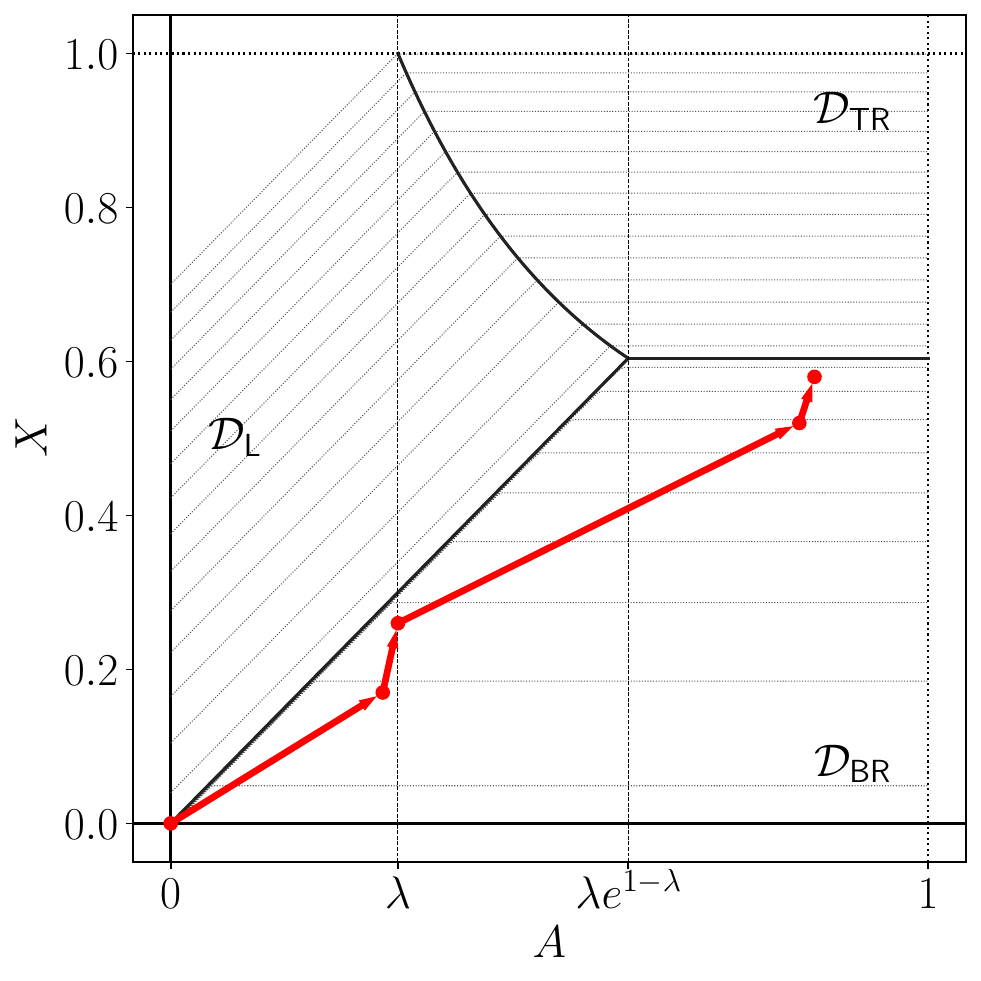}
        \caption{Due to Lemma \cref{lem:vw:fr:no2btw1}}
    \end{subfigure}
    \\
    \vspace{10pt}
    \begin{subfigure}{0.32\textwidth}
        \centering
        \includegraphics[width=\textwidth]{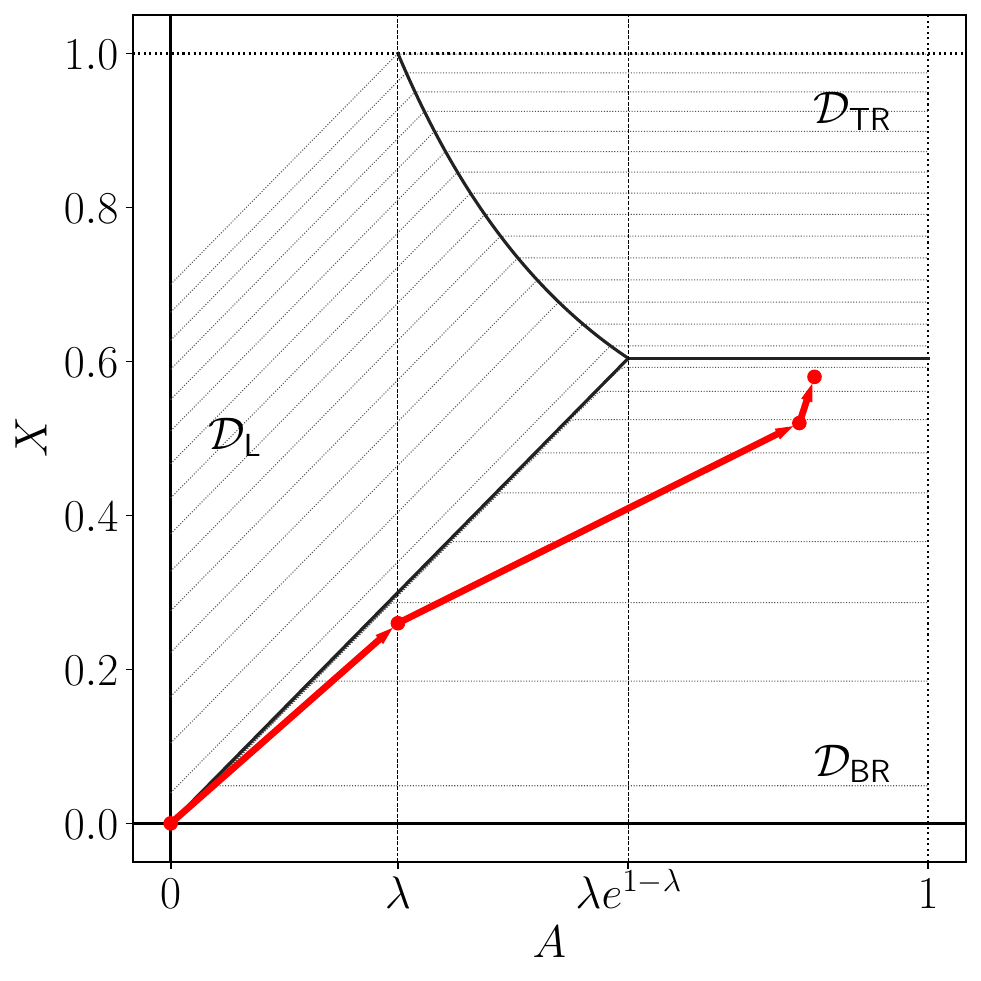}
        \caption{Due to Lemma \cref{lem:vw:fr:2to1}}
    \end{subfigure}
    \begin{subfigure}{0.32\textwidth}
        \centering
        \includegraphics[width=\textwidth]{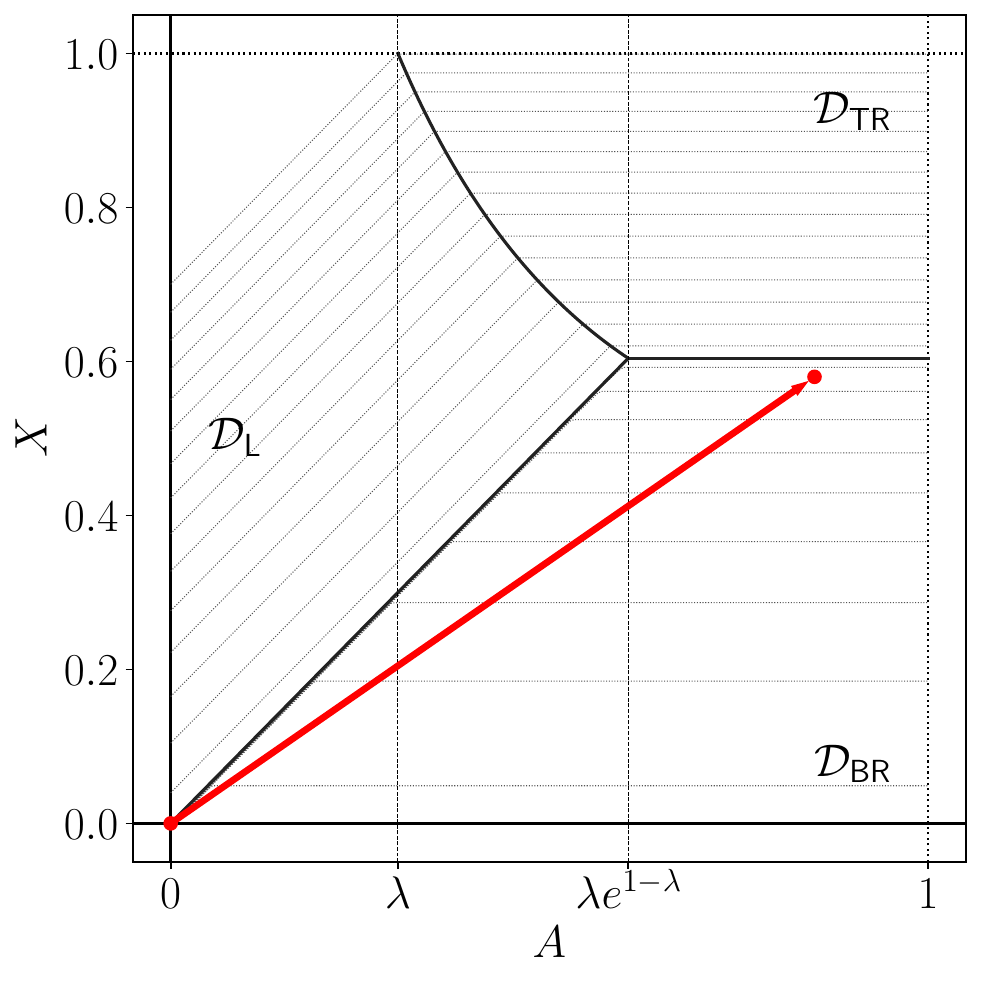}
        \caption{Irreducible trajectory}
    \end{subfigure}
    \caption{Illustration of the proof of \cref{lem:vw:fr:reduce} in the case where $(A_k, X_k) \in \region{BR}$.}
    \label{fig:vw:fr:BR}
\end{figure}

\begin{figure}
    \centering
    \begin{subfigure}{0.32\textwidth}
        \centering
        \includegraphics[width=\textwidth]{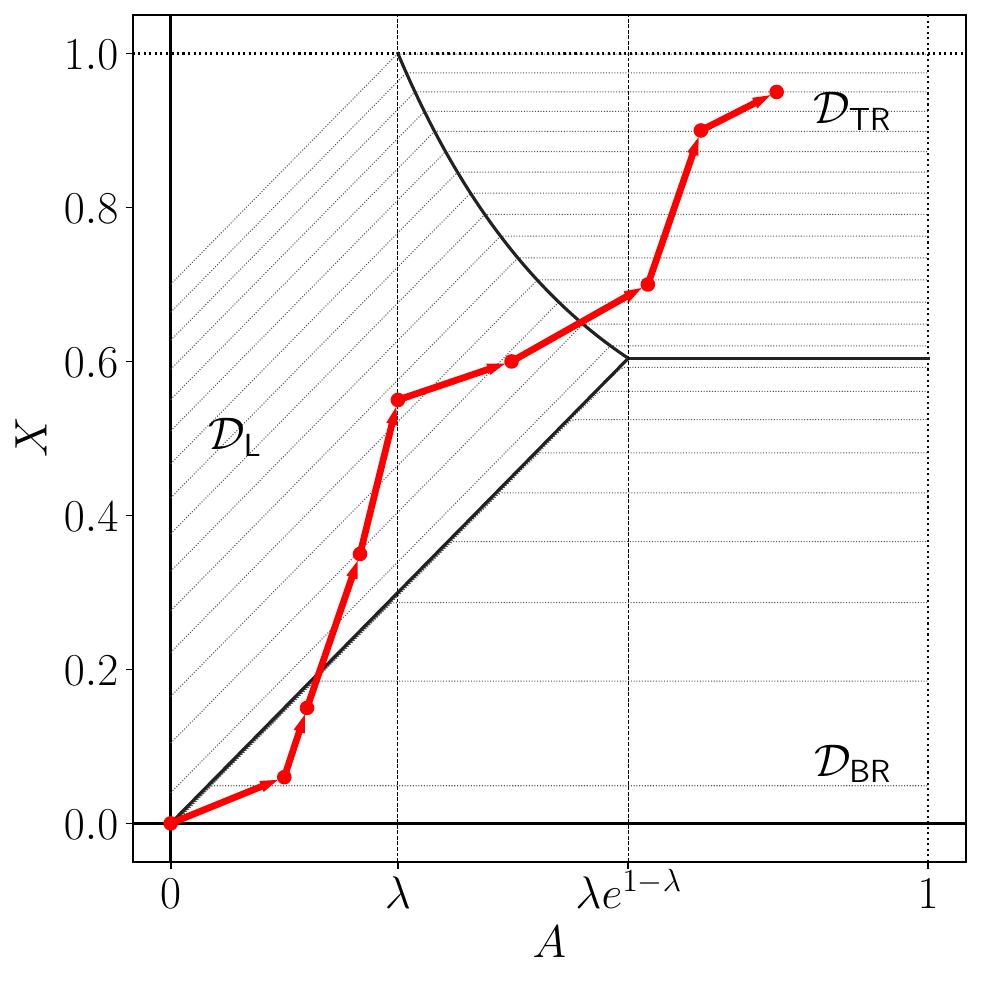}
        \caption{Initial trajectory}
    \end{subfigure}
    \begin{subfigure}{0.32\textwidth}
        \centering
        \includegraphics[width=\textwidth]{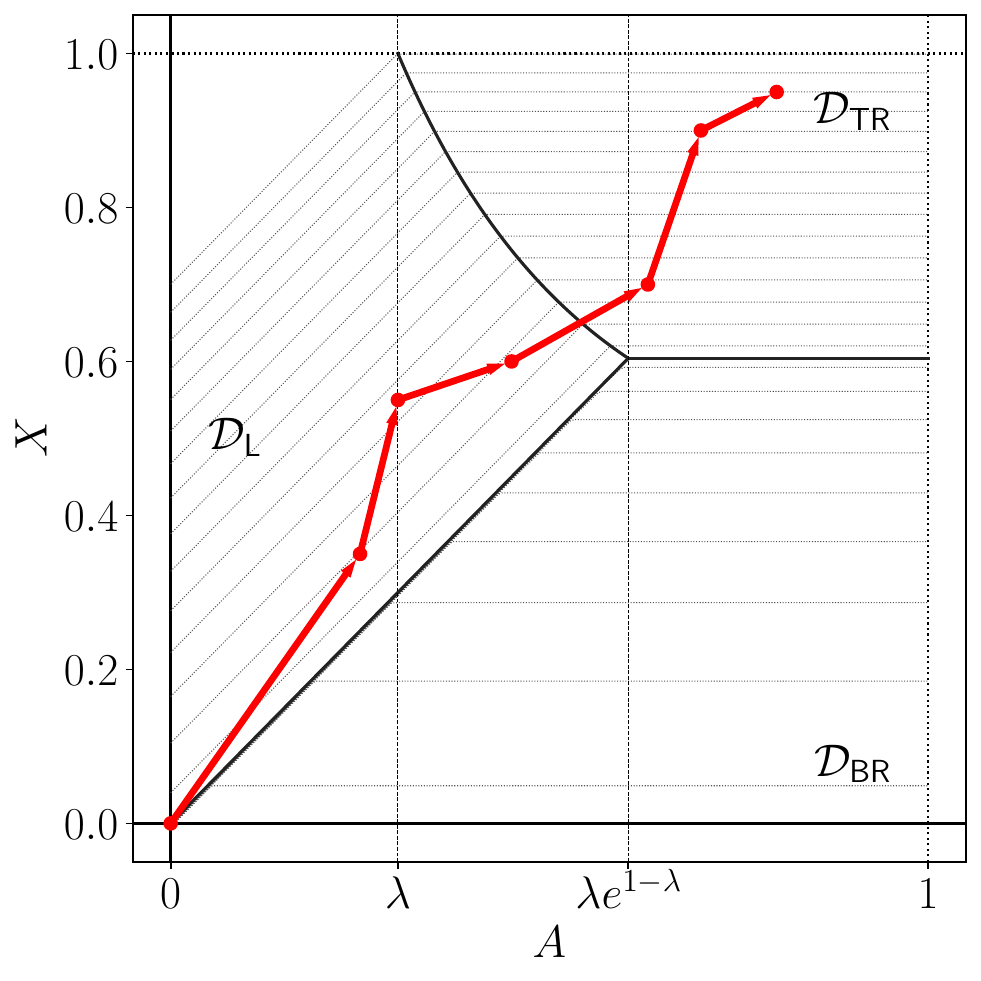}
        \caption{Due to Lemma \cref{lem:vw:fr:no1btw2}}
    \end{subfigure}
    \begin{subfigure}{0.32\textwidth}
        \centering
        \includegraphics[width=\textwidth]{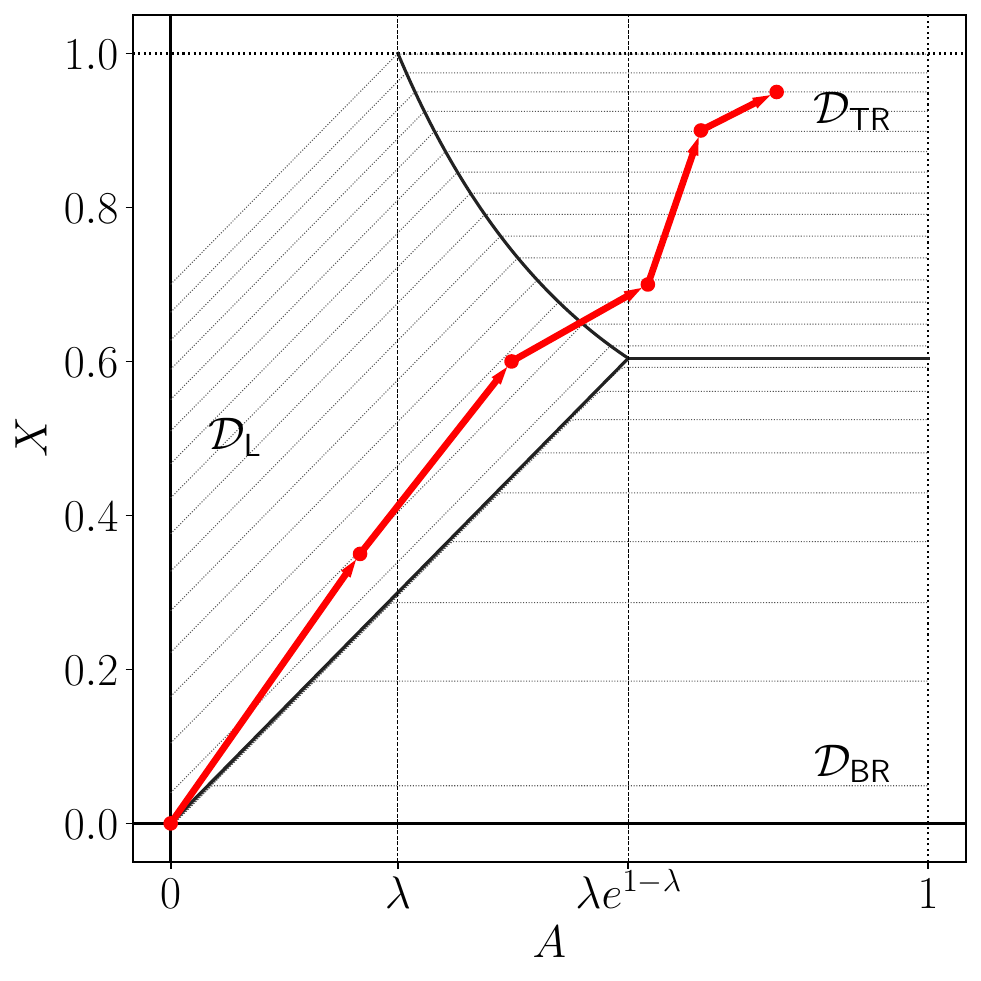}
        \caption{Due to Lemma \cref{lem:vw:fr:1to2}}
    \end{subfigure}
    \\
    \vspace{10pt}
    \begin{subfigure}{0.32\textwidth}
        \centering
        \includegraphics[width=\textwidth]{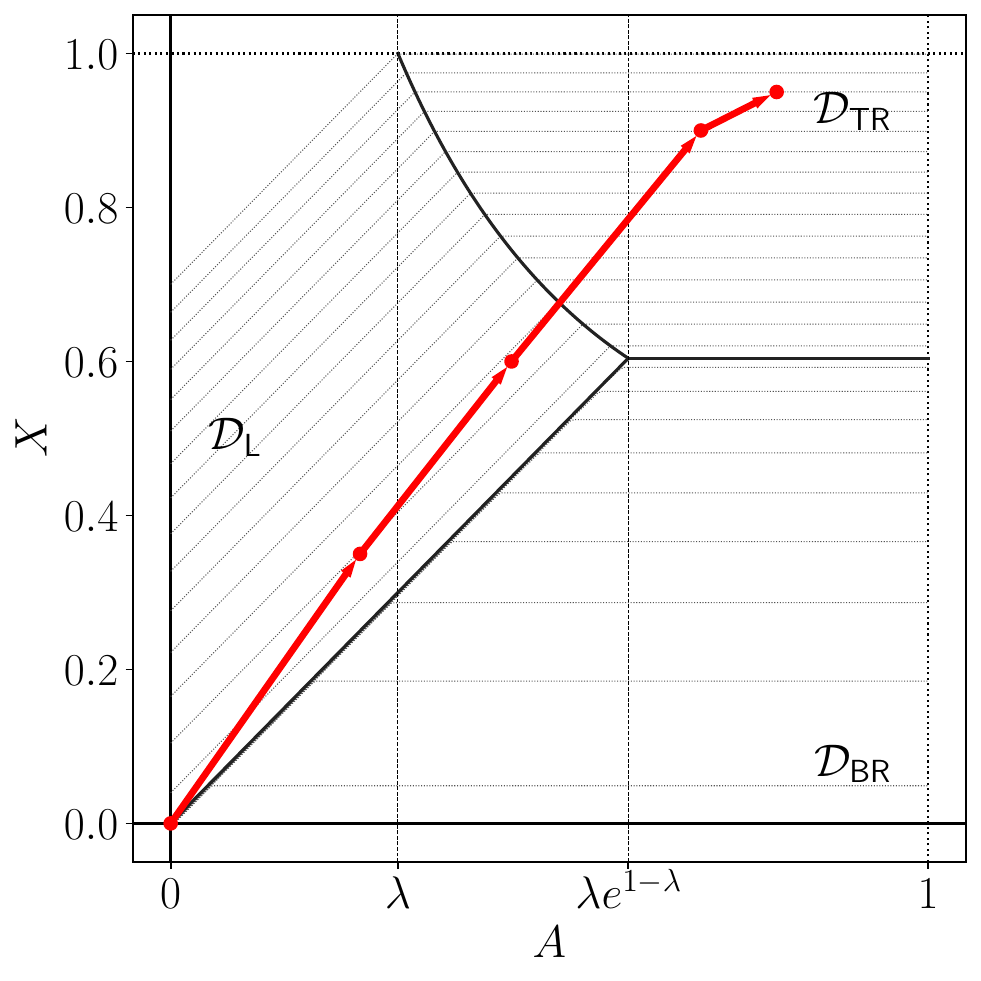}
        \caption{Due to Lemma \cref{lem:vw:fr:12to3}}
    \end{subfigure}
    \begin{subfigure}{0.32\textwidth}
        \centering
        \includegraphics[width=\textwidth]{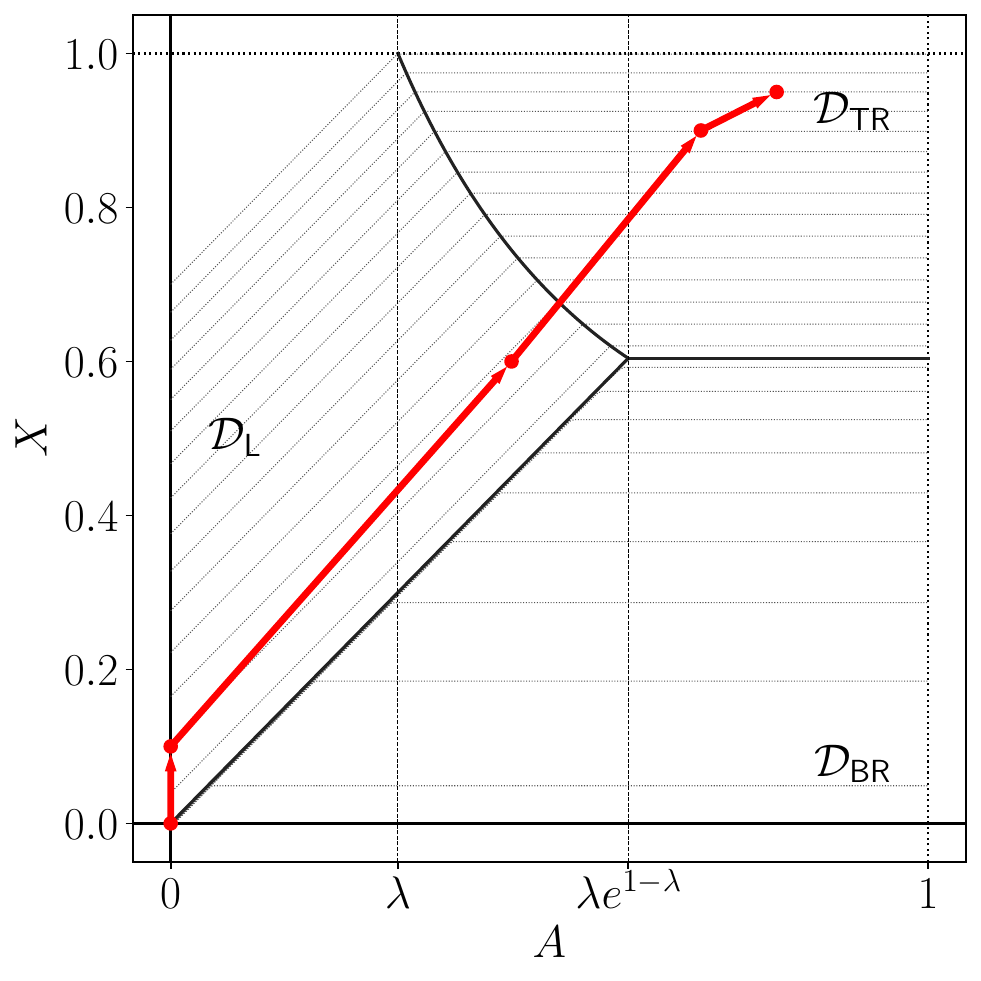}
        \caption{Intermediate trajectory $\pi_1$}
    \end{subfigure}
    \begin{subfigure}{0.32\textwidth}
        \centering
        \includegraphics[width=\textwidth]{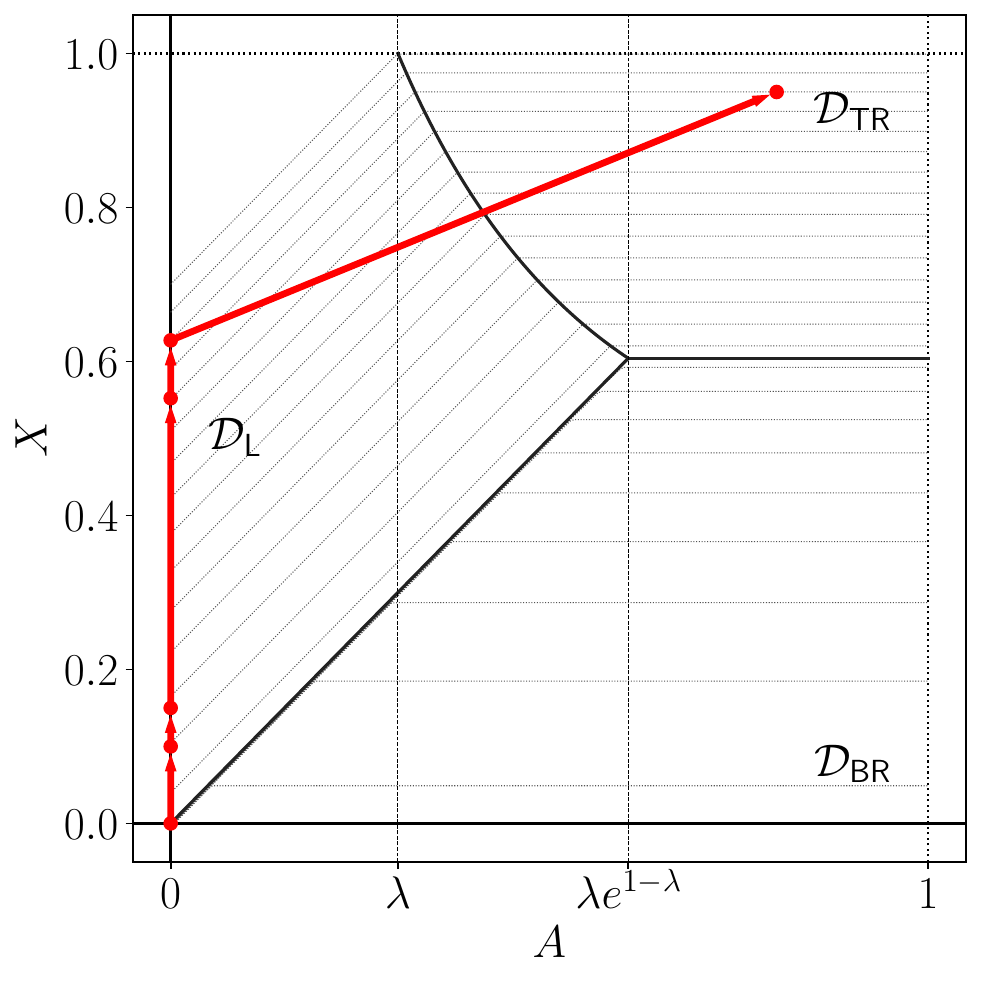}
        \caption{Irreducible trajectory $\pi_k$}
    \end{subfigure}
    \caption{Illustration of the proof of \cref{lem:vw:fr:reduce} in the case where $(A_k, X_k) \in \region{L} \cup \region{TR}$.}
    \label{fig:vw:fr:TR}
\end{figure}   
   
    \textbf{Case 1. $(A_k, X_k) \in \region{BR}$.}
    See \cref{fig:vw:fr:BR} for an illustration of the proof of this case.
    Due to \cref{lem:vw:fr:no2btw1}, we can assume that every pair in $\pi$ other than $(A_0, X_0) = (0, 0)$ is contained in $\region{BR}$.
    Note however that \cref{lem:vw:fr:2to1} is satisfied with $i := 1$ on $\pi$.
    We can therefore remove $(A_1, X_1)$ from $\pi$ without increasing the cost of the trajectory.
    By repeatedly applying this process, we end up with a trajectory $\pi' := ((0, 0), (A_1, X_1))$ with $(A_1, X_1) \in \region{BR}$.
    Recall that this $\pi'$ is irreducible, completing the proof of \cref{lem:vw:fr:reduce} for this case.
   
    \textbf{Case 2. $(A_k, X_k) \in \region{L} \cup \region{TR}$.}
    See \cref{fig:vw:fr:TR} for an illustration of the proof of this case.
    Note that, due to \cref{lem:vw:fr:no1btw2}, we can assume that every pair in $\pi$ is contained in $\region{L} \cup \region{TR}$.
    
    Let us now argue that we can further assume without loss of generality that the $f$ values of every pair in $\pi$ does not decrease, i.e.,
    \begin{equation} \label{eq:vw:fr:fmonotone}
        f(A_0, X_0) \leq \cdots \leq f(A_k, X_k).    
    \end{equation}
    Let $i^*$ be the last index of a pair contained in $\region{L}$, i.e., $(A_0, X_0), \ldots, (A_{i^*}, X_{i^*}) \in \region{L}$ and $(A_{i^*+1}, X_{i^*+1}), \ldots, (A_k, X_k) \in \region{TR}$.
    We may have either $i^* = 0$ (i.e., no pairs other than $(0, 0)$ are in $\region{L}$) or $i^* = k$ (i.e., no pairs are in $\region{TR}$).
    If $i^* > 0$, by \cref{lem:vw:fr:1to2} and the fact that $a_1 \leq x_1$, we can assume $a_i \leq x_i$ for every $i = 1, \ldots, i^*$, implying that $f(A_0, X_0) \leq \cdots \leq f(A_{i^*}, X_{i^*})$.
    Furthermore, if $i^* < k$, by \cref{lem:vw:fr:12to3}, we can also assume $a_{i^*+1} < x_{i^*+1}$, yielding that $f(A_{i^*}, X_{i^*}) < f(A_{i^*+1}, X_{i^*+1})$.
    For the remaining indices, we can easily see that $f(A_{i^*+1}, X_{i^*+1}) \leq \cdots \leq f(A_k, X_k)$ since $(A_{i^*+1}, X_{i^*+1}), \ldots, (A_k, X_k) \in \region{TR}$.

    For $i = 1, \ldots, k$, let $X'_i$ be the value such that $f(0, X'_i) = f(A_i, X_i)$; if $f(A_i, X_i) = 1$, we set $X'_i := \lambda$.
    Due to \cref{eq:vw:fr:fmonotone} and the fact that $f(0, \cdot)$ is strictly increasing in $[0, \lambda]$, we have $0 \leq X'_1 \leq \cdots \leq X'_k \leq \lambda$.
    Moreover, by the definition of $f$, we also have
    \begin{equation} \label{eq:vw:fr:yaxis}
        A_i \geq X_i - X'_i \text{ \; if } f(A_i, X_i) < 1.
    \end{equation}
    
    For $i = 1, \ldots, k-1$, let $\pi_i$ be the trajectory where $(A_1, X_1), \ldots, (A_i, X_i)$ in $\pi$ are replaced by $(0, X'_1), \ldots, (0, X'_i)$, respectively, i.e.,
    \[
        \pi_i := ((0, 0), (0, X'_1), \ldots, (0, X'_i),(A_{i+1}, X_{i+1}), \ldots, (A_k, X_k)).
    \]
    For consistency, let $\pi_0 := \pi$.
    We claim that, for every $i = 1, \ldots, k-1$, we have $\trajcost{\pi_{i-1}} \geq \trajcost{\pi_i}$.
    Indeed, since $\pi_{i-1}$ and $\pi_i$ differ only at index $i$, we have
    \begin{align*}
        & \trajcost{\pi_{i-1}} - \trajcost{\pi_i} \\
        & = \bigg[ (X_i - X'_{i-1}) f(A_i, X_i) + A_i (1 - f(A_i, X_i)) \\
        & \phantom{==} + (X_{i+1} - X_i) f(A_{i+1}, X_{i + 1}) + (A_{i+1} - A_i) (1 - f(A_{i+1}, X_{i + 1}))\bigg] \\
        & \phantom{=} - \bigg[ (X'_i - X'_{i-1}) f(0, X'_i) \\
        & \phantom{==} + (X_{i+1} - X'_i) f(A_{i+1}, X_{i + 1}) + A_{i+1} (1 - f(A_{i+1}, X_{i + 1})) \bigg] \\
        & = (f(A_{i+1}, X_{i+1}) - f(A_i, X_i)) \cdot (A_i - X_i + X'_i) \\
        & \geq 0,
    \end{align*}
    where the second equality comes from that $f(0, X'_i) = f(A_i, X_i)$ and the inequality from \cref{eq:vw:fr:fmonotone,eq:vw:fr:yaxis}.
    In particular, if $A_i < X_i - X'_i$, we have $f(A_i, X_i) = 1$ due to \cref{eq:vw:fr:yaxis}, and hence, $f(A_{i+1}, X_{i+1}) = 1$ due to \cref{eq:vw:fr:fmonotone}.

    Finally, let $\pi_k$ be the trajectory where $(0, X'_k)$ is inserted between $(0, X'_{k-1})$ and $(A_k, X_k)$ in $\pi_{k-1}$, i.e.,
    \[
        \pi_k := ((0, 0), (0, X'_1), \ldots, (0, X'_{k-1}), (0, X'_k), (A_k, X_k)).
    \]
    Note that $\trajcost{\pi_k} = \trajcost{\pi_{k-1}}$ since we have
    \begin{align*}
        & \trajcost{\pi_{k-1}} - \trajcost{\pi_k} \\
        & = \bigg[ (X_k - X'_{k-1}) f(A_k, X_k) + A_k (1 - f(A_k, X_k)) \bigg] \\
        & \phantom{=} - \bigg[ (X'_k - X'_{k-1}) f(0, X'_i) + (X_k - X'_k) f(A_k, X_k) + A_k (1 - f(A_k, X_k)) \bigg] \\
        & = 0
    \end{align*}
    due to the fact that $f(0, X'_k) = f(A_k, X_k)$.
    Observe that $\pi_k$ is irreducible, completing the proof of \cref{lem:vw:fr:reduce}.
\end{proof}

The proof of \cref{lem:vwalg:frac:main} then immediately follows from \cref{lem:vw:fr:irr,lem:vw:fr:reduce}.

\subsection{Extension to AdWords} \label{sec:vwalg:adw}
In this subsection, we present that \LearnAugBal{} naturally extends to AdWords under the small bids assumption, showing \cref{thm:adw:main} restated below.

\adwordsmain*

We first adapt \LAB{} into an algorithm for \emph{fractional} AdWords achieving the same robustness and consistency.
We then argue that integral AdWords under the small bids assumption can be reduced to fractional AdWords with small loss.

Recall that, in AdWords, offline $U$ and online $V$ are corresponding to \emph{advertisers} and \emph{impressions}, respectively.
Each advertiser $u \in U$ has a budget of $B_u$.
Whenever an impression $v \in V$ is revealed, the algorithm also learns bids $\{b_{u, v}\}_{u \in U}$ from the advertisers; the advice is the advertiser to which this impression should be assigned.
The algorithm then needs to assign each impression to an advertiser, making a revenue of its bid unless the advertiser has used up its budget.
The objective is to maximize the total revenue.
Recall from \cref{sec:prelim} the LP relaxation for AdWords and its dual LP:

\begin{minipage}{\textwidth}
    \centering
    \begin{minipage}[t]{0.49 \textwidth}
        \begin{align*}
            \max\; & \textstyle \sum_{u \in U} \sum_{v \in V} b_{u, v} x_{u,v} \\
            \text{s.t.}\; &\textstyle  \frac{1}{B_u} \sum_{v \in V} b_{u,v} x_{u, v} \leq 1, & \forall u \in U, \\
            &\textstyle  \sum_{u \in U} x_{u, v} \leq 1, & \forall v \in V, \\
            & x_{u, v} \geq 0, & \forall u \in U, v \in V; \\
        \end{align*}
    \end{minipage}
    \begin{minipage}[t]{0.49 \textwidth}
        \begin{align*}
            \min \; &\textstyle  \sum_{u \in U} \alpha_u + \sum_{v \in V} \beta_v \\
            \text{s.t.} \; &\textstyle  \frac{b_{u, v}}{B_u} \alpha_u + \beta_v \geq b_{u, v}, &  \forall u \in U, v \in V, \\
            & \alpha_u \geq 0, & \forall u \in U, \\
            & \beta_v \geq 0, & \forall v \in V.
        \end{align*}
    \end{minipage}
\end{minipage}

The fractional version of AdWords is corresponding to constructing online a feasible solution to the above primal LP relaxation.
We also admit a fractional advice; we again denote by $a \in \R^{U \times V}$ the advice feasible to the primal LP.

\paragraph{Fractional algorithm.}
When an impression $v \in V$ arrives together with the advice $\{a_{u, v}\}_{u \in U}$,  we denote by $A_u := \frac{1}{B_u} \sum_{t \preceq v} a_{u, t}$ for each advertiser $u \in U$ the fraction of $u$'s budget spent by the advice up to and including $v$.
Let $f$ be the same function defined in \cref{eq:vw:penaltyfrac} for vertex-weighted matching.
The algorithm then continuously assigns an infinitesimal unit of impression $v$ to an advertiser $u \in U$ with the highest value of $b_{u, v} (1 - f(A_u, X_u))$, where $X_u$ denotes the fraction of $u$'s budget spent by the algorithm right before this infinitesimal unit is assigned, until either impression $v$ is integrally assigned or the budgets of all advertisers are used up.

\paragraph{Primal-dual analysis.}
Recall from \cref{sec:prelim} the key lemma of the primal-dual analysis for AdWords:

\begin{lemma} [cf. \cref{lem:pre:adw:pda}] \label{lem:vw:adw:pd}
    Let $x \in \R^{U \times V}$ be the output of a fractional algorithm for AdWords.
    For some $\rho \in [0, 1]$, if there exists $(\alpha, \beta) \in \R^U \times \R^V$ satisfying
    \begin{itemize}
        \item (reverse weak duality) $ \sum_{u \in U} \sum_{v \in V} b_{u, v} x_{u, v} \geq \sum_{u \in U} \alpha_u + \sum_{v \in V} \beta_v$ and
        \item (approximate dual feasibility) $ \frac{b_{u, v}}{B_u} \alpha_u + \beta_v \geq \rho \cdot b_{u, v} \;$ for every $u \in U$ and $v \in V$,
    \end{itemize}
    we have $ \ALG \geq \rho \cdot \OPT$.
\end{lemma}

The construction of dual variables is almost identical to the vertex-weighted setting.
We first initialize $(\alpha, \beta) \gets (\mathbf{0}, \mathbf{0})$.
At the iteration when impression $v \in V$ is revealed, for each advertiser $u \in U$, let $x_{u, v}$ denote the fraction of $v$ assigned to $u$.
Let us also denote by $\Xup{v}_u := \frac{1}{B_u} \sum_{t \preceq v} b_{u, t} x_{u, t}$ and $\Aup{v}_u := \frac{1}{B_u} \sum_{t \preceq v} b_{u, t} a_{u, t}$ the fractions of $u$'s budget spent by the algorithm and the advice, respectively, at the end of this iteration.
We set
\begin{itemize}
    \item $\alpha_u \gets \alpha_u + x_{u, v} \cdot b_{u,v} f(\Aup{v}_u,\Xup{v}_u)$ for every $u \in U$, and
    \item $\beta_v \gets \max_{u \in U} \left\{ b_{uv} (1-f(\Aup{v}_u, \Xup{v}_u)) \right\}$.  
\end{itemize}
By the definition of the algorithm, observe that $b_{u, v}(1-f(\Aup{v}_u, \Xup{v}_u))$ is constant for all $u \in U$ with $x_{u, v} > 0$ and that $\beta_v$ is equal to this value.
Also, note that if $\sum_{u \in U} x_{u, v} < 1$, then $\beta_v = 0$ due to the fact that $\sum_{u \in U} x_{u, v} < 1$ implies all the advertisers have used up their budgets.

We now show that \cref{lem:pd_equal_vertex_weighted,lem:vw_r,lem:vw_c} can be easily adapted to fractional AdWords as follows.

\begin{lemma} [cf. \cref{lem:pd_equal_vertex_weighted}] \label{lem:vw:adw:rwd}
    The revenue of the algorithm is equal to the objective value of $(\alpha, \beta)$ in the dual LP. 
\end{lemma}
\begin{proof}
    Let $\ALG$ and $\DUAL$ denote the revenue of the algorithm and the objective value of $(\alpha, \beta)$ at any iteration.
    As we have $\ALG = \DUAL = 0$ at the start, it suffices to show that $\Delta \ALG = \Delta \DUAL$ at each iteration.
    When impression $v$ is revealed, we have
    $
        \Delta \ALG = \sum_{u \in U} b_{u,v} x_{u, v}.
    $
    Let us now calculate $\Delta \DUAL$.
    For clarity, let $A_u := \Aup{v}_u$ and $X_u := \Xup{v}_u$.
    We then have
    \begin{align*}
        \Delta \DUAL
        & = \sum_{u \in U} \Delta \alpha_u + \beta_v \\
        & \stackrel{(a)}{=} \sum_{u \in U} x_{u, v} \cdot b_{u,v} f(A_u, X_u) + \sum_{u \in U} x_{u, v} \cdot \beta_v \\
        & = \sum_{u \in U} x_{u, v} \cdot \left(b_{u,v} f(A_u, X_u) + \beta_v\right)\\
        & \stackrel{(b)}{=} \sum_{u \in U} x_{u, v} \cdot \left(b_{u,v} f(A_u, X_u) + b_{u,v}(1 - f(A_u, X_u))\right) \\
        & = \sum_{u \in U} b_{u,v} x_{u, v} \\
        & = \Delta \ALG,
    \end{align*}
    where (a) follows from that $\beta_v = 0$ if $\sum_{u \in U}x_{u, v} < 1$, and (b) from that $\beta_v = b_{u,v} (1 - f(A_u, X_u))$ if $x_{u, v} >  0$.
\end{proof}

For the counterparts of \cref{lem:vw_r,lem:vw_c}, we define another notation: for every $u \in U$ and $v \in V$, let $\Delta \Xup{v}_u := \frac{b_{u, v} x_{u, v}}{B_u}$ and $\Delta \Aup{v}_u := \frac{b_{u, v} a_{u, v}}{B_u}$ denote the increase of fractions of $u$'s budget spent by the algorithm and the advice, respectively, at iteration when $v$ is revealed.

\begin{lemma}[cf. \cref{lem:vw_r}] \label{lem:vw:adw:r}
    The algorithm is $r$-robust for any $r$ satisfying that, for any $u \in U$ and $v \in V$, 
    \begin{equation*}
        r \leq  \int_0^{\Xup{v}_u} f(\Aup{v}_u,z)dz + (1 - f(\Aup{v}_u, \Xup{v}_u)).
    \end{equation*}   
\end{lemma}
\begin{proof}
    Due to \cref{lem:vw:adw:pd,lem:vw:adw:rwd}, it suffices to prove a bound of the form
    $$
        \frac{b_{u,v}}{B_u}\alpha_u + \beta_v \geq r \cdot b_{u,v}  \text{for all $u \in U$ and $v \in V$}.
    $$
    Note that, for any $u \in U$ and $v \in V$, we have
    \begin{align*}
        \frac{b_{u,v}}{B_u}\alpha_u + \beta_v
        & \stackrel{(a)}{\geq} \frac{b_{u,v}}{B_u}\sum_{t \preceq v}b_{u,t} x_{u, t} \, f(\Aup{t}_u, \Xup{t}_u) + \max_{u' \in U} \left\{  b_{u',v}(1 - f(\Aup{v}_{u'}, \Xup{v}_{u'})) \right\} \\
        &\geq \frac{b_{u,v}}{B_u}\sum_{t \preceq v} b_{u,t} x_{u, t} \, f(\Aup{t}_u, \Xup{t}_u) +  b_{u,v} (1 - f(\Aup{v}_u, \Xup{v}_u)) \\
        &\stackrel{(b)}{\geq} b_{u,v}\sum_{t \preceq v} \Delta \Xup{t}_u f(\Aup{v}_u, \Xup{t}_u) +  b_{u,v}(1 - f(\Aup{v}_u, \Xup{v}_u))  \\
        &\stackrel{(c)}{\geq} b_{u,v} \cdot \left[ \int_0^{\Xup{v}_u} f(\Aup{v}_u, z) dz + (1 - f(\Aup{v}_u, \Xup{v}_u)) \right],
    \end{align*}
    where (a) is because $\alpha_u$ does not decrease throughout the execution, (b) is because $f(A,X)$ is decreasing in $A$, and (c) is because  $f(A,X)$ is increasing in $X$. 
\end{proof}

\begin{lemma} [cf. \cref{lem:vw_c}] \label{lem:vw:adw:c}
    The algorithm is $c$-consistent, for any value of $c$ satisfying that, for every $u \in U$,
    \begin{equation*}
        \sum_{t \in N(u)} \left[ \Delta \Xup{t}_u \cdot f(\Aup{t}_u, \Xup{t}_u) + \Delta \Aup{t}_u \cdot (1-f(\Aup{t}_u, \Xup{t}_u)) \right]
        \geq c\cdot A_u,
    \end{equation*}
    where $A_u := \sum_{t \in V} a_{u, t}$ denotes the fraction of the budget $u$ eventually spends by the advice.
\end{lemma}
\begin{proof}
    Our goal here is to prove that
    $
        \ALG \geq c \cdot \ADVICE, 
    $
    where $\ALG$ is the revenue earned by the algorithm, and $\ADVICE$ is the revenue of the advice.
    On the one hand, we have  
    \[
        \ADVICE = \sum_{u \in U} B_u A_u.
    \]
    On the other hand, due to \cref{lem:vw:adw:rwd}, we have
    \begin{align*}
        \ALG 
        & = \sum_{u \in U} \alpha_u + \sum_{t \in V} \beta_t \\
        & \geq \sum_{u \in U}\alpha_u + \sum_{t \in V} \left( \beta_t \cdot \sum_{u \in U} a_{u, t} \right) \\
        & = \sum_{u \in U}\alpha_u + \sum_{u \in U} \sum_{t \in V} a_{u, t}  \beta_t\\
        & = \sum_{u \in U} \left(\alpha_u + \sum_{t \in V} a_{u, t} \beta_t \right),
    \end{align*}
    where the inequality is due to the feasibility of the advice $a$.
    Therefore, to show $\ALG \geq c \cdot \ADVICE$, it suffices to show
    \begin{equation}\label{eq:consistency-adwords}
        \alpha_u + \sum_{t \in V} a_{u, t} \beta_t \geq c \cdot B_u A_u  \text{for all $u \in U$.}
    \end{equation}
    By construction, observe that the left-hand side of \cref{eq:consistency-adwords} is bounded by
    \begin{align*}
        \alpha_u + \sum_{t \in V} a_{u, t} \beta_t 
        &\geq \sum_{t \in V} x_{u, t} \cdot b_{u,t} f(\Aup{t}_u,\Xup{t}_u) + \sum_{t \in V}  a_{u, t} \cdot b_{u,t} (1-f(\Aup{t}_u, \Xup{t}_u)) \\
        &=   \sum_{t \in V} \left[ b_{u,t}x_{u, t} \cdot f(\Aup{t}_u, \Xup{t}_u) + b_{u,t}a_{u, t} \cdot (1-f(\Aup{t}_u, \Xup{t}_u)) \right].
    \end{align*}
    Therefore, \cref{eq:consistency-adwords} holds for any value of $c$ satisfying that, for all $u \in U$,
    \[
        \sum_{t \in V} \left[ \Delta \Xup{t}_u \cdot f(\Aup{t}_u, \Xup{t}_u) + \Delta \Aup{t}_u \cdot (1-f(\Aup{t}_u, \Xup{t}_u)) \right] \geq c \cdot A_u
    \]
\end{proof}

We can thus show the robustness and consistency of this fractional algorithm for AdWords with fractional advice.
\begin{theorem}
    For any tradeoff parameter $\lambda \in [0, 1]$, this algorithm is an $r(\lambda)$-robust and $c(\lambda)$-consistent algorithm for fractional AdWords with a fractional advice, where $r(\lambda)$ and $c(\lambda)$ are defined in \cref{thm:vertex_weighted}.
\end{theorem}
\begin{proof}
    Immediate from \cref{lem:vw:adw:pd,lem:vw:adw:rwd,lem:vw:adw:r,lem:vw:adw:c,lem:vw:frac:r,lem:vwalg:frac:main}.
\end{proof}

\paragraph{Reducing integral AdWords to fractional AdWords.}
It remains to see that integral AdWords under the small bids assumption can be reduced to fractional AdWords with small loss.
We remark that the reduction is inspired by \cite{feng2024batching}.

Recall that, under the small bids assumption, there exists a sufficiently small $\varepsilon > 0$ such that $b_{u, v} \leq \varepsilon B_u$ for all $u \in U$ and $v \in V$.
We will use the following concentration bound in the analysis of the reduction.
\begin{proposition}[Bernstein’s inequality for bounded independent variables]
\label{prop:bernstein}
Let $Z_1,\dots,Z_n$ be independent random variables satisfying that, for some constant $c > 0$,
\[
    \mathbb{E}[Z_i]=0
    \quad\text{and}\quad
    |Z_i|\le c
    \;\;\text{almost surely} \quad
\]
for all $i = 1, \ldots, n$.
Define the aggregate variance
$
    \sigma^{2} := \sum_{i=1}^{n}\operatorname{Var}(Z_i) = \sum_{i = 1}^n \E[Z_i^2].
$
Then, for every $t>0$,
\[
    \Pr\!\Bigl[\sum_{i=1}^{n} Z_i \;\ge t\Bigr]
    \;\le\;
    \exp\!\Bigl(
        -\,\frac{t^{2}}{2\sigma^{2}+\tfrac{2}{3} c\,t}
    \Bigr).
\]
\end{proposition}

We now prove that the reduction is possible with small loss.
\begin{theorem}
    For sufficiently small $\varepsilon > 0$, given any algorithm for fractional AdWords, we can construct a randomized algorithm for integral AdWords satisfying 
    \[
        \E[\ALGI] \geq \left(1 - 3\sqrt{\varepsilon \ln (1/\varepsilon)} \right)\ALGF,
    \]
    where $\ALGF$ and $\ALGI$ denote the revenue earned by the given fractional algorithm and the constructed randomized algorithm, respectively.
\end{theorem}
\begin{proof}
    We construct the integral algorithm as follows.
    When an impression $v \in V$ arrives, let $\{x_{u, v}\}_{u \in U}$ be the (fractional) assignment of the given fractional algorithm at this moment.
    The integral algorithm then simply samples an advertiser $u \in U$ independently with probability $\gamma x_{u,v}$ for some $\gamma \in (0, 1 - \varepsilon)$ to be chosen later, and assigns $v$ to $u$ only if $u$ has enough budget remaining.
    Otherwise, if the algorithm samples no advertisers or the sampled $u$ has insufficient budget, the algorithm does nothing and receives the next impression.
    
    Observe that we have
    \begin{align}
        \E[\ALGI]
        &= \sum_{u \in U} \sum_{v \in V} b_{u,v} \Pr\left[\,\text{$v$ is assigned to $u$}\,\right] \nonumber \\
        &= \sum_{u \in U} \sum_{v \in V} b_{u,v} \cdot \big[ \gamma x_{u,v} \cdot \Pr\left[\,\text{$u$ has enough budget to pay $b_{u, v}$ at $v$'s arrival} \,\right] \big] \nonumber \\
        &\geq \sum_{u \in U} \sum_{v \in V} b_{u,v} x_{u,v} \cdot \gamma \Pr\left[\,\text{$u$ has enough budget to pay $b_{u, v}$ at the end}\,\right]  \nonumber \\
        &\geq {\ALGF} \cdot \gamma \, \min_{u \in U} \Pr\left[\,\text{$u$ has enough budget to pay $\max_{v \in V} b_{u, v}$ at the end}\,\right]. \label{eq:intfrac}
    \end{align}

    We now claim that, for any advertiser $u \in U$,
    \begin{equation*}
        \Pr\left[\,\text{$u$ has enough budget to pay $\max_{v \in V} b_{u, v}$ at the end}\,\right]
        \geq 1 - \exp\left( -\frac{(1 - \gamma - \varepsilon)^2}{2\varepsilon}\right).
    \end{equation*}
    To this end, let us fix $u \in U$ and define
    \[
        X_{u, v} := \I \{ \text{ $u$ is sampled at $v$'s arrival} \}
    \]
    for every impression $v \in V$.
    Notice that $X_{u, v}$ is Bernoulli with mean $\gamma x_{u, v}$, and $ \{X_{u, v} \}_{v \in V}$ are mutually independent.
    Moreover, since the fractional algorithm outputs a feasible fractional assignment $x \in \R^{U \times V}$, we have
    \[
        \E \left[ \sum_{v \in V} b_{u, v} X_{u, v} \right]
        = \gamma \sum_{v \in V} b_{u, v} x_{u, v}
        \leq \gamma B_u.
    \]
    Let us further define, for every $v \in V$,
    \[
        Y_{v} := \frac{b_{u,v}X_{u,v}}{B_u} \in [0, \varepsilon] 
        \;\;\text{and}\;\;
        Z_v := Y_v - \E[Y_v].
    \]
    Observe that $\E[\sum_{v \in V} Y_v] \leq \gamma$, $\E[Z_v] = 0$, $|Z_v| \leq \varepsilon$, and $\Var(Z_v) = \Var(Y_v) = \frac{b_{u, v}^2}{B_u^2} \Var(X_{u, v})$.
    Moreover, note that $\{ Z_v \}_{v \in V}$ are also mutually independent.
    We can therefore derive that
    \begin{equation} \label{eq:vw:adw:varsumz}
    \Var \left( \sum_{v \in V} Z_v \right) 
    = \sum_{v \in V} \, \frac{b_{u,v}^2}{B_u^2} \Var(X_{u,v})
    \stackrel{(a)}{\leq} \sum_{v \in V} \, \frac{b_{u,v}^2}{B_u^2} \E[X_{u,v}]
    \stackrel{(b)}{\leq} \varepsilon \gamma \sum_{v \in V} \, \frac{b_{u,v} x_{u,v}}{B_u}
    \leq \varepsilon \gamma,
    \end{equation}
    where (a) follows from the fact that $X_{u,v}$ is Bernoulli and (b) from the small bids assumption.
    We can then show the claim because
    \begin{align*}
        \Pr\left[\,\text{$u$ has enough budget to pay $\max_{v \in V} b_{u, v}$ at the end}\,\right] 
        &\stackrel{(a)}{\geq} \Pr\left[ \sum_{v \in V} Y_v \leq  1 - \varepsilon \right]\\
        &\stackrel{(b)}{\geq} \Pr\left[ \sum_{v \in V} Z_v \leq 1 - \gamma - \varepsilon \right] \\
        &\stackrel{(c)}{\geq} 1 - \exp\left( -\frac{(1 - \gamma - \varepsilon)^2}{2 \varepsilon \gamma + \frac{2}{3} \varepsilon(1 - \gamma - \varepsilon)}\right) \\
        &\geq 1 - \exp\left( -\frac{(1 - \gamma - \varepsilon)^2}{2\varepsilon}\right),
    \end{align*}
    where (a) follows from the small bids assumption, (b) from the fact that $\E[\sum_{v \in V} Y_v] \leq \gamma$, and (c) from \cref{prop:bernstein} together with \cref{eq:vw:adw:varsumz} and that $1 - \gamma - \varepsilon > 0$.
    
    Substituting this lower bound into \cref{eq:intfrac}, we have
    \[
        \E[\ALGI] \geq \ALGF \cdot \gamma \left( 1 - \exp\left( -\frac{(1 - \gamma - \varepsilon)^2}{2\varepsilon}\right)\right).
    \]
    Hence, for sufficiently small $\varepsilon$, choosing $\gamma := 1 - \varepsilon - \sqrt{\varepsilon \ln (1 / \varepsilon)} \in (0, 1 - \varepsilon)$ gives
    \begin{align*}
    \frac{\E[\ALGI]}{\ALGF} &\geq  \left(1 - \varepsilon - \sqrt{\varepsilon \ln(1 / \varepsilon)} \right)\left(1 - \exp\left(-\frac{1}{2}\ln\frac1\varepsilon\right) \right)\\
    &= \left(1 - \varepsilon - \sqrt{\varepsilon \ln (1/\varepsilon)} \right)\left(1 - \sqrt{\varepsilon}\right)\\
    &\geq 1 - \varepsilon - \sqrt{\varepsilon} - \sqrt{\varepsilon \ln (1/\varepsilon)} \\
    &\geq 1 - 3\sqrt{\varepsilon \ln (1/\varepsilon)}.
    \end{align*} 
\end{proof}

\section{Unweighted matching with integral advice}
\label{sec:uwalg}

In this section, we introduce and analyze a new algorithm tailored to the unweighted setting with integral advice, which we call \PushAndWater{} (\PAW{}).
To motivate why we need a new algorithm, it is worth noting that our theoretical guarantees of \LAB{} are dominated by \Coinflip{} in the unweighted setting; see \cref{fig:intro:fig}.
This suggests that our previous analysis is not tight for unweighted instances.
However, since that analysis was independent of the vertex weights, we find it challenging to improve it for the unweighted setting, even when we are given integral advice.
As such, we propose \PAW{} for the setting of unweighted matching with integral advice.
In the following, we assume the advice is integral and represent it as a function $A: V \to U \cup \{\bot\}$, where $A(v)$ is the advised match for $v \in V$, and $A(v) = \bot$ indicates that $v$ is advised to remain unmatched.
Detailed pseudocode is given in \cref{sec:appendix-pseudocodes} and a full analysis is provided in the supplementary material.

\paragraph{Algorithm description.}
As before, we describe \PAW{} as a continuous-time process.
Define the level of an offline vertex $u \in U$ as the total amount of water it has received so far.
Upon arrival of online $v \in V$, with neighborhood $N(v)$ and advice $A(v)$, the algorithm proceeds in two phases:
\begin{enumerate}[label=\textbf{Phase \arabic*}]
\item \textbf{(Push)}: Push flow into $A(v)$ until its level reaches $\lambda$.
\item \textbf{(Waterfill)}: Distribute any remaining flow from $v$ across $N(v)$ via the standard waterfilling.
\end{enumerate}

\subsection{Primal-dual analysis}

We now analyze the performance of \PAW{}, showing \Cref{thm:uw:main} restated below.

\uwalgmain*

As in the analysis of \LAB{}, we use a primal-dual framework to characterize the robustness and consistency of \PAW{}.
To begin, we adapt \Cref{lem:pre:pda} to the unweighted setting:

\begin{lemma}[cf. \Cref{lem:pre:pda}] \label{lem:alg:1}
Let $x \in \R^E_+$ be the algorithm's output. For some $\rho \in [0, 1]$, if there exists $(\alpha, \beta) \in \R^U_+ \times \R^V_+$ satisfying 
\begin{enumerate}
    \item (reverse weak duality) $\sum_{e \in E} x_e \geq \sum_{u \in U} \alpha_u + \sum_{v \in V} \beta_v$, and
    \item (approximate dual feasibility) for every $(u,v) \in E$, $\alpha_u + \beta_v \geq \rho$,
\end{enumerate}
we have $\ALG \geq \rho \cdot \OPT$.
\end{lemma}

The dual variable construction differs from the vertex-weighted case and relies on a continuous and non-decreasing function $g : [0,1] \to [0,1]$ such that $g(1) = 1$.
We call such a function a \emph{splitting function}.

The dual variables $(\alpha, \beta)$ are initialized to zero, and are updated as follows.
When an online vertex $v$ sends an infinitesimal amount $dz$ of flow to a neighbor $u \in N(v)$ whose current level is $d_u$, split this $dz$ into $g(d_u) \;dz$ and $(1 - g(d_u)) \;dz$.
Then, we increase $\alpha_u$ by $g(d_u)\,dz$ and $\beta_v$ by $(1 - g(d_u))\,dz$.

Since $g(d_u) \in [0,1]$, both $\alpha$ and $\beta$ remain nonnegative.
Moreover, by construction, the reverse weak duality in \Cref{lem:alg:1} holds with equality:
every infinitesimal unit of flow is split \emph{exactly} into two values contributing to $\alpha_u$ and $\beta_v$, respectively.

\begin{lemma} \label{lem:uw:rwd}
    For any splitting function $g$, the constructed dual variables satisfy the reverse weak duality of \Cref{lem:alg:1} with equality.
\end{lemma}

Therefore, to analyze the robustness and consistency of the algorithm, it suffices to identify suitable splitting functions $g$ that ensure approximate dual feasibility, with the goal of maximizing the parameter $\rho$ in \Cref{lem:alg:1} for robustness and consistency, respectively.

We remark that the choice of splitting function $g$ does not affect the behavior of \PAW{}, in contrast to the vertex-weighted setting where the choice of penalty function $f$ directly influences \LAB{}'s execution.

\subsection{Robustness analysis}
In this subsection, we prove that \PAW{} is $r(\lambda)$-robust where $r(\lambda) := 1 - \left( 1 - \lambda + \frac{\lambda^2}{2} \right) \cdot e^{\lambda - 1}$.
To this end, we will prove the following lemma; observe that this lemma immediately implies the robustness result due to \Cref{lem:alg:1,lem:uw:rwd}.

\begin{lemma} \label{lem:uw:rob}
    There exists a splitting function $g$ such that the resulting dual solution $(\alpha, \beta)$ satisfies $\alpha_u + \beta_v \geq r(\lambda)$ for every $(u,v) \in E$.
\end{lemma}

We begin by giving a lower bound on $\beta_v$ for every $v \in V$.

\begin{lemma} \label{lem:uw:betav}
Fix an online vertex $v \in V$, and let $\ell$ be the minimum level of a neighbor of $v$ at the end of the iteration when $v$ arrives.
For any splitting function $g$, we have
\begin{equation} \label{eq:alg:bet}
    \beta_v \geq \begin{cases}
        \displaystyle \int_\ell^\lambda (1 - g(z)) \;dz + (1 - \lambda + \ell) (1 - g(\ell)), & \text{ if } \ell \in [0, \lambda), \\
        1 - g(\ell), & \text{ if } \ell \in [\lambda, 1].
    \end{cases}
\end{equation}
\end{lemma}
\begin{proof}
We break the proof into three cases.

\textbf{Case 1. $\ell < 1$ and $A(v) \in N(v)$.} Notice that $v$ is saturated in this case, i.e., $\sum_{u \in N(v)} x_{u, v} = 1$.
Let $d$ be the level of $A(v)$ at the beginning of $v$'s iteration.
Recall that, in Phase 1, the algorithm pushes $\tau := (\lambda - d)_+ = \lambda - \min(d, \lambda)$ units along $(A(v), v)$; in Phase 2, the algorithm distributes the remaining $1 - \tau$ units to its neighborhood $N(v)$ in the waterfilling manner. We thus deduce that
\begin{equation} \label{eq:alg:c1}
    \beta_v \geq \int_{\min(d, \lambda)}^\lambda (1 - g(z)) \;dz + (1 - \lambda + \min(d, \lambda)) (1 - g(\ell)).
\end{equation}
If $\ell < \lambda$, the right-hand side is further bounded by
\begin{align}
\text{(RHS of \cref{eq:alg:c1})}
& = \int_\ell^\lambda (1 - g(z)) \;dz + (1 - \lambda + \ell) (1 - g(\ell)) \nonumber \\
& \phantom{---} + \int_{\min(d, \lambda)}^\ell (1 - g(z)) \;dz - (\ell - \min(d, \lambda)) (1 - g(\ell))  \nonumber \\
& = \int_\ell^\lambda (1 - g(z)) \;dz + (1 - \lambda + \ell) (1 - g(\ell)) + \int_{\min(d, \lambda)}^\ell (g(\ell) - g(z)) \;dz \nonumber\\
& \geq \int_\ell^\lambda (1 - g(z)) \;dz + (1 - \lambda + \ell) (1 - g(\ell)), \label{eq:alg:lll}
\end{align}
where the inequality is satisfied due to the fact that, no matter whether $\ell \geq \min \{ d, \lambda \}$ or $\ell < \min \{ d, \lambda \}$, we have $\int_{\min(d, \lambda)}^\ell (g(\ell) - g(z)) dz\geq 0$ since $g$ is non-decreasing.

On the other hand, if $\ell \geq \lambda$, we have
\begin{align}
\text{(RHS of \cref{eq:alg:c1})}
& = 1 - g(\ell) + \int_{\min(d, \lambda)}^\lambda (1 - g(z)) \;dz - (\lambda - \min(d, \lambda)) (1 - g(\ell)) \nonumber\\
& = 1 - g(\ell) + \int_{\min(d, \lambda)}^\lambda (g(\ell) - g(z)) \;dz \nonumber\\
& \geq 1 - g(\ell), \label{eq:alg:lgl}
\end{align}
where the inequality is again due to that $g$ is non-decreasing.
This completes the proof for this case.

\textbf{Case 2. $\ell < 1$ and $A(v) = \bot$.}
Observe that the algorithm does nothing in Phase 1  while it distributes 1 unit to its neighbor $N(v)$. We thus have
\[
\beta_v \geq 1 - g(\ell),
\]
which is equivalent to \Cref{eq:alg:c1} with $d = \lambda$. Hence, \Cref{eq:alg:lll,eq:alg:lgl} also follow for this case.

\textbf{Case 3. $\ell = 1$.} 
In this case, we have the following trivial bound that
\begin{equation*}
    \beta_v \geq 0 = 1 - g(\ell),
\end{equation*}
where the equality holds since $g(1) = 1$.
\end{proof}

Let us define a splitting function $g_r$ for robustness as follows:
\[
g_r(z) := \begin{cases}
e^{\lambda - 1} (z + 1 - \lambda), & \forall z \in [0, \lambda), \\
e^{z - 1}, & \forall z \in [\lambda, 1].
\end{cases}
\]

Observe that $g_r$ is indeed a splitting function, i.e., $g_r$ is continuous and non-decreasing on $[0, 1]$ with $g(0) \geq 0$ and $g(1) = 1$.
Notice also that $g_r$ is differentiable.
Following is a technical lemma that will be used in the proof of \Cref{lem:uw:rob}.

\begin{lemma} \label{lem:uw:h}
    Let $h : [0, 1] \to [0, 1]$ be a function defined as
     \[
        h(\ell) := \begin{cases}
            \displaystyle \int_0^\ell g_r (z) dz + \int_\ell^\lambda (1 - g_r(z))dz + (1 - \lambda + \ell) (1 - g_r(\ell)), &\text{if $\ell \in [0,\lambda)$}, \\
            \displaystyle \int_0^\ell g_r(z) dz + ( 1 - g_r(\ell)), &\text{if $\ell \in [\lambda, 1]$.}
        \end{cases}
    \]
    Then, $h$ is a constant function of value $r(\lambda) = 1 - \left( 1 - \lambda + \frac{\lambda^2}{2} \right) \cdot e^{\lambda - 1}$.
\end{lemma}
\begin{proof}
    Note that the derivative of $h$ is
    \[
        h'(\ell) = \begin{cases}
             g_r (\ell) - g'_r (\ell)(1-\lambda+\ell), &\text{if $\ell \in (0,\lambda)$}, \\
             g_r (\ell) -g'_r (\ell), &\text{if $\ell \in (\lambda, 1)$.}
        \end{cases}
    \]
    By the definition of $g_r$, we have that $h'(\ell) = 0$ for all $\ell \in [0,1]$.
    (Indeed, $g_r$ was defined to make this true.)
    Therefore, $h$ is a constant function. The value of this constant is equal to 
    \[
        h(\lambda)
        = \int_0^\lambda g_r(z)dz + (1 - g_r(\lambda))
        = 1 -\left(1-\lambda + \frac{\lambda^2}{2}\right)e^{\lambda-1} = r(\lambda).
    \]
\end{proof}

We are now ready to prove \Cref{lem:uw:rob}.
\begin{proof} [Proof of \Cref{lem:uw:rob}]
    Fix an edge $(u, v) \in E$, and let $\ell$ denote the minimum level of a neighbor of $v$ at the end of the iteration of $v$.
    For any splitting function $g$, since $\alpha_u$ never decreases throughout the execution due to the definition of $g$, we have
    \[
        \alpha_u \geq \int_0^\ell g(z) \;dz.
    \]
    Together with \Cref{lem:uw:betav}, we can derive
    \[
        \alpha_u + \beta_v \geq \begin{cases}
            \displaystyle \int_0^\ell g (z) dz + \int_\ell^\lambda (1 - g(z))dz + (1 - \lambda + \ell) (1 - g(\ell)), &\text{if $\ell \in [0,\lambda)$}, \\
            \displaystyle \int_0^\ell g(z) dz + ( 1 - g(\ell)), &\text{if $\ell \in [\lambda, 1]$.}
        \end{cases}
    \]
    The proof of this lemma then immediately follows from \Cref{lem:uw:h} by choosing the splitting function $g := g_r$.
\end{proof}

\subsection{Consistency analysis}
We now show that \PAW{} is $c(\lambda)$-consistent where $c(\lambda) := 1 - (1 - \lambda) \cdot e^{\lambda - 1}$.
As in the previous robustness analysis, we will provide a good splitting function $g$ that satisfies the approximate dual feasibility with $c(\lambda)$.
However, contrary to the previous analysis, for the consistency, it suffices to have a \emph{relaxed} notion of the approximate dual feasibility, formally stated as follows:

\begin{lemma} [cf. \Cref{lem:alg:1}] \label{lem:uw:relpd}
    Let $x \in \R^E_+$ be the algorithm's output and $A \subseteq E$ be an integral advice. For some $\rho \in [0, 1]$, if there exists $(\alpha, \beta) \in \R^U_+ \times \R^V_+$ satisfying
    \begin{itemize}
        \item (reverse weak duality) $ \sum_{e \in E} x_e \geq \sum_{u \in U} \alpha_u + \sum_{v \in V} \beta_v,$ and
        \item (relaxed approximate dual feasibility) for every $(u, v) \in A$, $\alpha_u + \beta_v \geq \rho$,
    \end{itemize}
    we have $\ALG \geq \rho \cdot \ADVICE$.
\end{lemma}
\begin{proof}
    Observe that
    \begin{equation*}
        \ALG
        \geq \sum_{u \in U} \alpha_u + \sum_{v \in V} \beta_v
        \geq \sum_{(u, v) \in A} (\alpha_u + \beta_v)
        \geq \rho \cdot \ADVICE,
    \end{equation*}
    where the second inequality is due to the fact that $A$ is a matching.
\end{proof}

It thus suffices to prove the following lemma.
\begin{lemma} \label{lem:uw:con1}
    There exists a splitting function $g$ such that resulting dual solution $(\alpha, \beta)$ satisfies $\alpha_u + \beta_v \geq c(\lambda)$ for any $(u, v) \in A$.
\end{lemma}

We define a splitting function $g_c$ as follows:
\[
g_c(z) := \begin{cases}
    e^{\lambda - 1}, & \forall z \in [0, \lambda), \\
    e^{z - 1}, & \forall z \in [\lambda, 1].
\end{cases}
\]
It is easy to observe that $g_c$ satisfies the conditions of a splitting function.
We also remark that $g_c$ is differentiable on $(0,\lambda) \cup (\lambda, 1)$.
We need the following technical lemma in the proof of \Cref{lem:uw:con1}.

\begin{lemma} \label{lem:uw:gc}
    For any $\ell \in [0, 1]$, we have
    \[
    \int_0^{\max(\ell, \lambda)} g_c(z) dz + (1 - g_c(\ell)) \geq c(\lambda).
    \]
\end{lemma}
\begin{proof}
    Since $g_c$ is non-decreasing, for any $\ell \in [0, \lambda)$, the left-hand side is bounded from below by
    \[
        \int_0^{\max(\ell, \lambda)} g_c(z) dz + (1 - g_c(\ell))
        \geq \int_0^{\lambda} g_c(z) dz + (1 - g_c(\lambda)).
    \]
    Therefore, the infimum of the left-hand side is attained at $\ell \in [\lambda, 1]$.
    Let us denote $h(\ell) :=  \int_0^{\ell} g_c(z) dz + (1 - g_c(\ell))$, so that the left-hand side is equal to $h(\ell)$ on $[\lambda, 1]$. Note that by our choice of $g_c$, for any $\ell \in (\lambda, 1)$ we have
    \[
        h'(\ell) 
        = g(\ell) - g'(\ell)
        = e^{\ell-1} - e^{\ell-1} 
        = 0.
    \]
    Therefore, $h$ is a constant function on $[\lambda, 1]$. The value of this constant is equal to
    \[
        h(1)
        = \int_0^1g(z)dz + (1-g(1))
        = 1 - \left(1 - \lambda\right) \cdot e^{\lambda - 1}
        = c(\lambda).
    \]
\end{proof}

\begin{proof} [Proof of \Cref{lem:uw:con1}]
    Fix an edge $(u, v) \in A$, and let $\ell$ be the minimum level of a neighbor of $v$ at the end of the iteration of $v$.
    For any splitting function $g$, we have
    \[
        \beta_v \geq 1 - g(\ell)
    \]
    due to the monotonicity of $g$.
    On the other hand, since $u$ is matched by the advice $A$, the level of $u$ must be at least $\max (\ell, \lambda)$ due to Phase 1 of \PAW{} at this iteration, implying that
    \[
    \alpha_u \geq \int_0^{\max (\ell, \lambda)} g(z) dz.
    \]
    Therefore, choosing the splitting function $g := g_c$ immediately proves this lemma due to \Cref{lem:uw:gc}.
\end{proof}

\section{Upper bound on robustness-consistency tradeoff}
\label{sec:ub}

In this section, we present an upper bound result for the unweighted setting with integral advice.
This upper bound also applies to vertex-weighted matching and AdWords with fractional advice.
In \cref{sec:ub:desc}, we define two adversaries, $\robadv$ and $\conadv$, which target robustness and consistency, respectively, against any fractional matching algorithm $\tgtalg$.
Then, in \cref{sec:ub:lp}, we formulate a factor-revealing linear program (LP) that provides an upper bound on the best possible consistency value $c$ achievable against $\conadv$, subject to maintaining a robustness guarantee $r$ against $\robadv$.
The LP is constructed under a few assumptions about an algorithm's execution, which we later show in \cref{sec:ub:prop} to be without loss of generality.

\subsection{Description of adversaries} \label{sec:ub:desc}

We define two adversaries, $\robadv$ and $\conadv$, which target robustness and consistency, respectively, against any fractional matching algorithm $\tgtalg$.
For a given integer $n \in \mathbb{Z}_+$, both adversaries construct a bipartite instance with a set $U$ of $2n$ offline vertices and a set $V$ of $2n$ online vertices. See \Cref{fig:ub:desc} for an illustration of the upper bound instance.

The two adversaries behave identically during the first $n$ iterations, as follows:
In the first iteration ($t = 1$), they present the first online vertex $v_1$ to $\tgtalg$, with $v_1$ connected to all offline vertices in $U$. The advice $A(v_1)$ is chosen arbitrarily.
For each subsequent iteration $t = 2, \ldots, n$, the adversary presents online vertex $v_t$, which is adjacent to the neighbors $N(v_{t-1})$ of the previous vertex $v_{t-1}$, excluding two vertices: the previous advice $A(v_{t-1})$ and the offline vertex that has been filled the least so far by $\tgtalg$.

Starting from iteration $t = n+1$, the behaviors of the two adversaries diverge.
The robustness adversary $\robadv$ continues on the vertices advised to be matched so far as in the classical setting of online fractional bipartite matching without advice: each online vertex is adjacent to the same neighbors as the preceding one, except for the offline vertex that has been filled the least so far by $\tgtalg$.
In contrast, the consistency adversary $\conadv$ simply presents a matching to the offline vertices that were advised to be unmatched in the first $n$ iterations, allowing the algorithm to fully saturate them. 

The pseudocodes for $\robadv$ and $\conadv$ are given in \Cref{alg:ub:rob,alg:ub:con}, respectively.
Note that in both $\robadv$ and $\conadv$, the size of the maximum matching in hindsight is $2n$.
For simplicity of presentation, we also allow the adversaries to reorder the indices of offline vertices based on the behavior of $\tgtalg$ over time.

\def\offwidth{0.6}
\def\onwidth{0.6}
\def\ongap{3}

\begin{figure}
    \centering
    \begin{tikzpicture}[
        offstyle/.style={
            draw= black,
            minimum width= \offwidth cm,
            minimum height= \offwidth cm
        },
        onstyle/.style={
            circle,
            draw= black,
            minimum width= \onwidth cm,
            minimum height= \onwidth cm
        },
        advicestyle/.style={
            draw= red,
            thick,
            ->
        }
    ]

    \foreach \i in {1, 2, ..., 6}
        \draw (0, -\i) node[offstyle] (u\i) {$u_{\i}$};

    \foreach \j in {1, 2, ..., 3}
        \draw (-\ongap, -\j) node[onstyle] (v\j) {$v_{\j}$};
    \draw [thick, decorate, decoration = {brace, mirror}] (-\ongap-0.5, -0.75) --  (-\ongap-0.5, -3.25);
    \draw[anchor=east] (-\ongap-0.7, -2) node (cdesc) {Common Phase};
    
    \foreach \i in {1, 2, ..., 6}
        \draw (v1) -- (u\i);
    \foreach \i in {2, 3, ..., 5}
        \draw (v2) -- (u\i);
    \foreach \i in {3, 4}
        \draw (v3) -- (u\i);
    \draw[advicestyle] (v1) -- (u1);
    \draw[advicestyle] (v2) -- (u2);
    \draw[advicestyle] (v3) -- (u3);
    
    \foreach \j in {4, 5, 6}
        \draw (\ongap, -\j+3) node[onstyle] (vr\j) {$v_{\j}$};
    \draw [thick, decorate, decoration = {brace}] (\ongap+0.5, -0.75) --  (\ongap+0.5, -3.25);
    \draw[anchor=west] (\ongap+0.7, -2) node (cdesc) {Robustness Phase};
        
    \foreach \i in {1, 2, 3}
        \draw (vr4) -- (u\i);
    \foreach \i in {2, 3}
        \draw (vr5) -- (u\i);
    \foreach \i in {3}
        \draw (vr6) -- (u\i);

    \foreach \j in {4, 5, 6}
        \draw (\ongap, -\j) node[onstyle] (vc\j) {$v_{\j}$};
    \draw [thick, decorate, decoration = {brace}] (\ongap+0.5, -3.75) --  (\ongap+0.5, -6.25);
    \draw[anchor=west] (\ongap+0.7, -5) node (cdesc) {Consistency Phase};
    
    \foreach \i in {4}
        \draw (vc4) -- (u\i);
    \foreach \i in {5}
        \draw (vc5) -- (u\i);
    \foreach \i in {6}
        \draw (vc6) -- (u\i);
    \draw[advicestyle] (vc4) -- (u4);
    \draw[advicestyle] (vc5) -- (u5);
    \draw[advicestyle] (vc6) -- (u6);    
    
    \end{tikzpicture}
    \caption{An illustration of the hardness construction. The instance begins with the common phase, which is the same in both adversaries. After the common phase, the instance can proceed in one of two ways, designed to be hard for robustness or consistency respectively.
The robustness phase consists of an upper-triangular graph on the offline vertices $\{u_1, u_2, \ldots, u_n\}$, while the consistency phase forms a perfect matching to the offline vertices $\{u_{n+1}, u_{n+2}, \ldots, u_{2n}\}$.}
    \label{fig:ub:desc}
\end{figure}
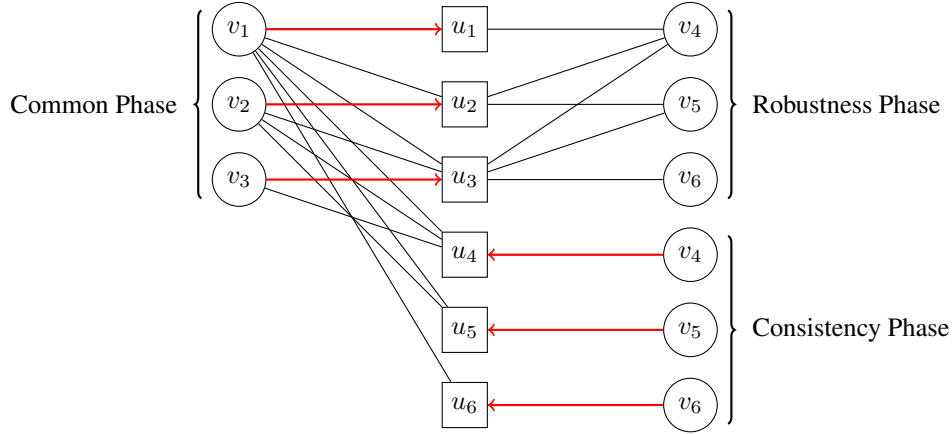

\begin{algorithm}
\caption{Adversary $\robadv$ for robustness}
\label{alg:ub:rob}
    \KwIn{A fractional matching algorithm $\tgtalg$ and $n \in \Z_{\geq 1}$}
    
    Feed $U := \{u_1, \ldots, u_{2n} \}$ to $\tgtalg$ \;
    $\lev{0}_u \gets 0$ for every $u \in U$ \hfill \CommentSty{// $\lev{t}_u$ traces the level of $u$ at the end of iteration~$t$} \;
    \CommentSty{// \textbf{Common phase}}\;
    \For{$t \gets 1, \ldots, n$}{
        Feed $v_t$, $N(v_t) := \{u_t, u_{t+1}, \ldots, u_{2n-t+1}\}$, and $A(v_t) := u_t$ to $\tgtalg$ \;
        Let $\{ x_{u, v_t} \}_{u \in N(v_t)}$ be the output of $\tgtalg$ \;
        \For{$u \in U$} {
            \eIf{$u \in N(v_t)$}{
                $\lev{t}_u \gets \lev{t-1}_u + x_{u, v_t}$ \;
            }{
                $\lev{t}_u \gets \lev{t-1}_u$ \;
            }
        }
        Reorder $\{ t+1, \ldots, 2n-t+1 \}$ so that $\lev{t}_{u_{t+1}} \geq \cdots \geq \lev{t}_{u_{2n - t + 1}}$ \label{line:ub:robro1} \;
    }
    
    \CommentSty{// \textbf{Robustness phase}}\;
    Reorder $\{1, \ldots, n\}$ so that $\lev{n}_{u_1} \leq \cdots \lev{n}_{u_n}$ \label{line:ub:robro2} \;
    \For{$t \gets n+1, \ldots, 2n$}{
        Feed $v_t$, $N(v_t) := \{u_{t-n}, \ldots, u_{n}\}$, and $A(v_t) := \bot$ to $\tgtalg$ \;
        Let $\{ x_{u, v_t} \}_{u \in N(v_t)}$ be the output of $\tgtalg$ \;
        \For{$u \in U$} {
            \eIf{$u \in N(v_t)$}{
                $\lev{t}_u \gets \lev{t-1}_u + x_{u, v_t}$ \;
            }{
                $\lev{t}_u \gets \lev{t-1}_u$ \;
            }
        }
        Reorder $\{ t-n, \ldots, n \}$ so that $\lev{t}_{u_{t-n}} \leq \cdots \leq \lev{t}_{u_n}$ \label{line:ub:robro3}  \;
    }
\end{algorithm}

\begin{algorithm}
\caption{Adversary $\conadv$ for consistency}
\label{alg:ub:con}
    \KwIn{A fractional matching algorithm $\tgtalg$ and $n \in \Z_{\geq 1}$}
    
    Feed $U := \{u_1, \ldots, u_{2n} \}$ to $\tgtalg$ \;
    $\lev{0}_u \gets 0$ for every $u \in U$ \hfill \CommentSty{// $\lev{t}_u$ traces the level of $u$ at the end of iteration~$t$} \;
    \CommentSty{// \textbf{Common phase}}\;
    \For{$t \gets 1, \ldots, n$}{
        Feed $v_t$, $N(v_t) := \{u_t, u_{t+1}, \ldots, u_{2n-t+1}\}$, and $A(v_t) := u_t$ to $\tgtalg$ \;
        Let $\{ x_{u, v_t} \}_{u \in N(v_t)}$ be the output of $\tgtalg$ \;
        \For{$u \in U$} {
            \eIf{$u \in N(v_t)$}{
                $\lev{t}_u \gets \lev{t-1}_u + x_{u, v_t}$ \;
            }{
                $\lev{t}_u \gets \lev{t-1}_u$ \;
            }
        }
        Reorder $\{ t+1, \ldots, 2n-t+1 \}$ so that $\lev{t}_{u_{t+1}} \geq \cdots \geq \lev{t}_{u_{2n - t + 1}}$ \label{line:ub:conro1}  \;
    }
    
    \CommentSty{// \textbf{Consistency phase}}\;
    \For{$t \gets n+1, \ldots, 2n$}{
        Feed $v_t$, $N(v_t) := \{u_t\}$, and $A(v_t) := u_t$ to $\tgtalg$ \;
        Let $\{ x_{u, v_t} \}_{u \in N(v_t)}$ be the output of $\tgtalg$ \;
        \For{$u \in U$} {
            \eIf{$u \in N(v_t)$}{
                $\lev{t}_u \gets \lev{t-1}_u + x_{u, v_t}$ \;
            }{
                $\lev{t}_u \gets \lev{t-1}_u$ \;
            }
        }
    }
\end{algorithm}

\subsection{Factor-Revealing LP} \label{sec:ub:lp}

We now formulate a factor-revealing LP that upper bounds the consistency ratio $c$ of any algorithm $\tgtalg$ against $\conadv$ while ensuring the algorithm is $r$-robust against $\robadv$, for any $r \in [\nicefrac{1}{2}, 1-\nicefrac{1}{e}]$.
To this end, we assume that $\tgtalg$ satisfies the following conditions:

\begin{enumerate}
    \item (Greediness) $\tgtalg$ saturates each online vertex unless its neighbors are all saturated.
    \item (Monotonicity) $\tgtalg$ never makes $\robadv$ and $\conadv$ reorder the indices of $U$ (e.g., Lines~\ref{line:ub:robro1}, \ref{line:ub:robro2}, and~\ref{line:ub:robro3} in \Cref{alg:ub:rob} and Line~\ref{line:ub:conro1} in \Cref{alg:ub:con}).
    \item (Uniformity) In the common phase, $\tgtalg$ pushes the same amount to the neighbors except the advised offline vertex at each iteration. In other words, for any $t \in \{ 1, \ldots, n \}$, the value $x_{u, v_t}$ is the same for all $u \in N(v_t)$ with $u \neq A(v_t)$.
\end{enumerate}

In \cref{sec:ub:prop}, we argue that $\tgtalg$ satisfies these conditions without loss of generality. 

\begin{lemma} \label{lem:ub:pareto}
    For any algorithm $\tgtalg$, there exists an algorithm $\modalg$ satisfying the above conditions (greediness, monotonicity, and uniformity), such that under both $\robadv$ and $\conadv$, the value of the matching returned by $\modalg$ is at least that of $\tgtalg$, 
\end{lemma}

We now write the factor-revealing LP assuming these conditions are satisfied. 
For each iteration $t = 1, \ldots, n$ in the common phase, let
\[
    x_t := x_{u_t, v_t} \geq 0
\]
denote the amount of water pushed by $\tgtalg$ to the suggested vertex $A(v_t) = u_t$, and let 
\[
    \xbar_t := x_{u_{t+1}, v_t} = \ldots = x_{u_{2n - t + 1}, v_t} \geq 0
\]
denote the amount pushed uniformly to the other neighbors $ N(v_t) \setminus \{u_t\}$.
Since $\tgtalg$ should output a fractional matching, we have
\[
    x_t + (2n-2t+1) \cdot \xbar_t \leq 1.
\]
Let $d_t := \lev{t}_{u_t}$ and $\dbar_t := \lev{t}_{u_{2n-t+1}}$ denote the levels of $u_t$ and $u_{2n-t+1}$, respectively, at the end of iteration~$t$.
Note that these levels remain unchanged until the end of the common phase by construction of our adversaries:
both adversaries will never reveal edges adjacent to $u_t$ and $u_{2n-t+1}$ from iterations $t+1$ to $n$.
Hence we have the constraints
\[
d_t = \sum_{i = 1}^{t-1} \xbar_i + x_t \leq 1 \text{ and } \dbar_t = \sum_{i = 1}^{t} \xbar_i \leq 1,
\]
from the capacity constraints of $u_t$ and $u_{2n-t+1}$ respectively.
Since $\tgtalg$ is monotone, we also have
\[
d_1 \leq d_2 \leq \ldots \leq d_n
\]
due to Line~\ref{line:ub:robro2} of $\robadv$.

Now consider the robustness adversary $\robadv$. For each $t \in \{1, \ldots, n\}$ and $i \in \{t, \ldots, n\}$, let $y_{i, t} := x_{u_i, v_{n+t}}$ be the amount pushed by $\tgtalg$ from $v_{n+t}$ to $u_i$.
Due to the degree constraint on $v_{n+t}$, we have
\[
\sum_{i = t}^n y_{i, t} \leq 1.
\]
Let $\roblev{t}_i := \lev{n+t}_{u_i}$ be the level of $u_i$ at the end of round~$n+t$ against $\robadv$.
Note that we have
\[
\roblev{t}_i = d_i + \sum_{s=1}^t y_{i, s} \leq 1,
\]
where the inequality is due to the degree constraint.
Since $\tgtalg$ is monotone, we have
\[
\roblev{t}_t \leq \roblev{t}_{t+1} \leq \ldots \leq \roblev{t}_{n}
\]
due to Line~\ref{line:ub:robro3} of $\robadv$.

Notice that, in order for $\tgtalg$ to have the robustness ratio of $r$, it should output a solution of size at least $2n \cdot r$ against $\robadv$, implying that we have
\[
\sum_{t = 1}^n (d_t + \dbar_t) + \sum_{t = 1}^n \sum_{i = t}^n y_{i, t} \geq 2 n r.
\]
Under this situation, the solution quality output by $\tgtalg$ against $\conadv$ bounds from above the consistency ratio $c$ of $\tgtalg$. We thus have
\[
\sum_{t = 1}^n d_t + n \geq 2 n c.
\]

Given a robustness $r$, our objective is to find the best-possible consistency $c$ while satisfying all of the above constraints:
\begin{align*}
    \text{maximize } & c \\
    \text{subject to } & x_t + (2n-2t+1) \cdot \xbar_t \leq 1, & \forall t \in \{1, \ldots, n\}, \\
    & \textstyle d_t = \sum_{i = 1}^{t-1} \xbar_i + x_t, &  \forall t \in \{1, \ldots, n\}, \\
    & \textstyle \dbar_t = \sum_{i = 1}^{t} \xbar_i, &  \forall t \in \{1, \ldots, n\}, \\
    & d_t \leq d_{t+1}, & \forall t \in \{1, \ldots, n-1\}, \\
    & \textstyle \sum_{i = t}^n y_{i, t} \leq 1, & \forall t \in \{1, \ldots, n\}, \\
    & \textstyle \roblev{t}_i = d_i + \sum_{s=1}^t y_{i, s}, &  \forall t \in \{1, \ldots, n\},\, \forall i \in \{t, \ldots, n\}, \\
    & \roblev{t}_i \leq \roblev{t}_{i+1}, & \forall t \in \{1, \ldots, n\}, \forall i \in \{t, \ldots, n-1\}, \\
    & \textstyle \sum_{t = 1}^n (d_t + \dbar_t) + \sum_{t = 1}^n \sum_{i = t}^n y_{i, t} \geq 2 n r, \\    
    & \textstyle \sum_{t = 1}^n d_t + n \geq 2 n c, \\
    & 0 \leq x_t, \xbar_t, d_t, \dbar_t \leq 1, & \forall t \in \{1, \ldots, n\}, \\
    & 0 \leq y_{i, t}, \roblev{t}_i \leq 1, & \forall t \in \{1, \ldots, n\},\, \forall i \in \{t, \ldots, n\}.
\end{align*}

We implemented this factor-revealing LP with PuLP\footnote{\url{https://github.com/coin-or/pulp}} and computationally solved this LP using Gurobi 12.0.2~\cite{gurobi} with $n =1000$ and a finite set of values for robustness $r \in [0.5, 1-\nicefrac{1}{e}]$.
See \Cref{tab:ub} and \Cref{fig:intro:fig} for the computational results.

\begin{table}[h]
\centering
\caption{Upper bound on the robustness-consistency tradeoff with $n=1000$.}
\label{tab:ub}
\begin{tabular}{c|ccccccccc}
$r$ & 0.500  & 0.525  & 0.550  & 0.575  & 0.600  & 0.625  & $1-\nicefrac{1}{e}$  \\ \hline
$c$ & 1.000 & 0.974 & 0.944  & 0.908 & 0.862 & 0.788 & 0.731 
\end{tabular}
\end{table}

\subsection{Proof of \texorpdfstring{\Cref{lem:ub:pareto}}{Lemma~\ref{lem:ub:pareto}}} \label{sec:ub:prop}

This subsection is devoted to the proof of \Cref{lem:ub:pareto}.
In what follows, for any algorithm $\tgtalg$, let us denote by $\tgtalg(\robadv)$ and $\tgtalg(\conadv)$ the size of the output of $\tgtalg$ against $\robadv$ and $\conadv$, respectively.

Recall the conditions that we want to prove without loss of generality for an algorithm $\tgtalg$ to be Pareto-efficient:

\begin{enumerate}
    \item (Greediness) $\tgtalg$ saturates each online vertex unless its neighbors are all saturated.
    \item (Monotonicity) $\tgtalg$ never makes $\robadv$ and $\conadv$ reorder the indices of $U$.
    \item (Uniformity) In the common phase, $\tgtalg$ pushes the same amount to the neighbors except the advised offline vertex at each iteration.
\end{enumerate}

For greediness, it is folklore that this condition is without loss of generality for online bipartite matching (see, e.g., \cite{karp1990optimal}).

Let us now consider monotonicity. 
Observe that there are three places (Lines~\ref{line:ub:robro1}, \ref{line:ub:robro2}, and~\ref{line:ub:robro3}) in \Cref{alg:ub:rob} and one place (Line~\ref{line:ub:conro1}) in \Cref{alg:ub:con} where the adversaries may reorder the indices of $U$.
Let us define the following.

\begin{itemize}
    \item For every $t \in \{ 1, \ldots, n-1 \}$, we say an algorithm is \emph{monotone at iteration $t$} if $\robadv$ (and $\conadv$) does not reorder the indices at Line~\ref{line:ub:robro1} (and Line~\ref{line:ub:conro1}, respectively) in iteration~$t$ of the common phase, i.e., for $\lev{t-1}_{u_{t+1}} \geq \cdots \geq \lev{t-1}_{u_{2n - t + 1}}$, we also have $\lev{t}_{u_{t+1}} \geq \cdots \geq \lev{t}_{u_{2n - t + 1}}$.
    Note that, in iteration~$n$, the adversary does not reorder the indices at that line since it has only one index $ n+1 $.
    \item We say an algorithm is \emph{monotone at iteration $n$} if $\robadv$ does not reorder the indices at Line~\ref{line:ub:robro2}.
    \item For every $t \in \{n+1, \ldots, 2n - 1\}$, we say an algorithm is \emph{monotone at iteration $t$} if $\robadv$ does not reorder the indices at Line~\ref{line:ub:robro3} in iteration~$t$ of the robustness phase, i.e., for  $\lev{t-1}_{u_{t-n}} \leq \cdots \leq \lev{t-1}_{u_n}$, we also have $\lev{t}_{u_{t-n}} \leq \cdots \leq \lev{t}_{u_n}$.
    Note that, in iteration~$2n$, the adversary has only $\{n\}$ at that line.
\end{itemize}

The following lemma implies that the monotonicity of $\tgtalg$ in the common phase can be assumed without loss of generality. 

\begin{lemma} \label{lem:ub:monotone}
    For any algorithm $\tgtalg$, there exists an algorithm $\modalg$ which is monotone at every iteration $t \in \{1, \ldots, n-1\}$ while satisfying $\modalg(\robadv) = \tgtalg(\robadv)$ and $\modalg(\conadv) = \tgtalg(\conadv)$.
\end{lemma}
\begin{proof}
    We inductively prove the lemma.
    Let $t$ be the first iteration where $\tgtalg$ reorders the indices at Line~\ref{line:ub:robro1} of $\robadv$ (and Line~\ref{line:ub:conro1} of $\conadv$).
    Let $\sigma : \{t+1, \ldots, 2n-t+1\} \to \{t+1, \ldots, 2n-t+1\}$ be the permutation such that $\lev{t}_{u_{\sigma(t+1)}} \geq \ldots \geq \lev{t}_{u_{\sigma(2n-t+1)}}$ right before the adversary reorders the indices. 
    Since we have $\lev{t-1}_{u_{t+1}} \geq \ldots \geq \lev{t-1}_{u_{2n-t+1}}$ due to the induction hypothesis from the previous iteration, we can derive that $\lev{t}_{u_{\sigma(i)}} \geq \lev{t-1}_{u_i}$ for every $i \in \{t+1, \ldots, 2n-t+1\}$.
    Let us now consider another algorithm $\modalg$, whose output is denoted by $\xbar$, such that, up to iteration~$t-1$, $\modalg$ behaves the same as $\tgtalg$, but in iteration~$t$, $\modalg$ pushes $\xbar_{u_i, v_t}$ towards $u_i \in N(v_t)$ for every $i \in \{t + 1, \ldots, 2n-t+1\}$, where
    \[
        \xbar_{u_i, v_t} := \begin{cases}
            \lev{t}_{u_t} - \lev{t-1}_{u_t}, & \text{ if } i=t, \\
            \lev{t}_{u_{\sigma(i)}} - \lev{t-1}_{u_i}, & \text{ otherwise.}
        \end{cases}
    \]
    Observe that the level $\modlev{t}_{u_i}$ of $u_i \in N(v_t)$ at this moment in the execution of $\modalg$ is
    \[
        \modlev{t}_{u_i} := \begin{cases}
            \lev{t}_{u_t}, & \text{ if } i=t, \\
            \lev{t}_{u_{\sigma(i)}}, & \text{ otherwise.}
        \end{cases}
    \]
    Therefore, $\modalg$ satisfies $\modlev{t}_{u_{t+1}} \geq \cdots \geq \modlev{t}_{u_{2n-t+1}}$, implying that the adversary would not reorder the indices of $N(v_{t}) \setminus \{u_{t}\}$ against $\modalg$.
    Moreover, $\tgtalg$ and $\modalg$ have the same configuration of levels of $N(v_t)$ at the end of iteration $t$, and hence, by letting $\modalg$ behave the same as $\tgtalg$ until termination, we can see that both $\modalg$ and $\tgtalg$ eventually return solutions of the same size against the adversary.
\end{proof}

By a similar argument, we can also assume without loss of generality the monotonicity of $\tgtalg$ against $\robadv$ in each iteration $t \in \{n+1, \ldots, 2n-1\}$ of the robustness phase.
We omit the proof of the next lemma.

\begin{lemma}
    For any algorithm $\tgtalg$, there exists an algorithm $\modalg$ which is monotone at every iteration~$t \in \{n+1, \ldots, 2n-1\}$ while satisfying $\modalg(\robadv) = \tgtalg(\robadv)$.
\end{lemma}

What remains to show is generality of monotonicity at iteration~$n$ and uniformity.
We first prove that it is no loss of generality to assume the uniformity of the algorithm, and then show that, given the algorithm is uniform, we can further assume without loss of generality that the algorithm is monotone at iteration~$n$.

Due to \Cref{lem:ub:monotone}, let us from now consider any algorithm $\tgtalg$ that is greedy and monotone in the common phase.
Due to the monotonicity, let us slightly abuse the notation and write $\lev{t}_i := \lev{t}_{u_i}$ for any $t \in \{0, \ldots, n\}$ and $i \in \{1, \ldots, 2n\}$ for simplicity.

In what follows, we will prove the generality by inductively modifying any algorithm $\tgtalg$ into another algorithm $\modalg$.
To this end, let us first define some notation related to $\modalg$.
Let us denote by $\modx$ the solution of $\modalg$.
For each $t \in \{0, \ldots, n\}$ and $i \in \{1, \ldots, 2n\}$, let us denote by $\modlev{t}_i$ the level of $u_i \in L$ at the end of iteration $t$ in the execution of $\modalg$ against the adversaries in the common phase (i.e., the value of $\lev{t}_{u_i}$ in \Cref{alg:ub:rob,alg:ub:con}).

We newly define the notation only for the common phase due to the following lemma.
\begin{lemma} \label{lem:ub:suffcond}
    Given any greedy algorithm $\tgtalg$, suppose we construct another greedy algorithm $\modalg$ and define the execution of $\modalg$ in the common phase to satisfy the following conditions: 
    \begin{enumerate}[label = (\Alph*)]
        \item \label{cond:ub:suffcond1} the sum of levels of $\{u_1, \ldots, u_n \}$ is the same in both $\modalg$ and $\tgtalg$, i.e., $\sum_{i = 1}^n \modlev{n}_i = \sum_{i = 1}^n \lev{n}_i$; and
        \item \label{cond:ub:suffcond2} in each iteration $t \in \{1, \ldots, n\}$, $\modalg$ pushes as much amount of water as $\tgtalg$, i.e., $\sum_{u \in N(v_t)} \xbar_{u, v_t} \geq \sum_{u \in N(v_t)} x_{u, v_t}$.
    \end{enumerate}
    We can then complete the execution of $\modalg$ in the following iterations to satisfy $\modalg(\robadv) \geq \tgtalg(\robadv)$ and $\modalg(\conadv) = \tgtalg(\conadv)$.
\end{lemma}
\begin{proof}
    Due to Condition~\ref{cond:ub:suffcond2}, the total amount pushed by $\modalg$ throughout the common phase is no less than that by $\tgtalg$, implying that we have $\sum_{i = 1}^{2n} \modlev{n}_i \geq \sum_{i = 1}^{2n} \lev{n}_i$.
    From Condition~\ref{cond:ub:suffcond1}, we can also derive $\sum_{i = n+1}^{2n} \modlev{n}_i \geq \sum_{i = n+1}^{2n} \lev{n}_i$.
    Therefore, against $\robadv$, by letting $\modalg$ behave the same as $\tgtalg$ in the robustness phase, we can obtain
    \[
        \modalg(\robadv)
        = \sum_{i = 1}^n \lev{2n}_{u_i} + \sum_{i=n+1}^{2n} \modlev{n}_i
        \geq \sum_{i = 1}^n \lev{2n}_{u_i} + \sum_{i=n+1}^{2n} \lev{n}_i
        = \tgtalg(\robadv).
    \]
    On the other hand, since any greedy algorithm would saturate $\{u_{n+1}, \ldots, u_{2n}\}$ against $\conadv$, it is easy to define the execution of $\modalg$ in the consistency phase to satisfy
    \[
        \modalg(\conadv) = \sum_{i = 1}^n \modlev{n}_{i} + n
        = \sum_{i = 1}^n \lev{n}_{i} + n = \tgtalg(\conadv). 
    \]
\end{proof}

We are now ready to prove that we can assume without loss of generality that the algorithm is uniform.
\begin{lemma} \label{lem:ub:uniform}
    Given any algorithm $\tgtalg$ that is greedy and monotone in the common phase, there exists an algorithm $\modalg$ that is uniform as well while satisfying $\modalg(\robadv) \geq \tgtalg(\robadv)$ and $\modalg(\conadv) = \tgtalg(\conadv)$.
\end{lemma}

\begin{proof}
    For $t \in \{1, \ldots, n\}$, we say that an algorithm is \emph{uniform at iteration $t$} if the algorithm pushes water uniformly towards $N(v_t) \setminus \{ u_t \}$.
    Recall that the amount pushed towards $A(v_t) = u_t$ does not affect the uniformity of the algorithm at iteration $t$.
    
    We prove the lemma by inductively showing that, for some $\tstar \in \{1, \ldots, n\}$, if $\tgtalg$ is monotone up to iteration $\tstar-1$, we can construct $\modalg$ that is monotone up to iteration $\tstar$ while the execution of $\modalg$ in the common phase satisfies the conditions of \Cref{lem:ub:suffcond}.
    In particular, instead of Condition~\ref{cond:ub:suffcond1}, we will consider a stronger condition:
    \begin{enumerate} [label = (\Alph*$'$)]
        \item \label{cond:ub:suffcond1p} in each iteration~$t \in \{1, \ldots, n\}$, the level of $u_t$ is equal in both $\modalg$ and $\tgtalg$, i.e., $\modlev{t}_t = \lev{t}_t$.
    \end{enumerate}
    We assume without loss of generality that $\tgtalg$ is \emph{non-uniform} at iteration~$\tstar$ for the first time; otherwise, $\modalg = \tgtalg$ would immediately satisfy the conditions.
    Observe that, under this assumption, $\tgtalg$ saturates $v_{\tstar}$ in this iteration.
    
    Let us now describe the execution of $\modalg$ in the common phase.
    For iterations~$t \in \{1, \ldots, \tstar-1\}$, $\modalg$ behaves the same as $\tgtalg$.
    In iteration~$\tstar$, however, since we now want $\modalg$ to be uniform at this iteration while guaranteeing the conditions of \Cref{lem:ub:suffcond}, $\modalg$ pushes the same amount $x_{u_{\tstar}, v_{\tstar}}$ as $\tgtalg$ towards its advice $A(v_{\tstar}) = u_{\tstar}$ while it distributes the remaining $1 - x_{u_{\tstar}, v_{\tstar}}$ units uniformly to the other neighbors $N(v_{\tstar}) \setminus \{u_{\tstar}\}$, i.e., we have, for any $i \in \{\tstar, \ldots, 2n - \tstar + 1\}$,
    \[
        \xbar_{u_i, v_{\tstar}} := \begin{cases}
            x_{u_{\tstar}, v_{\tstar}}, & \text{ if } i = \tstar, \\
            \frac{1 - x_{u_{\tstar}, v_{\tstar}}}{2n - 2 \tstar + 1}, & \text{ otherwise.}
        \end{cases}
    \]
    In the subsequent iterations~$t \in \{\tstar + 1, \ldots, n\}$, $\modalg$ iterates $i \in \{t, \ldots, 2n-t+1\}$ in this order and pushes water through $(u_i, v_t)$ until the level of $u_i$ reaches $\lev{t}_i$ or $v_t$ gets saturated.
    We will later argue that, for any neighbor $u_i \in N(v_t)$ that is filled a positive amount by $\modalg$ in this iteration, it is guaranteed to have $\lev{t}_i \geq \modlev{t-1}_i$, implying that $\modalg$ is well-defined.

    It is trivial to see that $\modalg$ is still greedy and monotone in the common phase.
    Note also that it is now uniform up to iteration~$\tstar$.    
    It remains to prove that the execution of $\modalg$ in the common phase is indeed well-defined while satisfying Conditions~\ref{cond:ub:suffcond1p} and~\ref{cond:ub:suffcond2} as well.
    It is true up to iteration~$\tstar - 1$ since the executions of $\tgtalg$ and $\modalg$ are identical.
    Since every offline vertex in $\{u_1, \ldots, u_{\tstar - 1}\} \cup \{ u_{2n - \tstar + 2}, \ldots, u_{2n} \}$ is not adjacent with any online vertex arriving after iteration $\tstar - 1$ in the common phase, we have
    \begin{equation} \label{eq:ub:prevlev}
        \modlev{\tstar - 1}_i = \cdots = \modlev{n}_i = \lev{\tstar - 1}_i = \cdots = \lev{n}_i
    \end{equation}
    for any $i \in \{1, \ldots, \tstar - 1\} \cup \{2n - \tstar + 2, \ldots, 2n\}$.

    For the remaining iterations in the common phase, recall that $\tgtalg$ saturates $v_{\tstar}$.
    Suppose $\tgtalg$ does not saturate $v_n$ in the last iteration $n$ of the phase.
    In this case, since $\tgtalg$ is greedy while $N(v_{t+1}) \subsetneq N(v_t)$ for every $t \in \{1, \ldots, n\}$ in the common phase, there must exist $\tsat \in \{\tstar, \ldots, n-1\}$ such that $\tgtalg$ saturates $v_t$ in every iteration $t \leq \tsat$ whereas all offline vertices in $N(v_{\tsat +1})$ become saturated at iteration $\tsat + 1$; for any following iterations~$t \geq \tsat + 2$, $\tgtalg$ pushes nothing, i.e., $\sum_{u \in N(v_t)} x_{u, v_t} = 0$.
    For the other case when $\tgtalg$ saturates $v_n$ in iteration $n$, let us say $\tsat := n$.

    Consider iteration~$\tstar$.
    It is easy to see that the execution of $\modalg$ in this iteration is well-defined and satisfies Conditions~\ref{cond:ub:suffcond1p} and~\ref{cond:ub:suffcond2}.
    Let $\modavg = \modlev{\tstar}_{\tstar+1} = \cdots = \modlev{\tstar}_{2n - \tstar + 1}$ denote the uniform level of the neighbors $N(v_{\tstar}) \setminus \{u_{\tstar}\}$ other than the advice $A(v_{\tstar}) = u_{\tstar}$ in $\modalg$.
    Since $\lev{\tstar}_{\tstar + 1} > \modavg$ and $\lev{\tstar}_{2n - \tstar + 1} < \modavg$, there always exists $\pat{\tstar}$ and $\qat{\tstar}$ such that
    \begin{itemize}
        \item $\tstar \leq \pat{\tstar} < \qat{\tstar} \leq 2n-\tstar$;
        \item for every $i \in \{\tstar, \ldots, \pat{\tstar}\}$, the levels of $u_i$ in $\modalg$ and $\tgtalg$ are the same, i.e., $\modlev{\tstar}_i = \lev{\tstar}_i$;
        \item for every $i \in \{\pat{\tstar} + 1, \ldots, \qat{\tstar}\}$, the level of $u_i$ in $\modalg$ is strictly less than that in $\tgtalg$, i.e., $ \modlev{\tstar}_i = \modavg < \lev{\tstar}_i$; and
        \item for every $i \in \{\qat{\tstar} + 1, \ldots, 2n - \tstar + 1\}$, the level of $u_i$ in $\modalg$ is no less than that in $\tgtalg$, i.e.,  $\modlev{\tstar}_i = \modavg \geq \lev{\tstar}_i$.
    \end{itemize}
    In fact, for each subsequent iteration~$t \in \{\tstar + 1, \ldots, \tsat \}$ where $\tgtalg$ saturates $v_t$, we inductively show that there exists $\pat{t}$ and $\qat{t}$ that satisfy the following properties:
    \begin{enumerate}[label=(\roman*)]
        \item \label{inv:ub:prop:u1} $t \leq \pat{t} < \qat{t} \leq 2n-\tstar$;
        \item \label{inv:ub:prop:u2} for every $i \in \{\tstar, \ldots, \pat{t}\}$, the levels of $u_i$ in $\modalg$ and $\tgtalg$ are identical, i.e., $\modlev{t}_i = \lev{t}_i$;
        \item \label{inv:ub:prop:u3} for every $i \in \{\pat{t} + 1, \ldots, \qat{t}\}$, the level of $u_i$ in $\modalg$ is strictly less than that in $\tgtalg$, i.e., $\modlev{t}_i < \lev{t}_i$; and
        \item \label{inv:ub:prop:u4} for every $i \in \{\qat{t} + 1, \ldots, 2n - \tstar + 1\}$, the level of $u_i$ in $\modalg$ is at least that in $\tgtalg$, i.e., $\modlev{t}_i \geq \lev{t}_i$.
    \end{enumerate}
    Note that Properties~\ref{inv:ub:prop:u1} and~\ref{inv:ub:prop:u2} immediately imply Condition~\ref{cond:ub:suffcond1p}.
    We will see the other propositions --- $\modalg$ is well-defined and satisfies Condition~\ref{cond:ub:suffcond2} --- while we prove the existence of $\pat{t}$ and $\qat{t}$.
    
    We break the analysis into the three cases as follows.

    \textbf{Case 1. $\pat{t-1} \geq 2n-t+1 $.}
    Observe that we have $\modlev{t-1}_i = \lev{t-1}_i $ for every $i \in \{t, \ldots, 2n-t+1\} $ due to Property~\ref{inv:ub:prop:u2} at the beginning of iteration~$t$.
    We can therefore see that $\modalg$ effectively behaves the same as $\tgtalg$ in this iteration, implying that $\modalg$ is well-defined in this iteration.
    Note also that $\modalg$ saturates $v_t$ as $\tgtalg$ does in this iteration, implying that Condition~\ref{cond:ub:suffcond2} is satisfied.
    Moreover, since $v_t$ is adjacent only with $N(v_t) = \{u_t, \ldots, u_{2n-t+1}\}$, for every $i \in \{1, \ldots, t-1\} \cup \{2n-t+2, \ldots, 2n\}$, we have $\modlev{t}_i = \modlev{t-1}_i$ and $\lev{t}_i = \lev{t-1}_i$, implying that $\pat{t} := \pat{t-1}$ and $\qat{t} := \qat{t-1}$ would satisfy the properties at the end of iteration~$t$.

    \textbf{Case 2. $\pat{t-1} < 2n-t+1 \leq \qat{t-1} $.}
    Notice that, in this case, we have
    \begin{align*}
        \sum_{i=t}^{2n-t+1} \left( \lev{t}_i - \modlev{t-1}_i \right)
        & = \sum_{i=t}^{2n-t+1} \left( \lev{t}_i - \lev{t-1}_i \right) + \sum_{i=t}^{2n-t+1} \left( \lev{t-1}_i - \modlev{t-1}_i \right) \\
        & = 1 + \sum_{i=\pat{t-1}+1}^{2n-t+1} \left( \lev{t-1}_i - \modlev{t-1}_i \right) \\
        & > 1,
    \end{align*}
    where the second equality follows from the fact that $\tgtalg$ saturates $v_t$.
    Moreover, since $\lev{t}_t - \modlev{t-1}_t \leq 1$, we can deduce that there exists $p \in \{t, \ldots, 2n-t\}$ such that, for every $i \in \{t, \ldots, 2n-t+1\}$,
    \[
        \xbar_{u_i, v_t} = \begin{cases}
            \lev{t}_i - \modlev{t-1}_i, & \text{if } i \leq p, \\
            z, & \text{if } i = p+1, \\
            0, & \text{if } i \geq p+2,
        \end{cases}
        \text{ and hence }
        \modlev{t}_i = \begin{cases}
            \lev{t}_i, & \text{if } i \leq p, \\
            \modlev{t-1}_i + z, & \text{if } i = p+1, \\
            \modlev{t-1}_i, & \text{if } i \geq p+2,
        \end{cases}
    \]
    where $z := 1 - \sum_{i=t}^p \left( \lev{t}_i - \modlev{t-1}_i \right)$ denotes the remaining amount of $v_t$ to be saturated in this iteration.
    We can thus see that $\modalg$ is well-defined and satisfies Condition~\ref{cond:ub:suffcond2} in this iteration.
    Observe also that $\pat{t} := p$ and $\qat{t} := \qat{t-1}$ would satisfy the properties at the end of this iteration.    

    \textbf{Case 3. $\pat{t-1} < \qat{t-1} < 2n-t+1 $.}    
    Let $q$ denote the index such that $\lev{t}_q > \modlev{t-1}_q $ and $\lev{t}_{q+1} \leq \modlev{t-1}_{q+1} $.
    Observe that $q$ is unique and $q \in \{\qat{t-1}, \ldots, 2n-t+1\}$ due to the properties at iteration~$t-1$ and the monotonicity of $\tgtalg$.
    We first show that $\modalg$ would iterate up to $q$ to push water in this iteration, i.e.,
    $
        \sum_{i=t}^{q} \left( \lev{t}_i - \modlev{t-1}_i \right) > 1,
    $
    implying that $\modalg$ is well-defined and satisfies Condition~\ref{cond:ub:suffcond2}.
    Note that it suffices to show that 
    \begin{equation} \label{eq:ub:prop:pu}
        \sum_{i = t}^{q} \left( \lev{t-1}_i - \modlev{t-1}_i \right)
        > \sum_{i =q+1}^{2n-t+1} \left( \lev{t}_i - \lev{t-1}_i \right)
    \end{equation}
    since, if the above inequality is indeed true, we can immediately derive
    \begin{align*}
        \sum_{i=t}^{q} \left( \lev{t}_i - \modlev{t-1}_i \right)
        & = \sum_{i=t}^{q} \left( \lev{t}_i - \lev{t-1}_i \right) + \sum_{i=t}^{q} \left( \lev{t-1}_i - \modlev{t-1}_i \right) \\
        & > \sum_{i = t}^{2n-t+1} \left( \lev{t}_i - \lev{t-1}_i \right) \\
        & = 1.
    \end{align*}
    Since both $\modalg$ and $\tgtalg$ have so far saturated $\{v_1, \ldots, v_{t-1}\}$, we have $\sum_{i = 1}^{2n} \lev{t-1}_i = \sum_{i=1}^{2n} \modlev{t-1}_i$, implying that
    \begin{align*}
        0
        & = \sum_{i = 1}^{2n} \left( \lev{t-1}_i - \modlev{t-1}_i \right) \\
        & = \sum_{i = t}^{2n - \tstar + 1} \left( \lev{t-1}_i - \modlev{t-1}_i \right) \\
        & = \sum_{i = t}^{\qat{t-1}} \left( \lev{t-1}_i - \modlev{t-1}_i \right) - \sum_{i = \qat{t-1} + 1}^{2n - \tstar + 1} \left(\modlev{t-1}_i - \lev{t-1}_i \right),
    \end{align*}
    where the second equality follows from \Cref{eq:ub:prevlev} and the fact that $\pat{t-1} \geq t - 1$ due to Property~\ref{inv:ub:prop:u1}.
    We can thus obtain
    \begin{align*}
        \sum_{i=t}^{q} \left( \lev{t-1}_i - \modlev{t-1}_i \right)
        & = \sum_{i=t}^{\qat{t-1}} \left( \lev{t-1}_i - \modlev{t-1}_i \right) - \sum_{i=\qat{t-1}+1}^{q} \left( \modlev{t-1}_i - \lev{t-1}_i \right) \\
        & = \sum_{i=q+1}^{2n - \tstar + 1} \left( \modlev{t-1}_i - \lev{t-1}_i \right) \\
        & = \sum_{i={q}+1}^{2n-t+1} \left( \modlev{t-1}_i - \lev{t-1}_i \right) + \sum_{i=2n-t+2}^{2n - \tstar + 1} \left( \modlev{t-1}_i - \lev{t-1}_i \right) \\
        & \geq \sum_{i={q}+1}^{2n-t+1} \left( \modlev{t-1}_i - \lev{t-1}_i \right) \\
        & > \sum_{i = {q}+1}^{2n-t+1} \left( \lev{t}_i - \lev{t-1}_i \right),
    \end{align*}
    where the first inequality is derived from the condition that $\qat{t-1} < 2n - t + 1$ and the last from the definition of $q$, showing \Cref{eq:ub:prop:pu}.
    Together with the fact that $\lev{t}_t - \modlev{t-1}_t \leq 1$, we can also infer that there exists $p \in \{t, \ldots, q-1 \}$ in the execution of $\modalg$ such that, for every $i \in \{t, \ldots, 2n-t+1\}$,
    \[
        \xbar_{u_i, v_t} = \begin{cases}
            \lev{t}_i - \modlev{t-1}_i, & \text{if } i \leq p, \\
            z, & \text{if } i = p+1, \\
            0, & \text{if } i \geq p+2,
        \end{cases}
        \text{ and hence }
        \modlev{t}_i = \begin{cases}
            \lev{t}_i, & \text{if } i \leq p, \\
            \modlev{t-1}_i + z, & \text{if } i = p+1, \\
            \modlev{t-1}_i, & \text{if } i \geq p+2,
        \end{cases}
    \]
    where we again denote by $z := 1 - \sum_{i=t}^p \left( \lev{t}_i - \modlev{t-1}_i \right)$ the remaining amount of $v_t$ to be saturated in this iteration.
    Observe that $\pat{t} := p$ and $\qat{t} := q$ would satisfy the properties at the end of this iteration.

    We have shown that, for every iteration $t \in \{1, \ldots, \tsat\}$ where $\tgtalg$ saturates $v_t$, $\modalg$ also saturates $v_t$ while $\modlev{\tsat}_t = \lev{\tsat}_t$. Therefore, if $\tsat = n$, this immediately completes the proof of \Cref{lem:ub:uniform}.
    
    Suppose, on the other hand, that $\tsat + 1 \leq n$. In iteration~$\tsat + 1$, $\tgtalg$ saturates $N(v_{\tsat + 1})$, i.e., $\lev{\tsat + 1}_i = 1$ for every $i \in \{t, \ldots, 2n - \tsat \}$.
    We claim that
    \[
        \sum_{i=\tsat + 1}^{2n - \tsat} \left( 1-\modlev{\tsat}_i \right)
        \geq \sum_{i= \tsat + 1}^{2n- \tsat} \left( 1-\lev{\tsat}_i \right),
    \]
    implying that $\modalg$ is well-defined and satisfies Condition~\ref{cond:ub:suffcond2} in this iteration.
    Indeed, if $\qat{\tsat} \geq 2n-\tsat+1$, it is easy to see that the claim is satisfied due to Properties~\ref{inv:ub:prop:u2} and~\ref{inv:ub:prop:u3} of iteration~$\tsat$.
    Otherwise, if $\qat{\tsat} < 2n-\tsat+1$, we have
    \[
        \sum_{i=\tsat + 1}^{\qat{\tsat}} \left( \lev{\tsat}_i - \modlev{\tsat}_i \right)
        = \sum_{i=\qat{\tsat} + 1}^{2n - \tstar + 1} \left( \modlev{\tsat}_i - \lev{\tsat}_i \right)
        \geq \sum_{i = \qat{\tsat}+1}^{2n-\tsat} \left( \modlev{\tsat}_i - \lev{\tsat}_i \right),
    \]
    where the equality follows from the fact that both $\tgtalg$ and $\modalg$ have saturated all the online vertices that have arrived until this iteration (i.e., $\sum_{i = 1}^{2n} \lev{\tsat}_i = \sum_{i = 1}^{2n} \modlev{\tsat}_i$) while the inequality is due to Property~\ref{inv:ub:prop:u4}.
    We can then prove the claim since
    \begin{align*}
        \sum_{i = \tsat + 1}^{2n-\tsat} \left( 1 - \modlev{\tsat}_i \right) 
        & = \sum_{i = \tsat + 1}^{2n-\tsat} \left( 1 - \lev{\tsat}_i \right) + \sum_{i = \tsat + 1}^{2n - \tsat} \left( \lev{\tsat}_i - \modlev{\tsat}_i \right) \\
        & = \sum_{i = \tsat + 1}^{2n-\tsat} \left( 1 - \lev{\tsat}_i \right) + \sum_{i = \tsat + 1}^{\qat{\tsat}} \left( \lev{\tsat}_i - \modlev{\tsat}_i \right) - \sum_{i = \qat{\tsat} + 1}^{2n-\tsat} \left( \modlev{\tsat}_i - \lev{\tsat}_i \right) \\
        & \geq \sum_{i = \tsat+1}^{2n-\tsat} (1 - \lev{\tsat}_i).
    \end{align*}
    Note that Condition~\ref{cond:ub:suffcond1p} is satisfied in this iteration due to the execution of $\modalg$.
    
    Finally, for any iteration $t \geq \tsat+2$, it is easy to see that $\modalg$ is well-defined and satisfies Conditions~\ref{cond:ub:suffcond1p} and~\ref{cond:ub:suffcond2}.
\end{proof}

Let us lastly prove that the monotonicity at iteration~$n$ is also without loss of generality, i.e., at the beginning of Line~\ref{line:ub:robro2} in \Cref{alg:ub:rob}, $\tgtalg$ would satisfy $\lev{n}_1 \leq \lev{n}_2 \leq \ldots \leq \lev{n}_n$.

\begin{lemma} \label{lem:ub:montoneatn}
    Given any algorithm $\tgtalg$ that is greedy and uniform in the common phase, there exists an algorithm $\modalg$ that is monotone at iteration~$n$ as well while satisfying $\modalg(\robadv) \geq \tgtalg(\robadv)$ and $\modalg(\conadv) = \tgtalg(\conadv)$.
\end{lemma}
\begin{proof}
    We prove this lemma by inductively showing that, if there exists $\tstar \in \{1, \ldots, n-1\}$ such that $\lev{n}_{\tstar} = \lev{\tstar}_{\tstar} > \lev{\tstar + 1}_{\tstar + 1} = \lev{n}_{\tstar + 1}$ in the execution of $\tgtalg$, we can construct another algorithm $\modalg$ that satisfies the conditions of \Cref{lem:ub:suffcond} with, for every $t \in \{1, \ldots, n\}$,
    \begin{equation} \label{eq:ub:prop:monn}
        \modlev{n}_t = \modlev{t}_t = \begin{cases}
            \lev{\tstar+1}_{\tstar+1} = \lev{n}_{\tstar + 1}, & \text{if } t = \tstar, \\
            \lev{\tstar}_{\tstar}  = \lev{n}_{\tstar}, & \text{if } t = \tstar+1, \\
            \lev{t}_t  = \lev{n}_{t}, & \text{otherwise.}
        \end{cases}
    \end{equation}
    Notice that Condition~\ref{cond:ub:suffcond1} immediately follows from \Cref{eq:ub:prop:monn}.
    
    Let us now describe the execution of $\modalg$ in the common phase.
    In fact, as we want $\modalg$ to be greedy and uniform, the execution of $\modalg$ is determined by \Cref{eq:ub:prop:monn} (as long as it is feasible).
    Precisely speaking, for each iteration $t \in \{1, \ldots, n\}$, $\modalg$ first pushes water through $(u_t, v_t)$ until the level $\modlev{t}_t$ of $u_t$ satisfies \Cref{eq:ub:prop:monn}.
    The remaining amount is then distributed uniformly towards $N(v_t) \setminus \{u_t\}$.

    We need to prove that the execution of $\modalg$ is feasible and that Condition~\ref{cond:ub:suffcond2} is met by $\modalg$.
    To this end, let $\tsat$ be the last round at which $\tgtalg$ fully saturates $v_t$ in the common phase.
    Since $\tgtalg$ is greedy while we have $\lev{\tstar+1}_{\tstar+1} < \lev{\tstar}_{\tstar} \leq 1$, we can observe that $\tgtalg$ saturates $v_{\tstar + 1}$ in iteration~$\tstar + 1$, i.e., $\tsat \geq \tstar + 1$.
    Note also that, up to iteration~$\tstar - 1$, $\modalg$ would behave the same as $\tgtalg$, implying that the execution of $\modalg$ up to this iteration is feasible, and hence, Condition~\ref{cond:ub:suffcond2} is also satisfied.

    Let $z := \lev{\tstar}_{\tstar} - \lev{\tstar+1}_{\tstar+1} > 0 $.
    In iteration~$\tstar$, $\modalg$ pushes water towards $u_{\tstar}$ only until its level $\modlev{\tstar}_{\tstar}$ becomes $\lev{\tstar+1}_{\tstar+1}$.
    Therefore, $\modalg$ distributes $\frac{z}{2n - 2\tstar + 1}$ more units than $\tgtalg$ towards each neighbor $u \in N(v_{\tstar}) \setminus \{u_{\tstar}\}$, implying that, for any $i \in \{\tstar + 1, \ldots, 2n - \tstar + 1\}$,
    \begin{align}
        \modlev{\tstar}_i 
        & = \lev{\tstar}_i + \frac{z}{2n - 2\tstar +1} \label{eq:ub:mltsi} \\
        & \leq \lev{\tstar + 1}_{\tstar + 1} + \frac{z}{2n - 2\tstar +1} \nonumber \\
        & < \lev{\tstar + 1}_{\tstar + 1} + z \nonumber \\
        & = \lev{\tstar}_{\tstar} \nonumber \\
        & \leq 1, \nonumber
    \end{align}
    where the first inequality follows from the fact that $\lev{\tstar}_i = \lev{\tstar}_{\tstar+1} \leq \lev{\tstar+1}_{\tstar+1}$ due to the uniformity of $\tgtalg$, and the second from that $\tstar + 1 \leq n$.
    This shows that the execution of $\modalg$ in this iteration is indeed feasible and also satisfies Condition~\ref{cond:ub:suffcond2}.
    
    For each subsequent iteration $t \in \{\tstar + 1, \ldots, \tsat\}$ where $\tgtalg$ saturates $v_t$, we claim that the level of $N(v_t) \setminus \{u_t\}$ in $\modalg$ is less than that in $\tgtalg$, i.e., for every $i \in \{t+1, \ldots, 2n-t+1\}$,
    \begin{equation} \label{eq:ub:modlevalwaysless}
        \modlev{t}_i < \lev{t}_i.
    \end{equation}
    Note that this immediately implies that the execution of $\modalg$ is feasible and Condition~\ref{cond:ub:suffcond2} is also satisfied in this iteration.
    Let us inductively prove the claim.
    Indeed, in iteration~$\tstar+1$, note that the amount of water pushed towards $u_{\tstar + 1}$ is
    \begin{align*}
        \modlev{\tstar + 1}_{\tstar + 1} - \modlev{\tstar}_{\tstar + 1}
        & = \lev{\tstar}_{\tstar} - \modlev{\tstar}_{\tstar + 1} \\
        & = \lev{\tstar + 1}_{\tstar + 1} + z - \lev{\tstar}_{\tstar + 1} - \frac{z}{2n - 2\tstar + 1} \\
        & = \lev{\tstar + 1}_{\tstar + 1} - \lev{\tstar}_{\tstar + 1} + \frac{2n - 2\tstar}{2n - 2\tstar + 1} \cdot z,
    \end{align*}
    where the second equality comes from the definition of $z$ and \Cref{eq:ub:mltsi}, meaning that $\modalg$ pushes $\frac{2n - 2\tstar}{2n - 2\tstar + 1} \cdot z$ more units towards $u_{\tstar + 1}$ than $\tgtalg$.
    We can therefore deduce that every neighbor $N(v_{\tstar + 1}) \setminus \{u_{\tstar + 1}\}$ other than the advice $u_{\tstar + 1}$ would gain $\frac{2n - 2\tstar}{(2n - 2\tstar - 1)(2n - 2\tstar + 1)} \cdot z$ less units of water in $\modalg$ than in $\tgtalg$ at this iteration while it has gained $\frac{z}{2n - 2\tstar + 1}$ more units in the previous iteration, i.e., for every $i \in \{\tstar+2, \ldots, 2n-\tstar\}$,
    \[
        \modlev{\tstar+1}_i = \lev{\tstar+1}_i + \frac{z}{2n - 2\tstar +1} - \frac{2n-2\tstar}{(2n-2\tstar-1)(2n-2\tstar+1)} \cdot z < \lev{\tstar+1}_i
    \]
    as claimed in \Cref{eq:ub:modlevalwaysless}.

    For the remaining iterations $t \in \{\tstar +2, \ldots, \tsat\}$, let $\varepsilon := \lev{t-1}_{t} - \modlev{t-1}_{t}$.
    Observe that $\modalg$ pushes $\varepsilon$ more units towards $u_t$ than $\tgtalg$ to have $\modlev{t}_t = \lev{t}_t$.
    This implies that $\modalg$ pushes$\frac{2n-2t+2}{2n-2t+1} \cdot \varepsilon$ less units towards each neighbor in $N(v_t) \setminus \{u_t\}$, yielding that
    \[
        \modlev{t}_i = \lev{t}_i - \frac{2n-2t+2}{2n-2t+1} \cdot \varepsilon < \lev{t}_i
    \]
    for every $i \in \{t + 1, \ldots, 2n - t + 1\}$.
    This completes the proof of \Cref{eq:ub:modlevalwaysless}.

    Note that, if $\tsat = n$, the proof of \Cref{lem:ub:montoneatn} immediately follows from the claim.
    On the other hand, if $\tsat < n$, it is easy to observe that the execution of $\modalg$ is feasible and that Condition~\ref{cond:ub:suffcond2} is also satisfied since every neighbor $N(v_{\tsat + 1})$ of $v_{\tsat + 1}$ becomes saturated in $\tgtalg$, and therefore, $\tgtalg$ pushes no water from iteration~$\tsat + 2$ (if any).
\end{proof}

\section{Experiments}
\label{sec:exp}

In this section, we present experimental results for empirical evaluation of our algorithms.
We experimented on synthetic random graphs defined in \cref{sec:exp-synthetic} and real-world graphs defined in \cref{sec:exp-real}.
For weighted instances, each offline vertex is given a uniform random weight between $0$ and $1000$.
\cref{sec:exp-advgen} presents the way of generating advice given a noise parameter, followed by description of benchmarked algorithms in \cref{sec:exp-algs}.
Each plot is generated by letting each algorithm solve $10$ instances for $10$ different noise parameter values.
That is, a plot for weighted instances with $6$ algorithms involved solving $600$ instances while a plot for unweighted instances with $10$ algorithms involved solving $1000$ instances.
We defer the full plots of our experimental results to \cref{sec:appendix-experiments}.
All experiments were performed on a personal laptop (Apple Macbook 2024, M4 chip, 16GB memory).

\subsection{Synthetic random graphs}
\label{sec:exp-synthetic}

\paragraph{Erd\H{o}s-R\'enyi (ER) graphs.}
Given a number of nodes $n \in \{100, 200, 300\}$ and edge probability $p \in \{0.1, 0.2, 0.5\}$, an ER graph is generated with $n$ offline nodes and $n$ online nodes and each edge in the complete bipartite graph exists independently with probability $p$.

\paragraph{Upper Triangular (UT) graphs.}
Given a number of nodes $n \in \{100, 200, 300\}$, a UT graph is generated with $n$ offline nodes and $n$ online nodes where the $i$-th online node is connected to the last $n-i+1$ offline nodes.

\subsection{Real-world graphs}
\label{sec:exp-real}

To evaluate our algorithmic performance on real-world graph structures, we considered 6 publicly available graphs from the Network Data Repository \cite{rossi2015network} and pre-processed them in a similar manner to \cite{borodin2020experimental} to obtain random bipartite graphs: first, shuffle all $n$ node indices in the real-world graph, take the first $\lfloor n/2 \rfloor$ as the offline vertices and the next $\lfloor n/2 \rfloor$ as online vertices and only keep the bipartite crossing edges.
Each random shuffle of the real-world graph induces a random bipartite graph instance which we then experiment on.
Note that such a pre-processing step is necessary because these real-world graphs are not bipartite to begin with.

\subsection{Advice generation}
\label{sec:exp-advgen}

For each graph $\cG$ with $n$ vertices and a given noise parameter $\gamma \in [0,1]$, we generate a noisy prediction $\wh{\cG}_\gamma$ of $\cG$ as follows: each online vertex $v$ retains a random $(1 - \gamma)$ fraction of its true neighbors and gains a random $\gamma$ fraction of its non-neighbors.
Thus, when $\gamma = 0$, the prediction is exact ($\wh{\cG}_0 = \cG$), and when $\gamma = 1$, it corresponds to the complement graph ($\wh{\cG}_1 = \overline{\cG}$).

To generate the advice for the $t$-th arriving online vertex (for $t \in [n]$), we solve a linear program that maximizes the (weighted) matching objective.
This is done subject to two components: the actual decisions made for the first $t-1$ arrivals in the true graph $\cG$, and a noisy prediction of the future arrivals from time $t+1$ to $n$, based on $\wh{\cG}_\gamma$.
Importantly, the current arrival at time $t$ is not included in the noisy future but is instead the decision variable of interest.
In more detail, the advice at time $t$ is generated by perturbing the true future subgraph (i.e., the part of $\cG$ involving vertices $t+1$ to $n$) to create a noisy forecast.
We then solve for the optimal decision at time $t$ that maximizes the matching value, given the past decisions up to $t-1$ (in $\cG$) and the predicted future (in $\wh{\cG}_\gamma$).
Since we use the true graph up to and including time $t$, this process ensures that the advice at each time step is always feasible and based on a valid optimization problem over a fully specified $n$-vertex instance.

\subsection{Benchmarked algorithms}
\label{sec:exp-algs}
The two baselines are \textsc{Greedy} and \textsc{Balance}.
The former greedily matches the online vertex with its highest weighted available offline neighbor while the latter fractionally matches based on the penalty function $g(z) = e^{z-1}$.
In the unweighted setting, \textsc{Balance} is equivalent to the classic \textsc{Waterfilling} algorithm.
Note that both \textsc{Greedy} and \textsc{Balance} are independent of any predictions so they would achieve constant performance for any noise parameter $\gamma \in [0,1]$.
We also implemented and benchmarked our \textsc{LearningAugmentedBalance} (\textsc{LAB}; \cref{alg:lab}) and \textsc{PushAndWaterfill} (\textsc{PAW}; \cref{alg:paw}) algorithms, where each takes as inputs $\lambda_{\textsc{LAB}}$ and $\lambda_{\textsc{PAW}}$ respectively.
Note that the guarantees for \textsc{PAW} only hold for unweighted instances.

Recall from \cref{thm:vertex_weighted,thm:uw:main} that \textsc{LAB} and \textsc{PAW} have different consistency values with respect to their parameters: the consistency of \textsc{LAB} is $1 + \lambda_{\textsc{LAB}} - \exp(\lambda_{\textsc{LAB}}-1)$ while the consistency of \textsc{PAW} is $1 - (1-\lambda_{\textsc{PAW}}) \exp(\lambda_{\textsc{PAW}}-1)$.
To compare between them at the same consistency value, we set $\lambda_{\textsc{PAW}} = 1 + W(\lambda_\textsc{LAB} - \exp(\lambda_\textsc{LAB}-1))$.
Since \LAB{} with $\lambda_{\LAB} = 0$ and \PAW{} with $\lambda_{\PAW} = 0$ are already equivalent with \Balance{} of consistency $1-\nicefrac{1}{e}$,
we consider consistency ratios of $\{0.7, 0.8. 0.9, 1.0\}$ when running \textsc{LAB} and \textsc{PAW}.
This translates to the following parameters:
\begin{itemize}
    \item For consistency $0.7$, $\lambda_{\textsc{LAB}} \approx 0.111113$ and $\lambda_{\textsc{PAW}} \approx 0.510598$.
    \item For consistency $0.8$, $\lambda_{\textsc{LAB}} \approx 0.293239$ and $\lambda_{\textsc{PAW}} \approx 0.740829$.
    \item For consistency $0.9$, $\lambda_{\textsc{LAB}} \approx 0.516817$ and $\lambda_{\textsc{PAW}} \approx 0.888167$.
    \item For consistency $1.0$, $\lambda_{\textsc{LAB}} = \lambda_{\textsc{PAW}} = 1$.
\end{itemize}

\begin{figure}[htb]
\centering
\includegraphics[width=\textwidth]{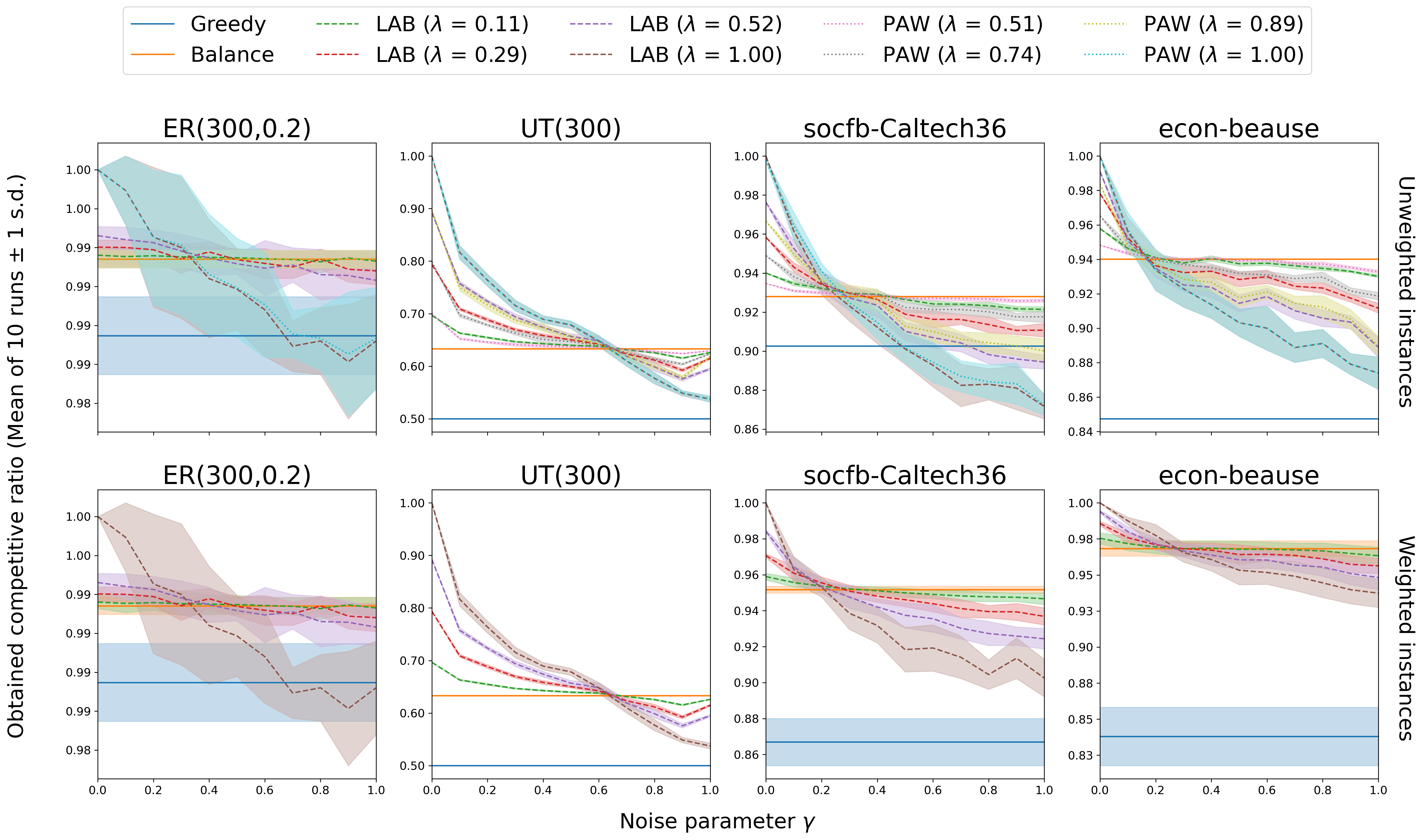}
\caption{Subset of empirical results: ER($300,0.2$), UT($300$), and 2 real-world graphs (socfb-Caltech36, econ-beause). See the \cref{sec:appendix-experiments} for our full set of experiments.}
\label{fig:short-exp}
\end{figure}

\subsection{Qualitative takeaways}
\label{sec:exp-qualitative-takeaways}

\cref{fig:short-exp} illustrates a subset of our empirical results.
As predicted by our analysis, the competitive ratio attained by both \textsc{LAB} and \textsc{PAW} degrades as the noise parameter $\gamma$ increases.
In particular, when $\gamma = 0$ (i.e., perfect advice), both \textsc{LAB} and \textsc{PAW} achieve a competitive ratio of 1 when $\lambda_{\textsc{LAB}} = \lambda_{\textsc{PAW}} = 1$.
As $\gamma$ grows large, the advice becomes increasingly uninformative, and it is unsurprising that the advice-free algorithm \textsc{Balance} eventually outperforms both learning-augmented algorithms, with the crossing point depending on the underlying graph instance.

Interestingly, across all our experiments --- including those in the appendix --- we consistently observe a phenomenon where there appears to exist a critical noise level $\gamma^*$ such that the competitive ratios of all runs of \textsc{LAB} and \textsc{PAW} (across different $\lambda$ values) converge and coincide with that of \textsc{Balance}.
This suggests that at $\gamma^*$, the advice becomes effectively uncorrelated with the input, causing the behavior of \textsc{LAB} and \textsc{PAW} to resemble that of \textsc{Balance} regardless of the weighting parameter $\lambda$.
While we do not currently have a theoretical explanation for this convergence, it is a compelling empirical observation that may point to deeper structure in the robustness-consistency tradeoff and warrants further investigation in future work.

\section{Conclusion and Open Problems}
\label{sec:conclusion}

We studied the robustness-consistency tradeoffs of learning-augmented algorithms for online bipartite fractional matching.
We proposed and analyzed two algorithms, \textsc{LAB} and \textsc{PAW}, and established an improved hardness result.

In our current work, \textsc{PAW} relies on integral advice while \textsc{LAB} can accommodate fractional advice.
While it is a natural question to ask if there can be a unified algorithm and analysis, our current analytical framework is unable to do so.
The analysis of \textsc{LAB} is agnostic to the weights, making it unclear how to demonstrate an improved tradeoff in the unweighted case.
Meanwhile, the analysis of \textsc{PAW} crucially relies on the integrality of the advice, and we were unable to obtain a comparable bound in the fractional case.
We do not rule out the possibility of a unified analysis, and we view this as a compelling direction for future work.
We do not rule out the possibility of a unified analysis and view this as an intriguing direction for future work.

Besides unifying the two variants, there are serveral other natural open and interesting research directions.
Firstly, it would be interesting to develop a theoretical explanation for the crossing point phenomenon observed in our experiments; see the discussion in \cref{sec:exp-qualitative-takeaways}.
Another direction would be to close the gap between our algorithmic results and the impossibility bounds.
Progress on this front could come from an analytic proof of the impossibility result, as well as a tight analysis of \textsc{LAB} in the unweighted setting.
Finally, it would be interesting to extend our framework to broader variants of online matching, including Display Ads, the generalized assignment problem \cite{feldman2009onlinead, spaeh2023online}, and the multi-stage setting \cite{feng2024batching}.

\begin{ack}
This research/project is supported by the National Research Foundation, Singapore under its AI Singapore Programme (AISG Award No: AISG-PhD/2021-08-013).
This work was partly supported by Institute of Information \& communications Technology Planning \& Evaluation (IITP) grant funded by the Korea government (MSIT) (No. RS-2021-II212068, Artificial Intelligence Innovation Hub).
This work was partly supported by an IITP grant funded by the Korean Government (MSIT) (No. RS-2020-II201361, Artificial Intelligence Graduate School Program (Yonsei University)).
Supported by NCN grant number 2020/39/B/ST6/01641.
\end{ack}

\bibliographystyle{alpha}
\bibliography{refs}

\newpage
\appendix
\section{Pseudocodes of our algorithms}
\label{sec:appendix-pseudocodes}

\begin{algorithm}
\caption{Learning-Augmented Balance Algorithm (\textsc{LAB})}
\label{alg:lab}
\KwIn{Offline vertices $U$, tradeoff parameter $\lambda \in [0,1]$}
\KwData{Online vertices $V$, edges $E$, and fractional advice $a \in \mathbb{R}^{E}$}
\KwOut{Fractional matching $x \in \mathbb{R}^{E}$}

\ForEach{$u \in U$}{
    $X_u \gets 0$ \tcp*{Amount allocated by algorithm} 
    $A_u \gets 0$ \tcp*{Amount allocated by advice}
}

\ForEach{arrival of $v \in V$ with neighbors $N(v)$ and advice $\{a_{u,v}\}_{u \in N(v)}$}{
    \ForEach{$u \in N(v)$}{
        $A_u \gets A_u + a_{u,v}$ \tcp*{Accumulate advice} \;
    }

    Find the smallest $\ell \geq 0$ such that
    $\sum_{u \in N(v)} x_{u,v} \leq 1$,
     where
    $x_{u,v} := \min \{ z \in [0, 1 - X_u] \mid w_u \cdot (1 - f(A_u, X_u + z)) \leq \ell \}$
     \tcp*{e.g. via binary search}\;

    \ForEach{$u \in N(v)$}{
        $X_u \gets X_u + x_{u,v}$ \tcp*{Accumulate actual fractional matching} \;
    }
}
\Return{$x$}
\end{algorithm}

\begin{algorithm}
\caption{Push-and-Waterfill Algorithm (\textsc{PAW})}
\label{alg:paw}
\KwIn{Offline vertices $U$, trade-off parameter $\lambda \in [0,1]$}
\KwData{Online vertices $V$, edges $E$, and integral advice $A: V \to U \cup \{\bot\}$}
\KwOut{Fractional matching $x \in \mathbb{R}^{E}$}

\ForEach{$u \in U$}{
    $d_u \gets 0$ \tcp*{Level of $u$}
}

\ForEach{arrival of $v \in V$ with neighbors $N(v)$ and advice $A(v)$}{
    \textbf{(Phase 1)}: Push to advised neighbor $A(v)$, up to $\tau = \max\{0, \lambda - d_{A(v)}\}$ amount\;
    
    \eIf{$A(v) \in N(v)$}{
        $\tau \gets \max\{0, \lambda - d_{A(v)}\}$ \;
        $x_{A(v), v} \gets \tau$ \;
        $d_{A(v)} \gets d_{A(v)} + \tau$ \;
    }{
        $\tau \gets 0$ \;
    }

    \textbf{(Phase 2)}: Waterfill the remaining $1 - \tau$\;
    
    Find the largest $\ell$ such that $\sum_{u \in N(v)} \max\{0, \ell - d_u\} \leq 1 - \tau$ \;
    
    $\ell \gets \min\{\ell, 1\}$\;
    
    \ForEach{$u \in N(v)$}{
        $x_{u,v} \gets x_{u,v} + \max\{0, \ell - d_u\}$ \;
        $d_u \gets d_u + \max\{0, \ell - d_u\}$ \;
    }
}
\Return{$x$}
\end{algorithm}

\newpage
\section{Plots of our experimental results}
\label{sec:appendix-experiments}

\begin{figure}[htb]
    \centering
    \begin{subfigure}{0.32\textwidth}
        \centering
        \includegraphics[width=\textwidth]{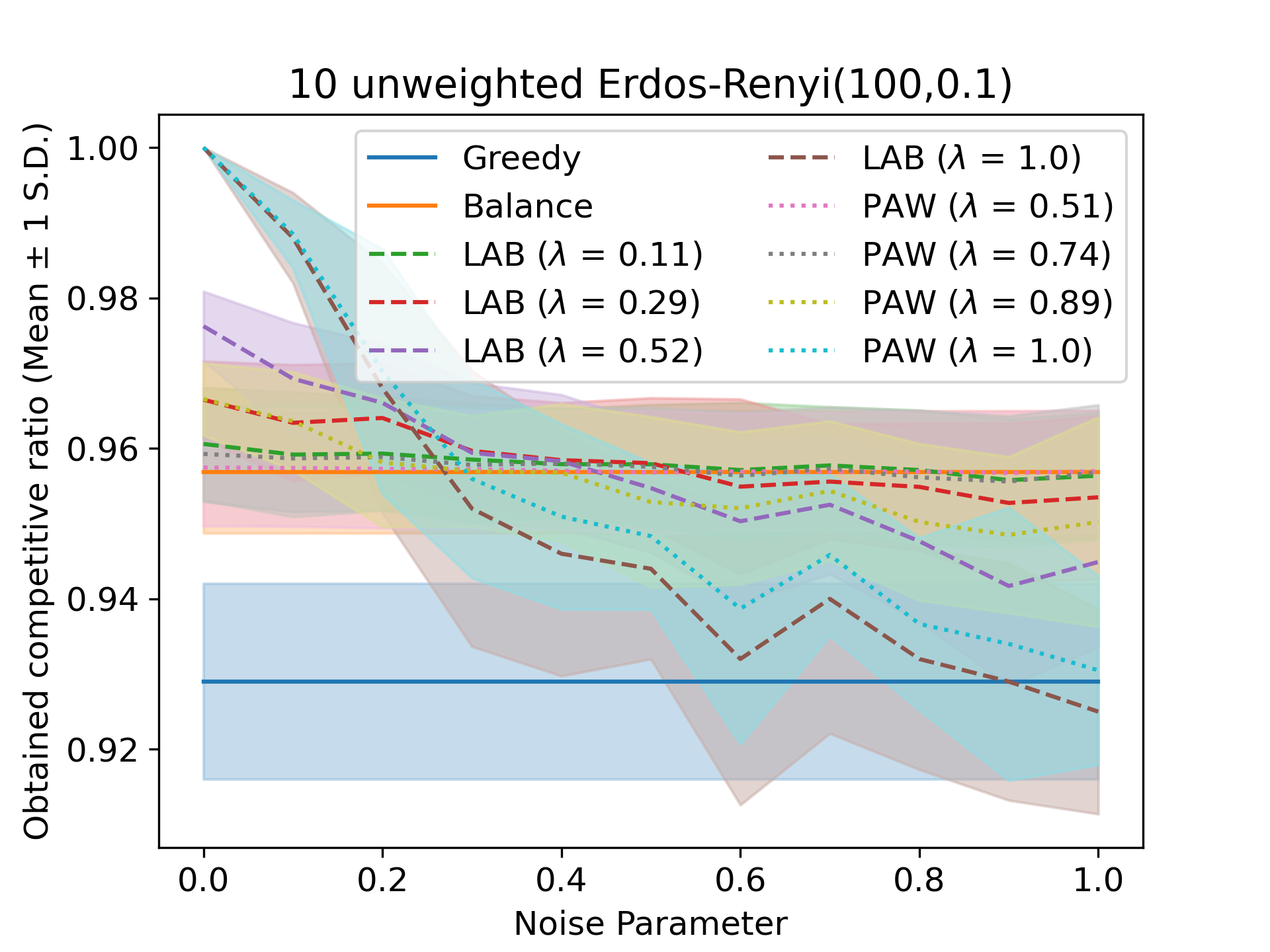}
    \end{subfigure}
    \begin{subfigure}{0.32\textwidth}
        \centering
        \includegraphics[width=\textwidth]{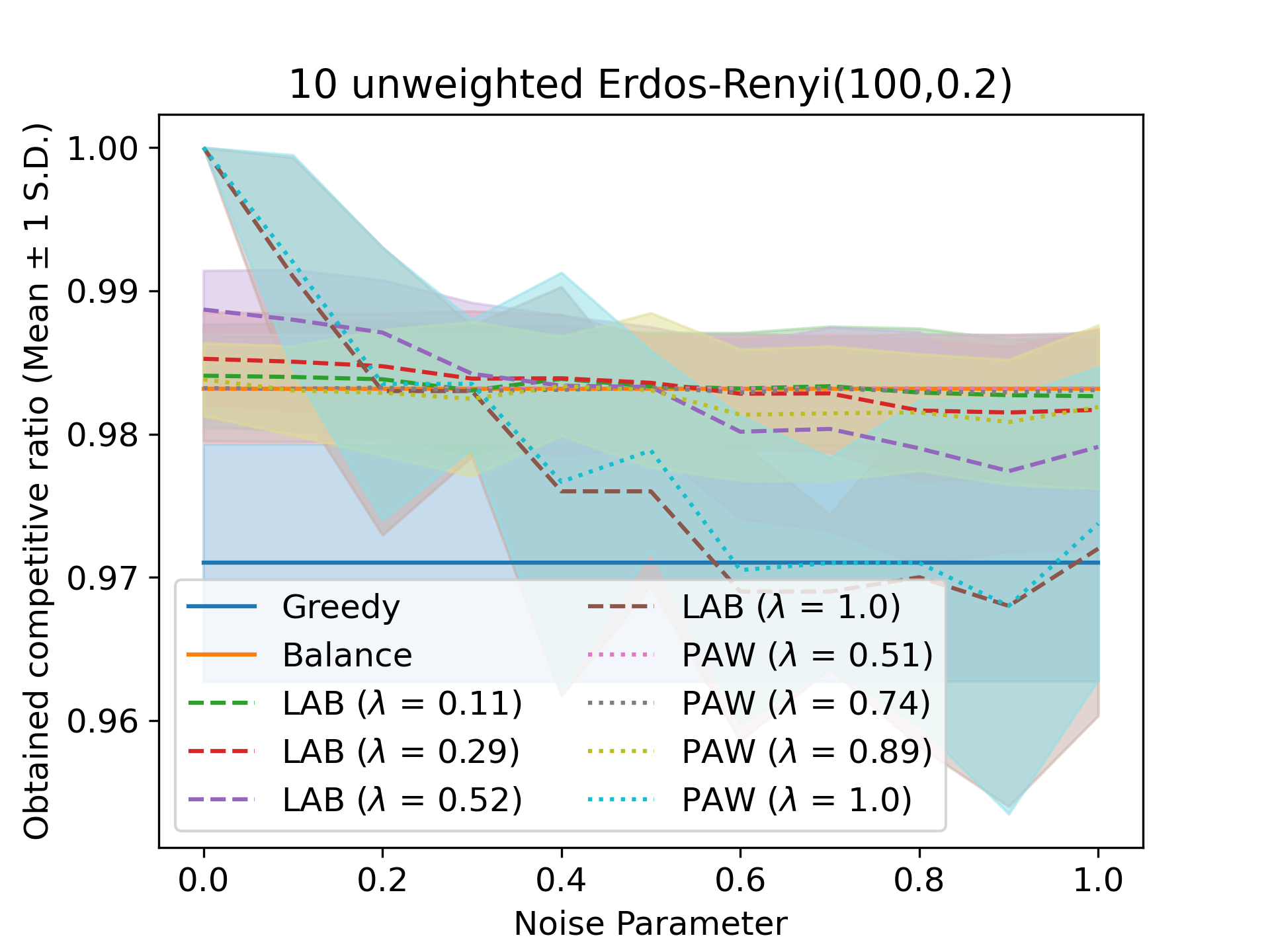}
    \end{subfigure}
    \begin{subfigure}{0.32\textwidth}
        \centering
        \includegraphics[width=\textwidth]{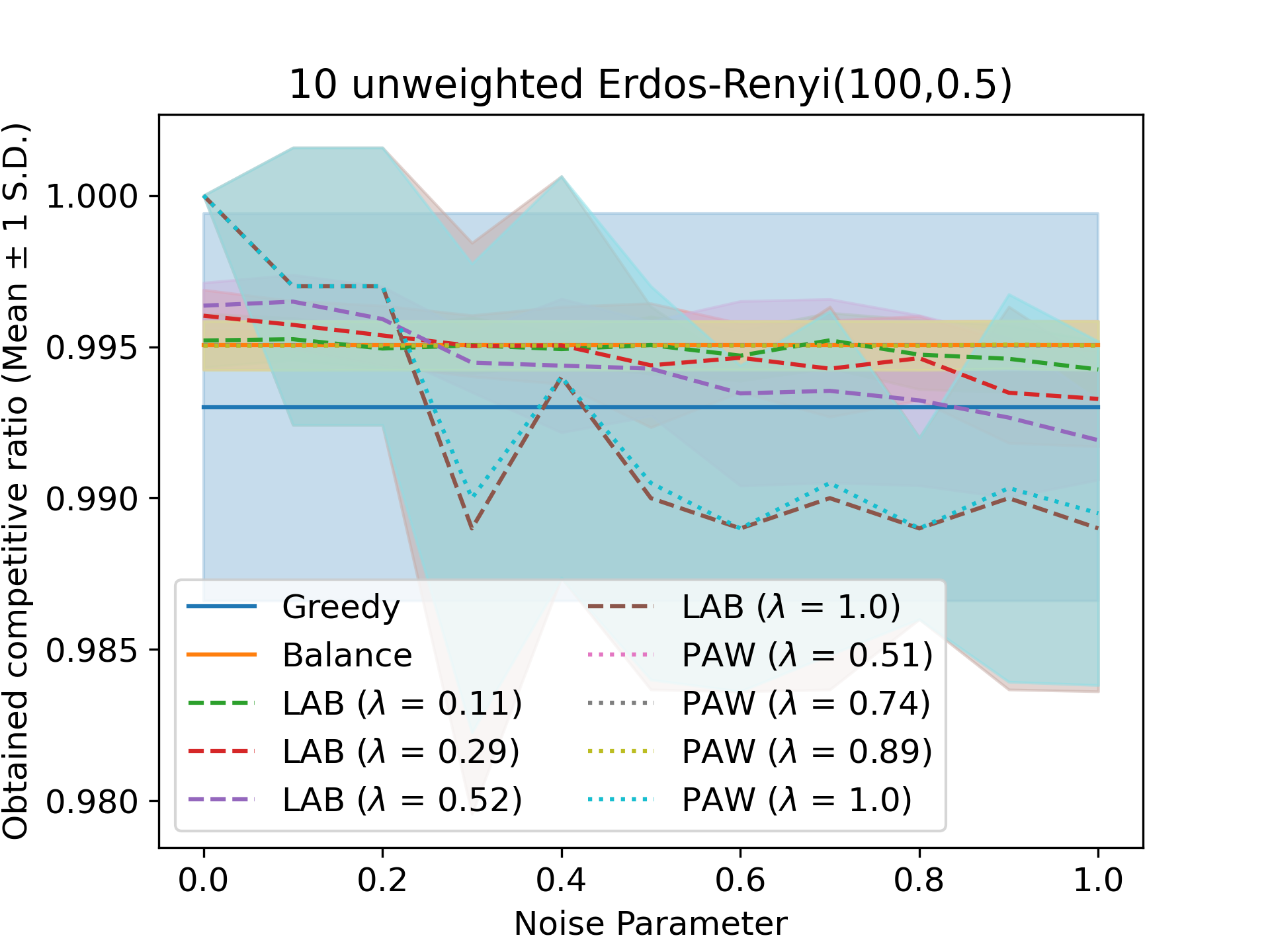}
    \end{subfigure}
    \\
    \vspace{10pt}
    \begin{subfigure}{0.32\textwidth}
        \centering
        \includegraphics[width=\textwidth]{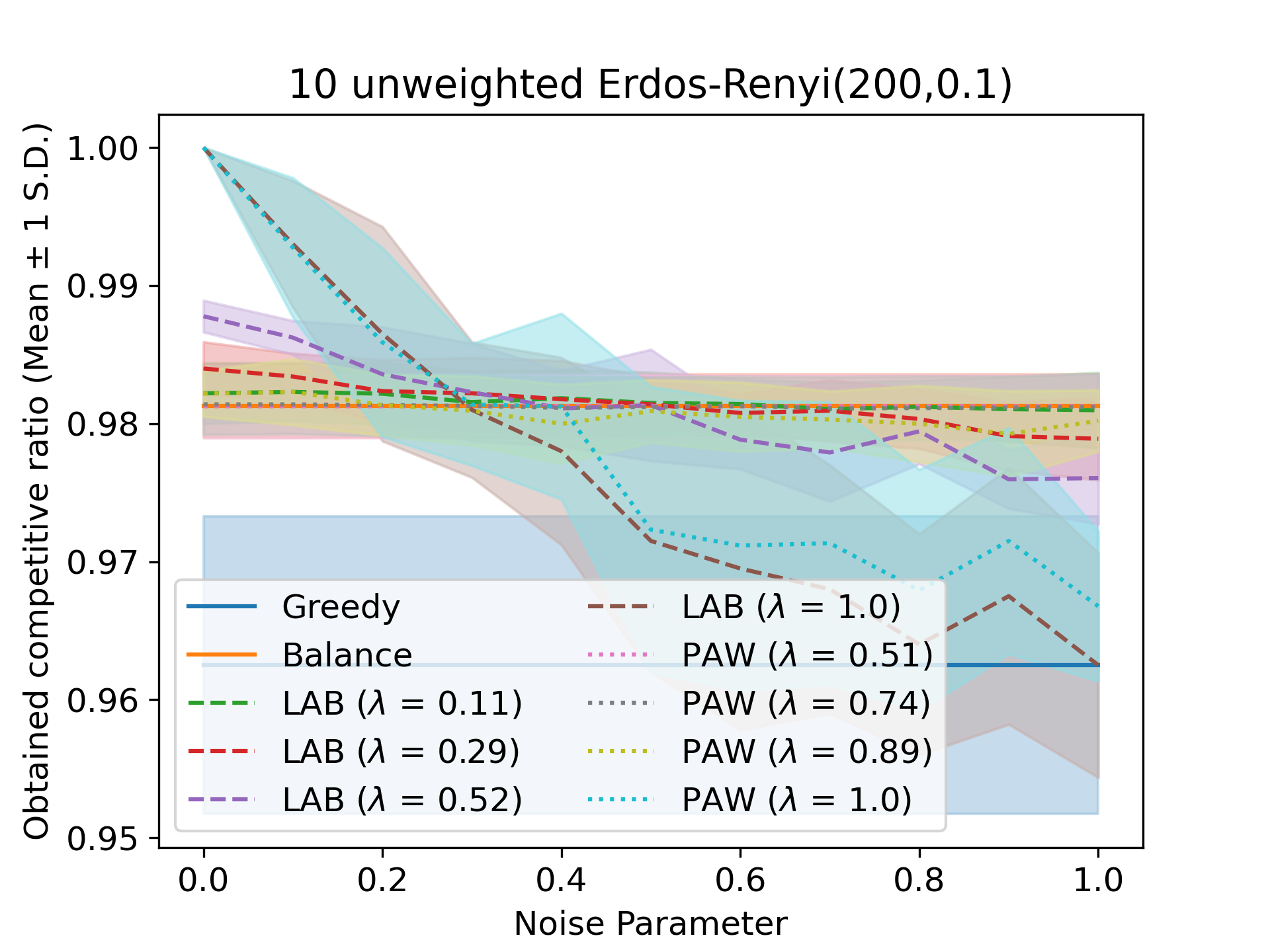}
    \end{subfigure}
    \begin{subfigure}{0.32\textwidth}
        \centering
        \includegraphics[width=\textwidth]{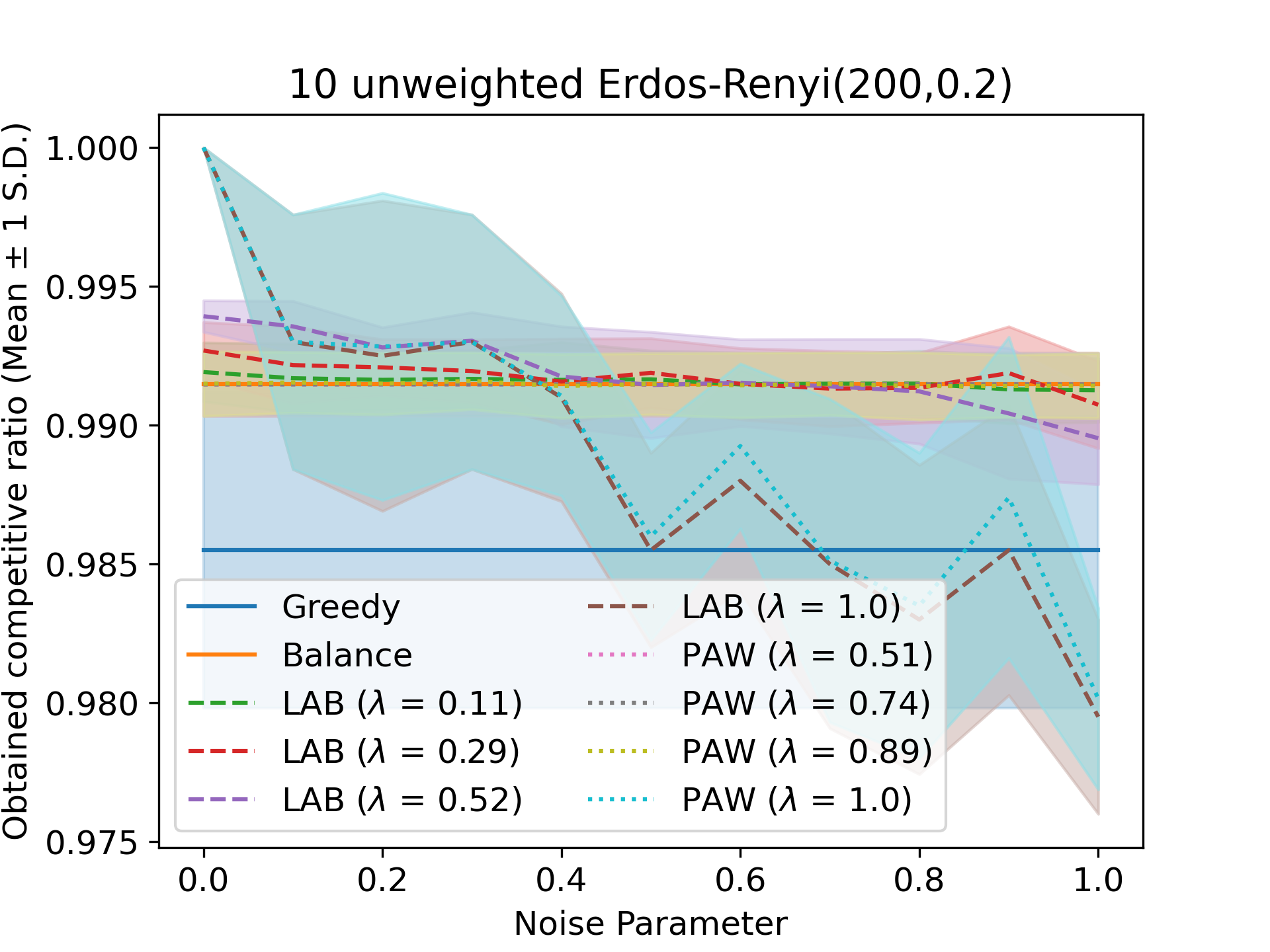}
    \end{subfigure}
    \begin{subfigure}{0.32\textwidth}
        \centering
        \includegraphics[width=\textwidth]{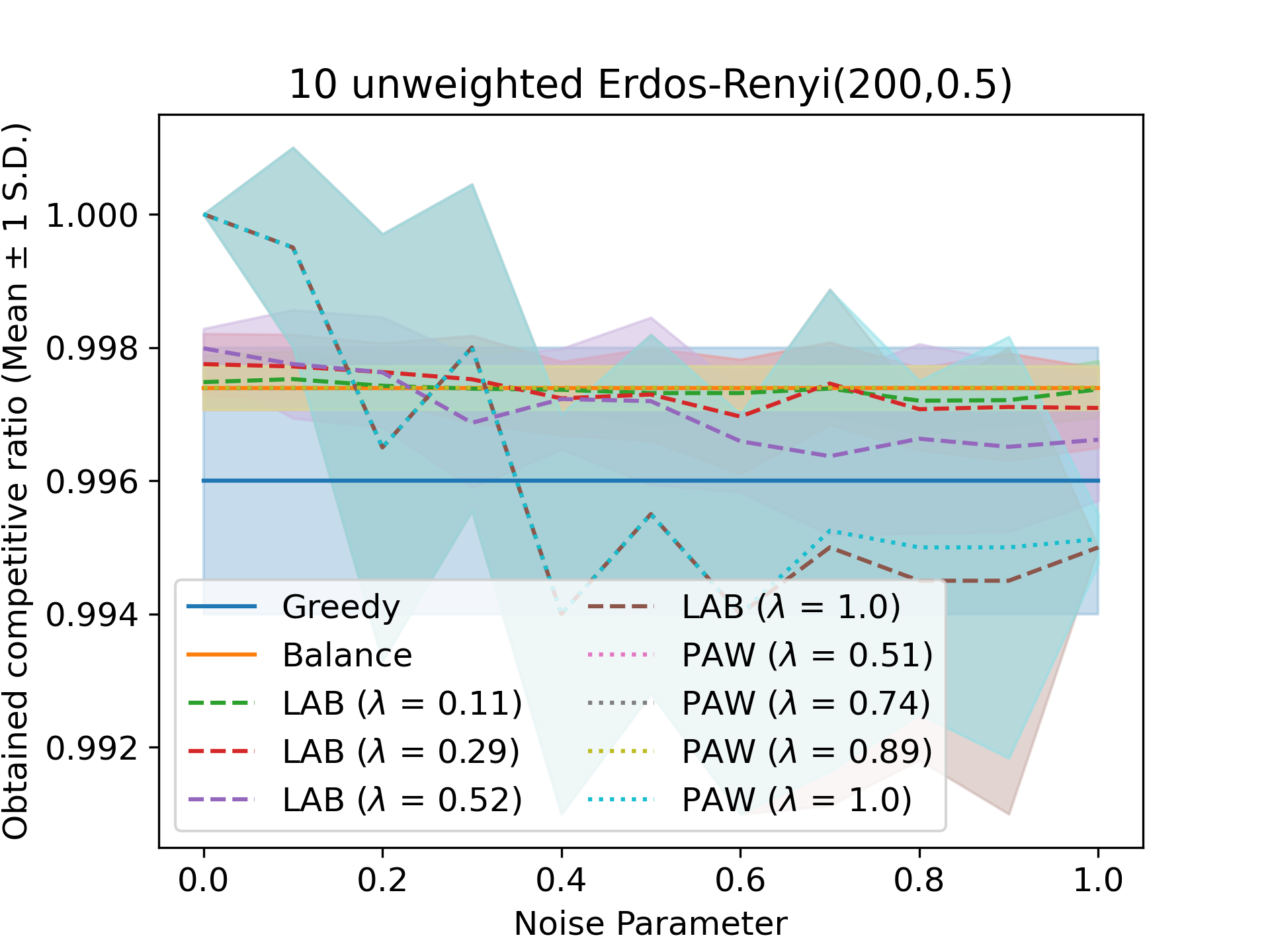}
    \end{subfigure}
    \\
    \vspace{10pt}
    \begin{subfigure}{0.32\textwidth}
        \centering
        \includegraphics[width=\textwidth]{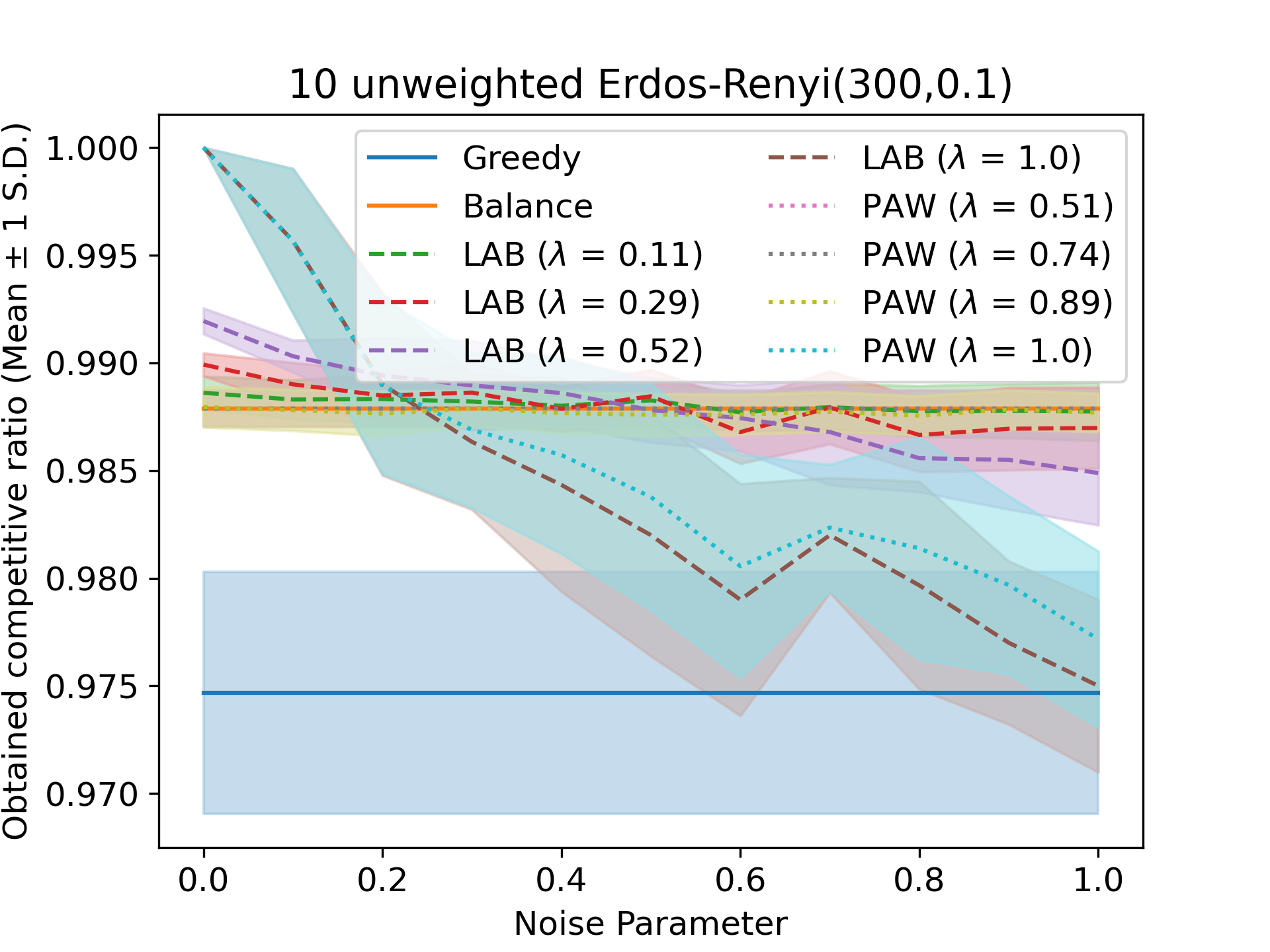}
    \end{subfigure}
    \begin{subfigure}{0.32\textwidth}
        \centering
        \includegraphics[width=\textwidth]{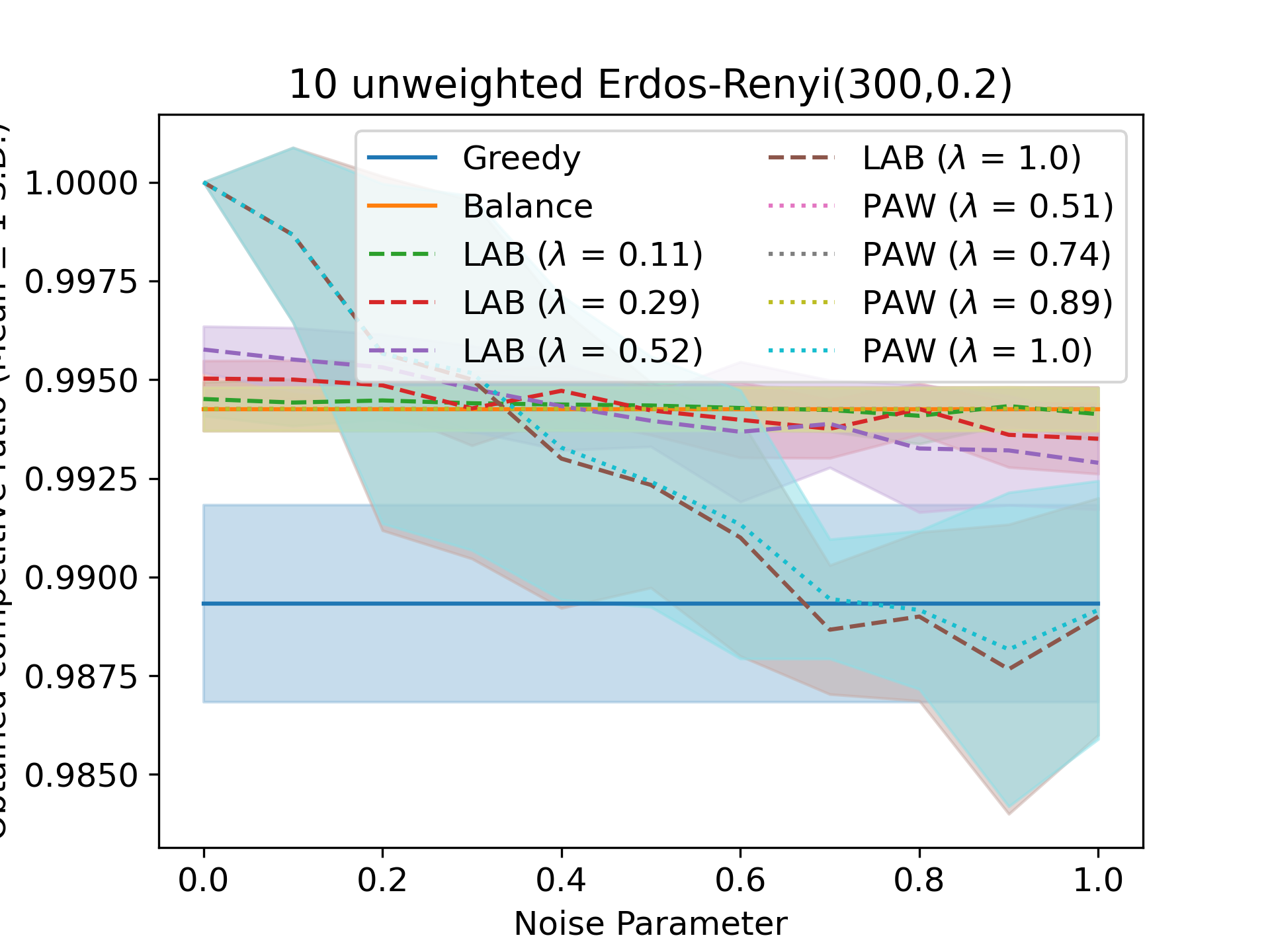}
    \end{subfigure}
    \begin{subfigure}{0.32\textwidth}
        \centering
        \includegraphics[width=\textwidth]{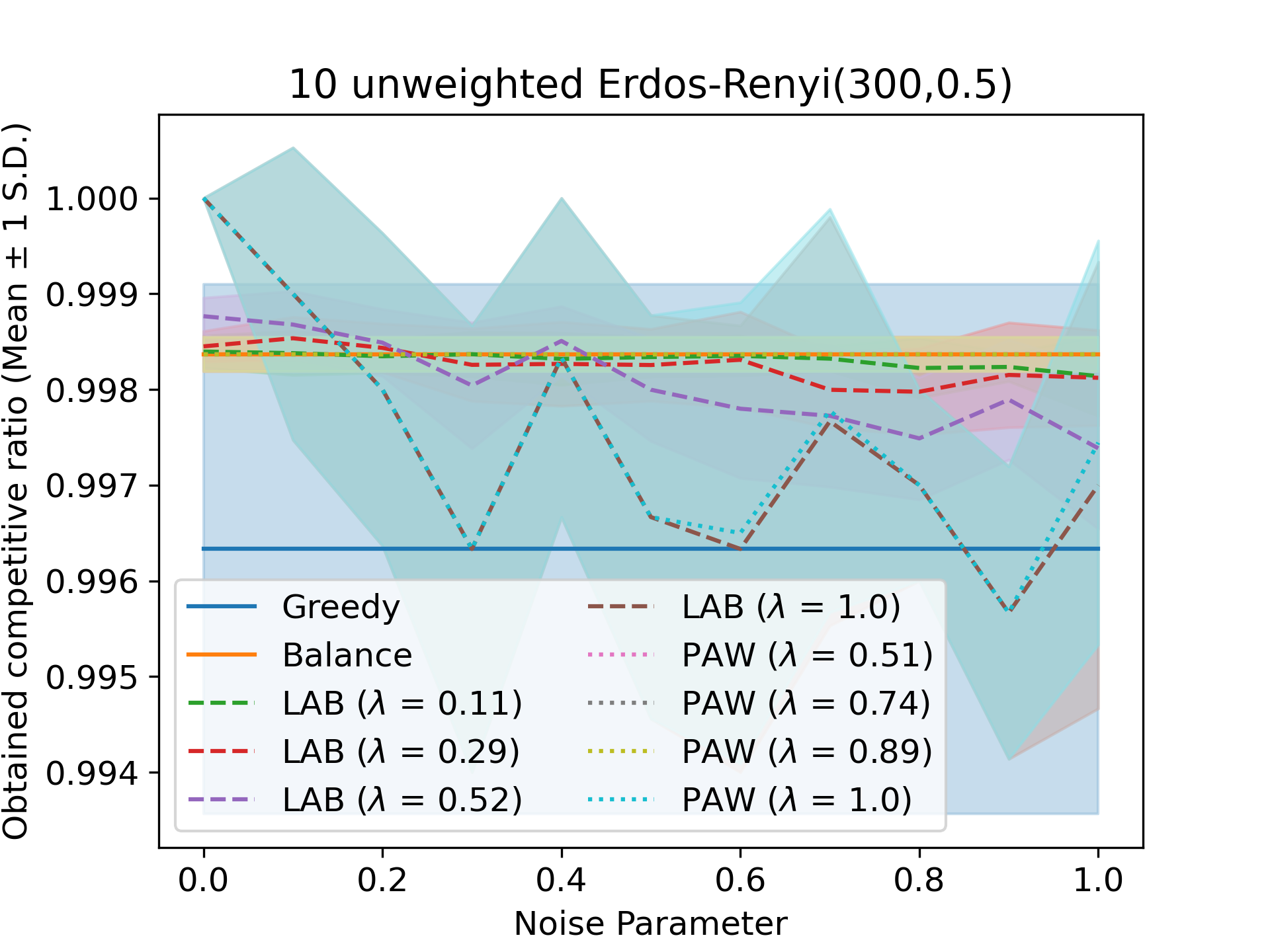}
    \end{subfigure}
    \caption{Empirical results for unweighted Erd\H{o}s-R\'enyi graph instances with $n \in \{100, 200, 300\}$ and $p \in \{0.1, 0.2, 0.5\}$}
    \label{fig:exp-ER-unweighted}
\end{figure}

\begin{figure}[htb]
    \centering
    \begin{subfigure}{0.32\textwidth}
        \centering
        \includegraphics[width=\textwidth]{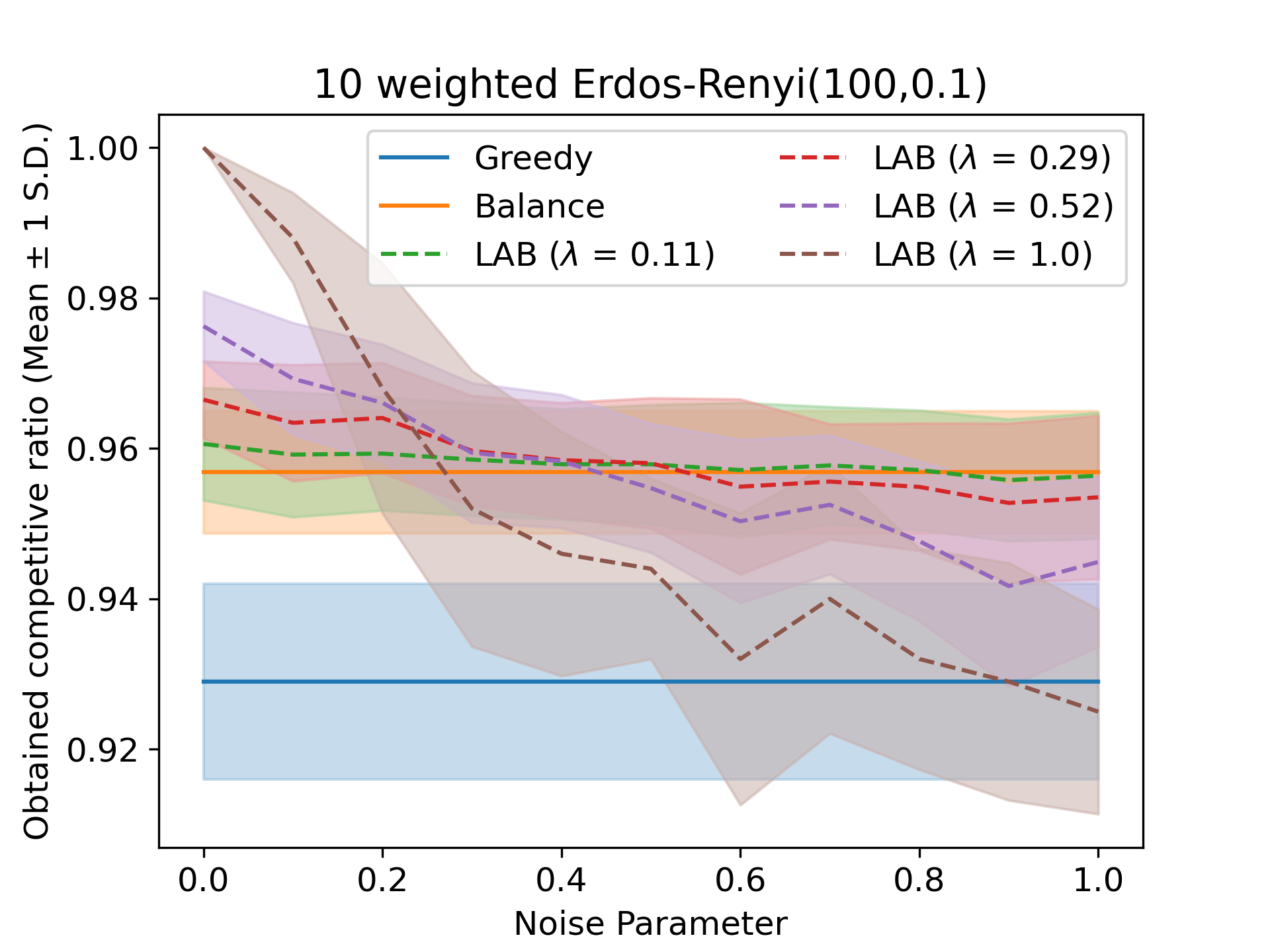}
    \end{subfigure}
    \begin{subfigure}{0.32\textwidth}
        \centering
        \includegraphics[width=\textwidth]{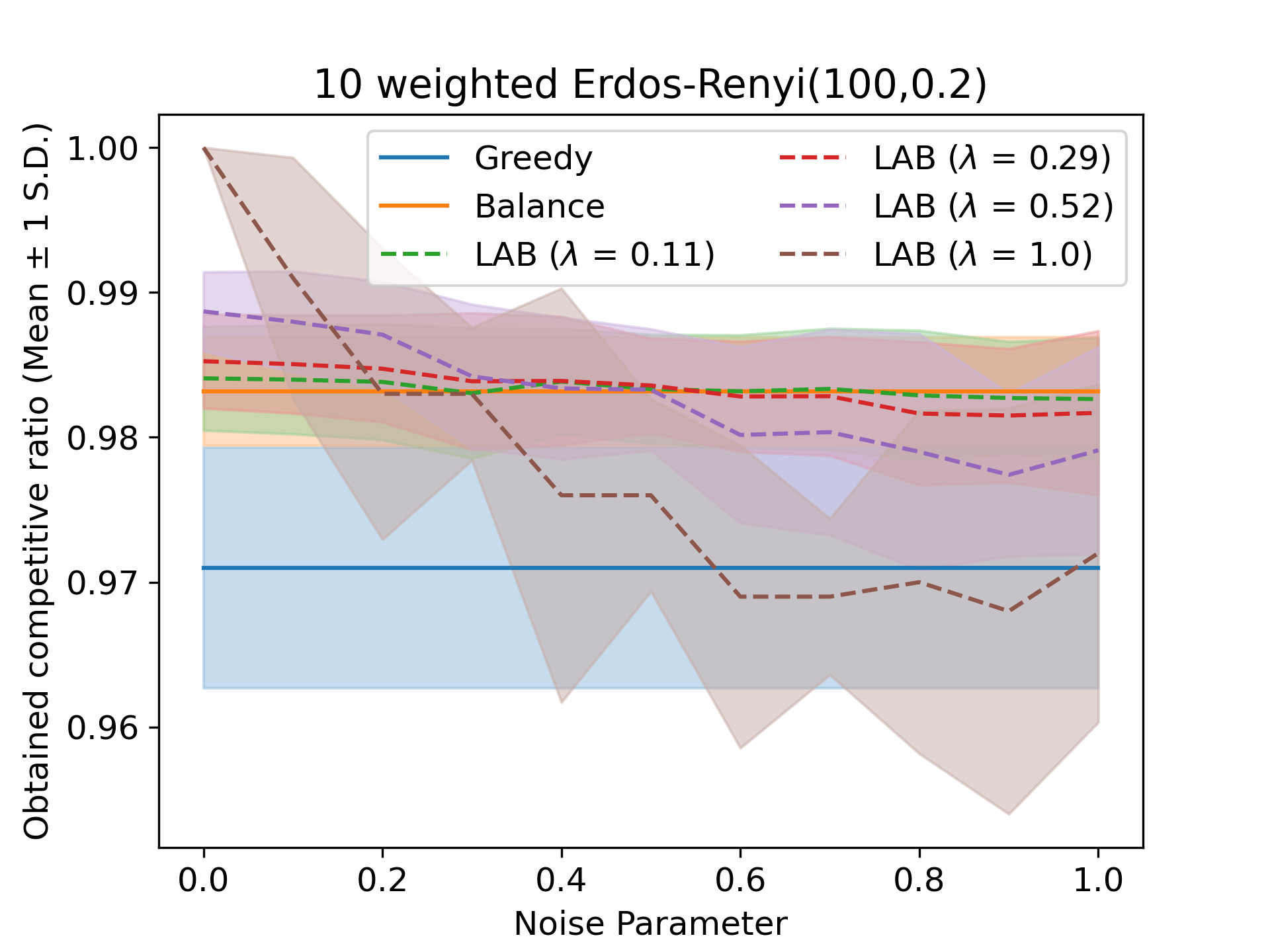}
    \end{subfigure}
    \begin{subfigure}{0.32\textwidth}
        \centering
        \includegraphics[width=\textwidth]{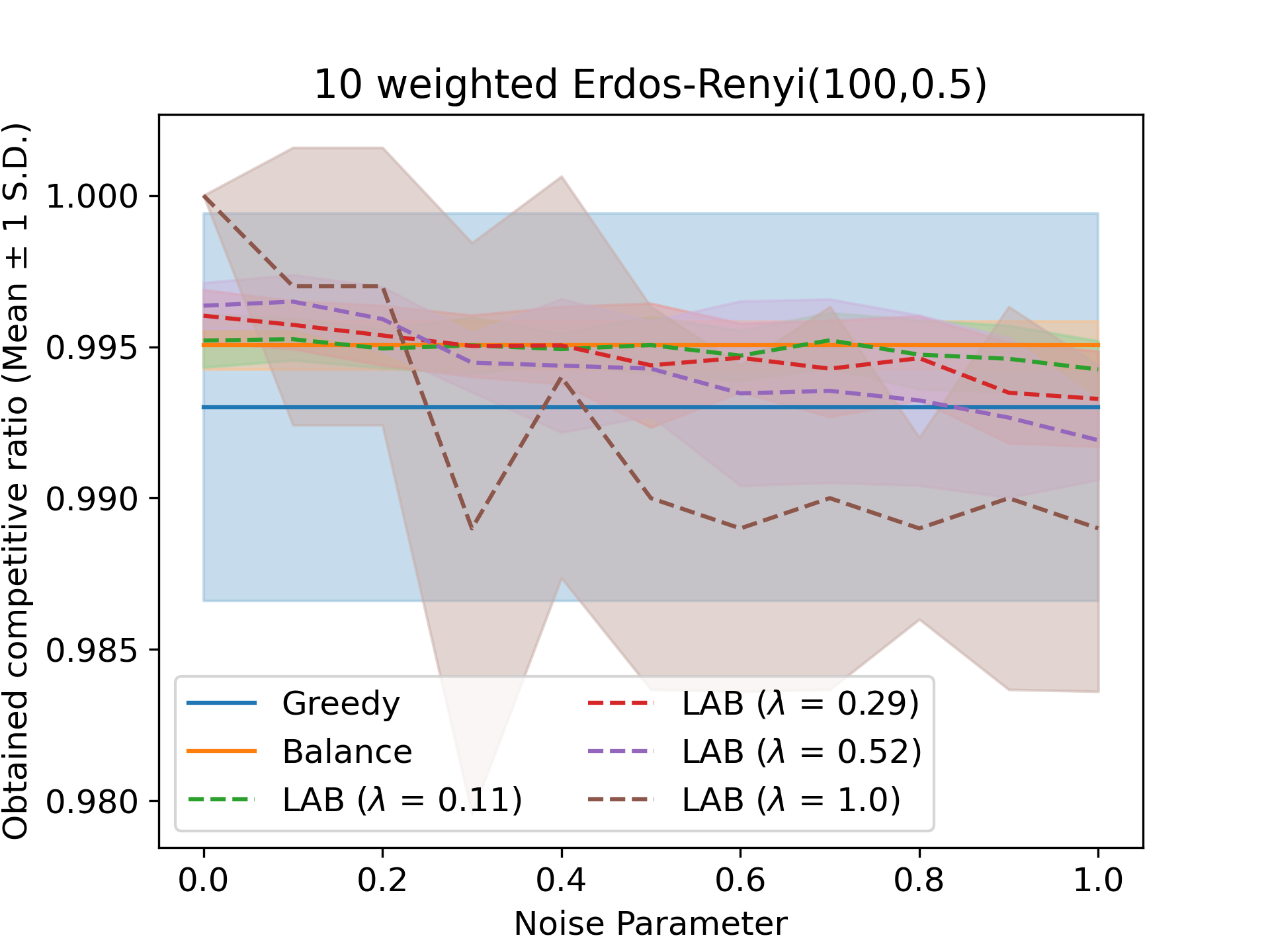}
    \end{subfigure}
    \\
    \vspace{10pt}
    \begin{subfigure}{0.32\textwidth}
        \centering
        \includegraphics[width=\textwidth]{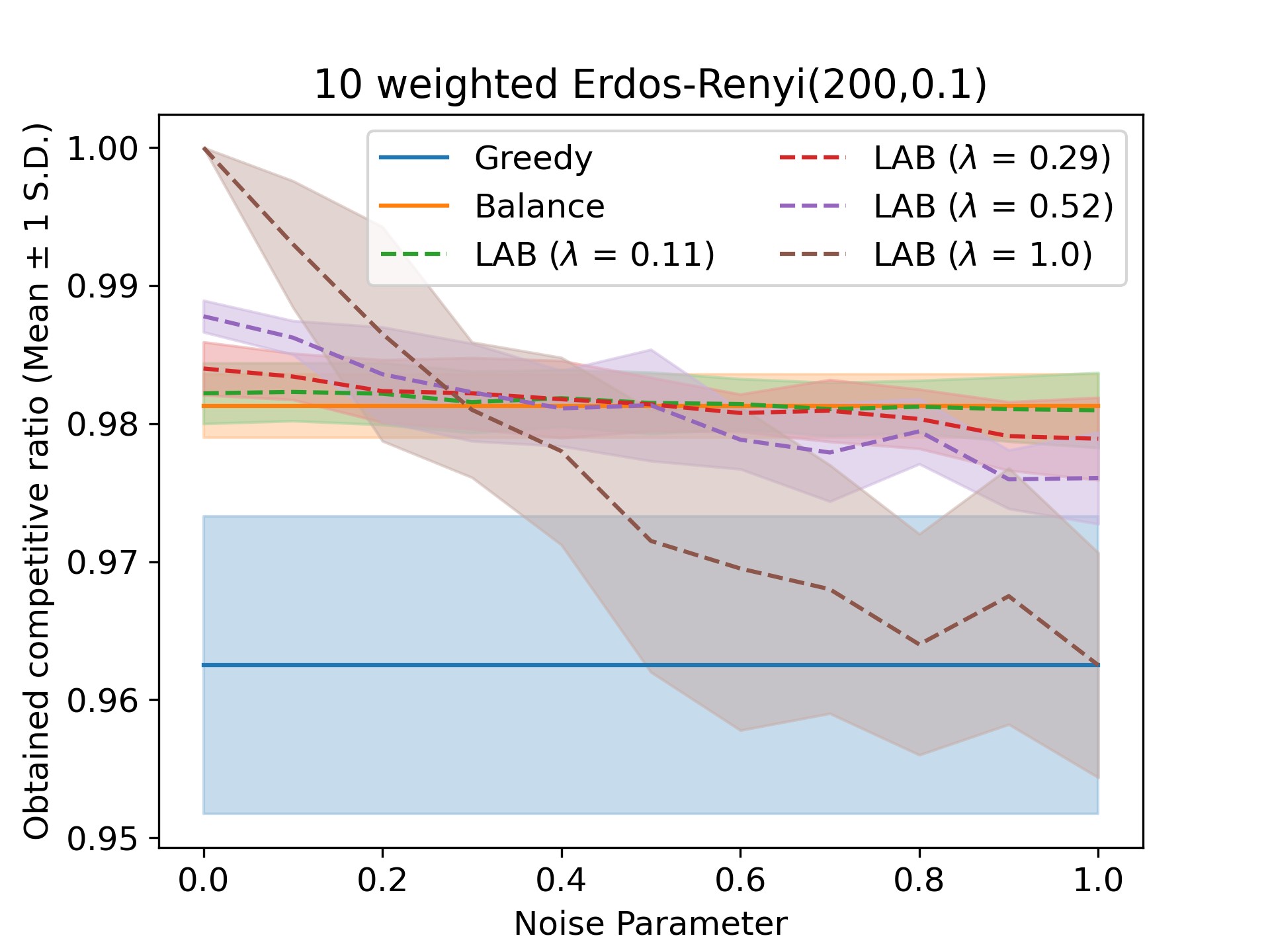}
    \end{subfigure}
    \begin{subfigure}{0.32\textwidth}
        \centering
        \includegraphics[width=\textwidth]{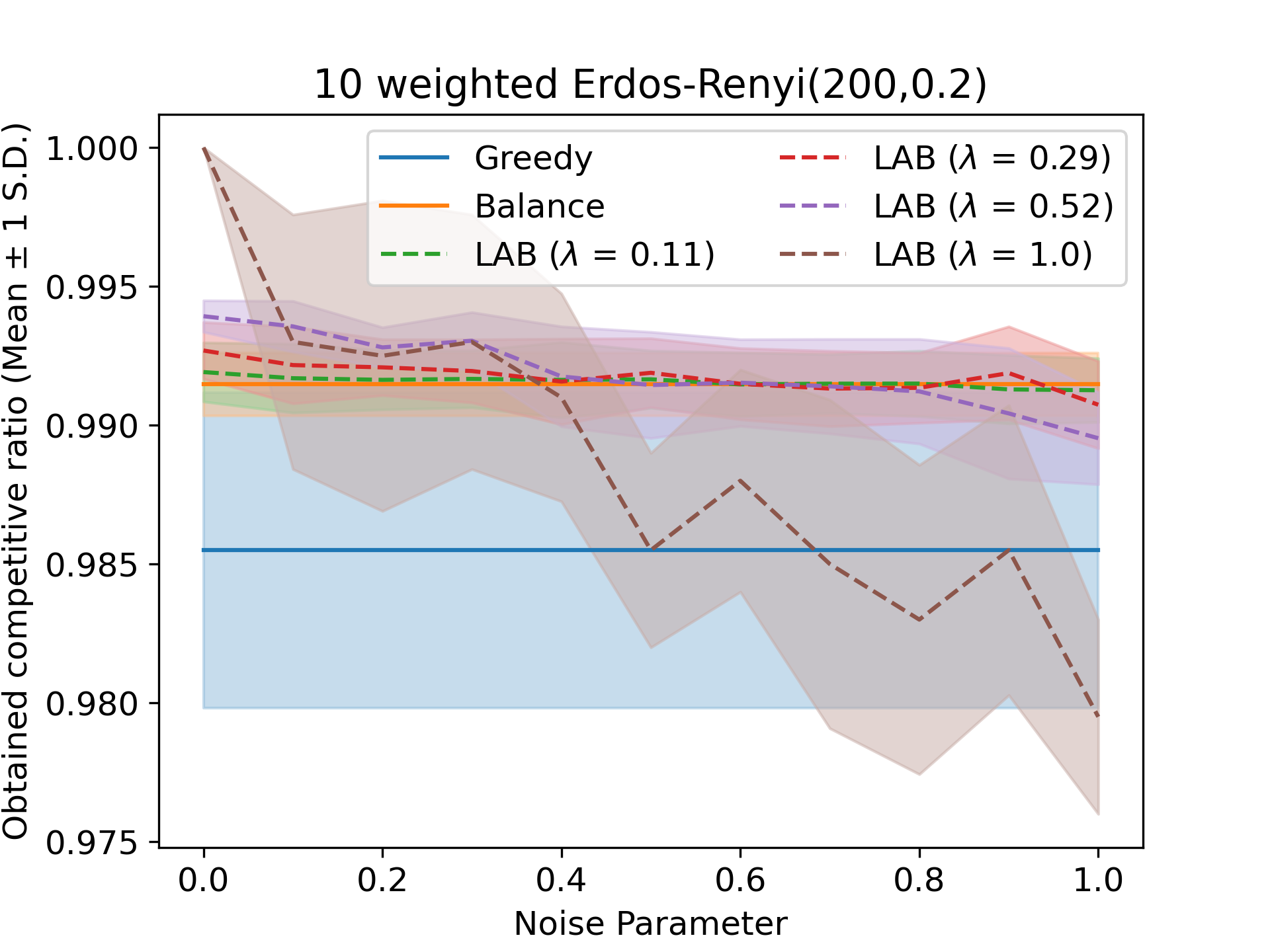}
    \end{subfigure}
    \begin{subfigure}{0.32\textwidth}
        \centering
        \includegraphics[width=\textwidth]{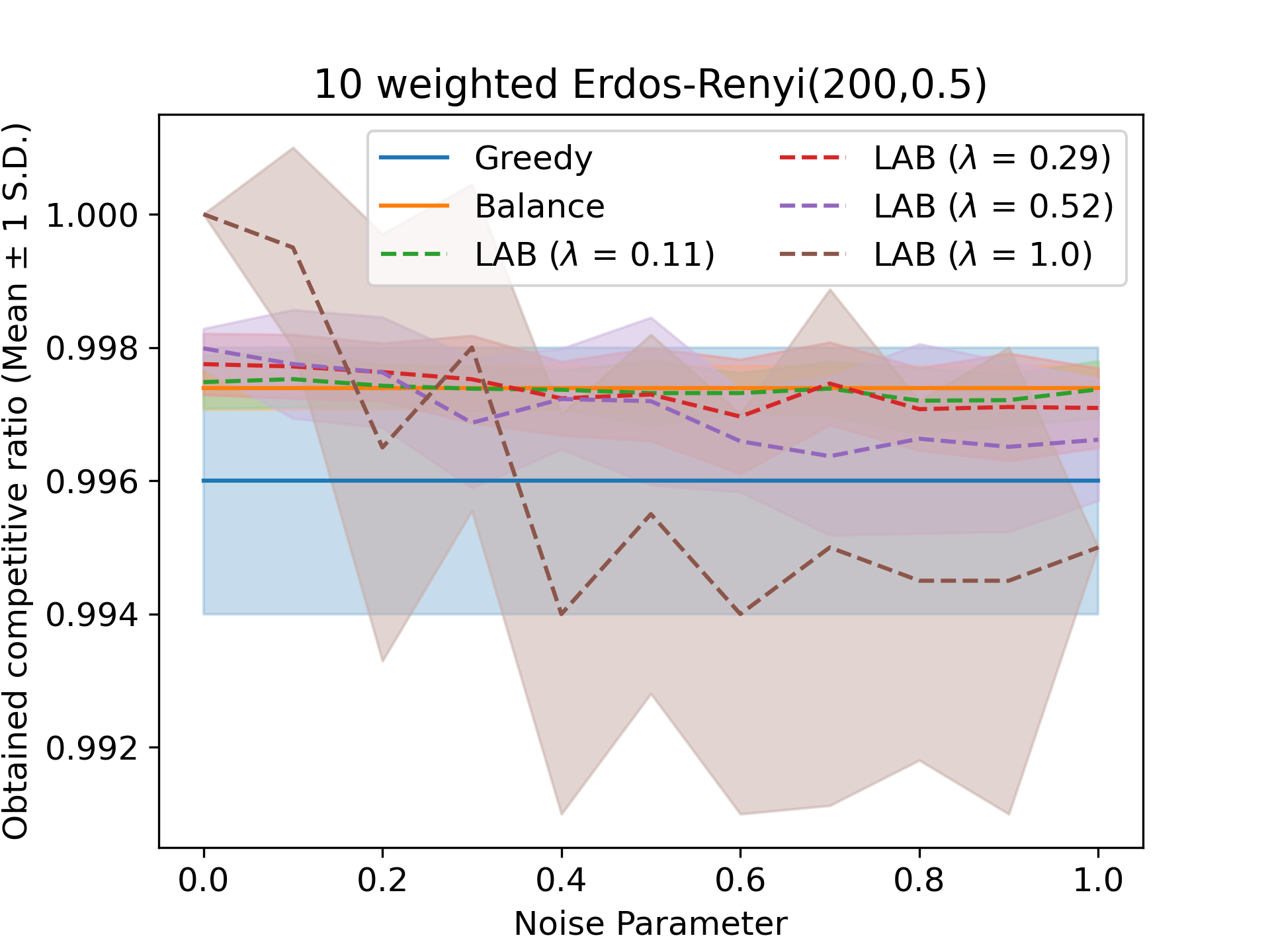}
    \end{subfigure}
    \\
    \vspace{10pt}
    \begin{subfigure}{0.32\textwidth}
        \centering
        \includegraphics[width=\textwidth]{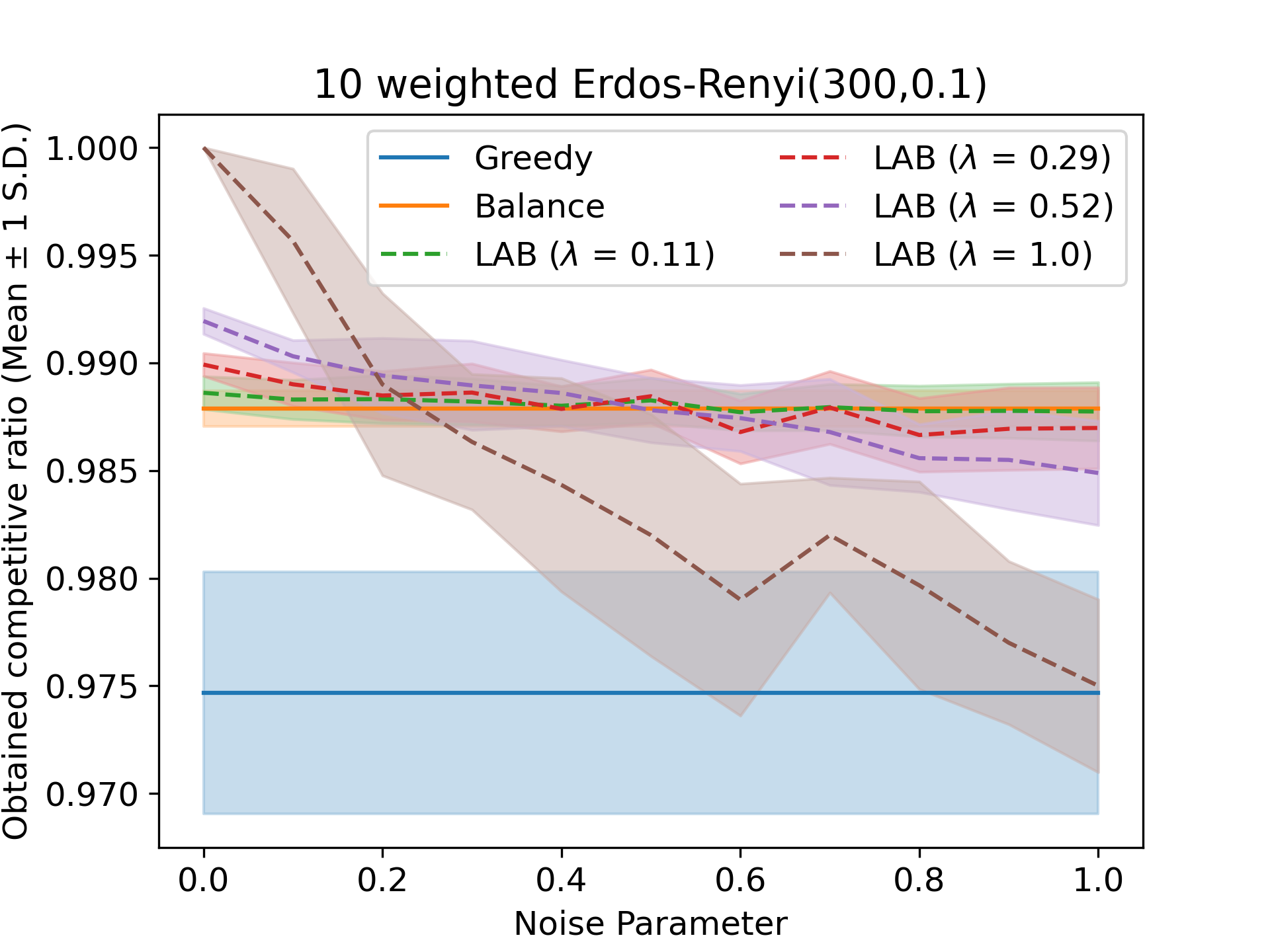}
    \end{subfigure}
    \begin{subfigure}{0.32\textwidth}
        \centering
        \includegraphics[width=\textwidth]{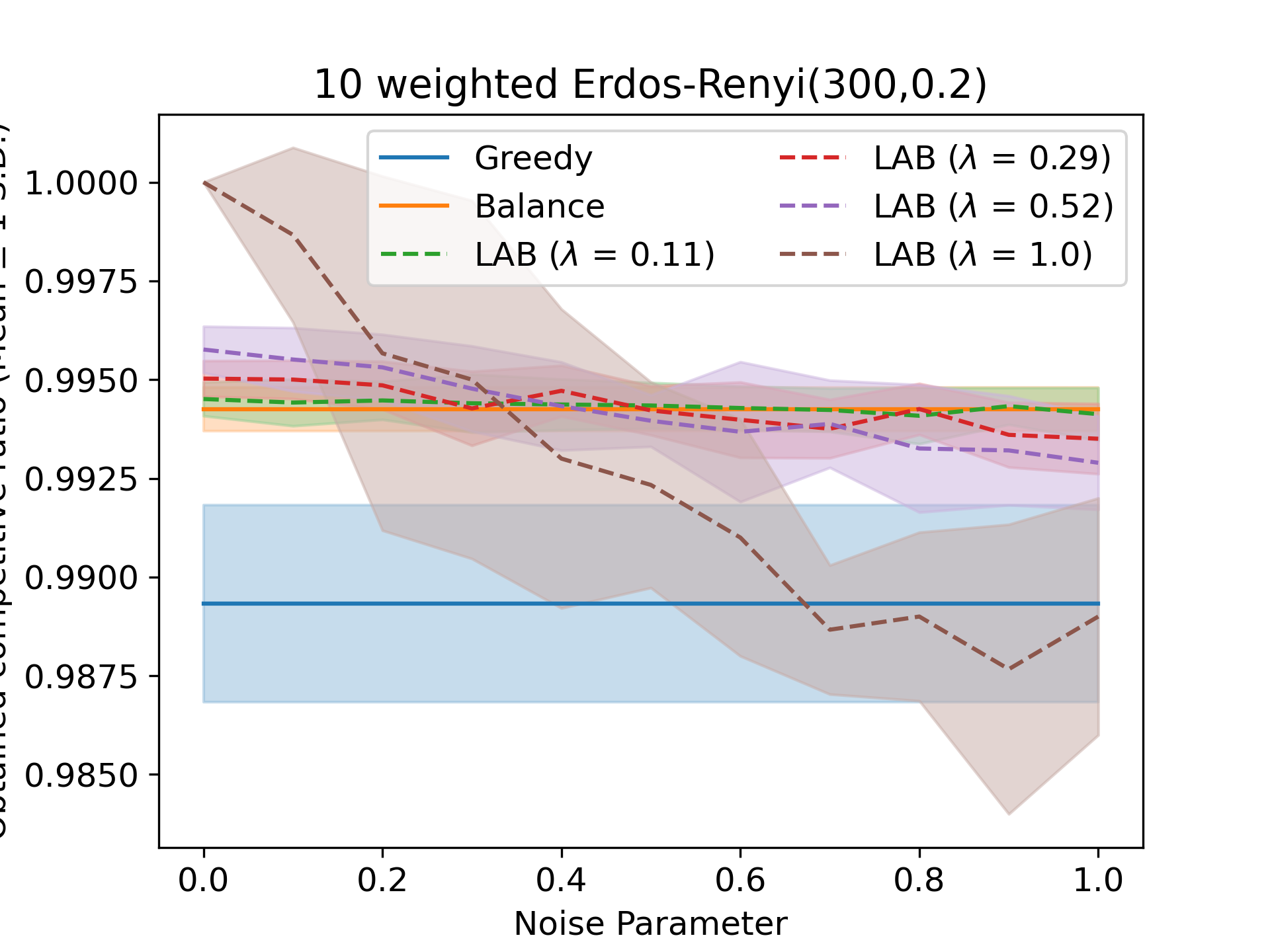}
    \end{subfigure}
    \begin{subfigure}{0.32\textwidth}
        \centering
        \includegraphics[width=\textwidth]{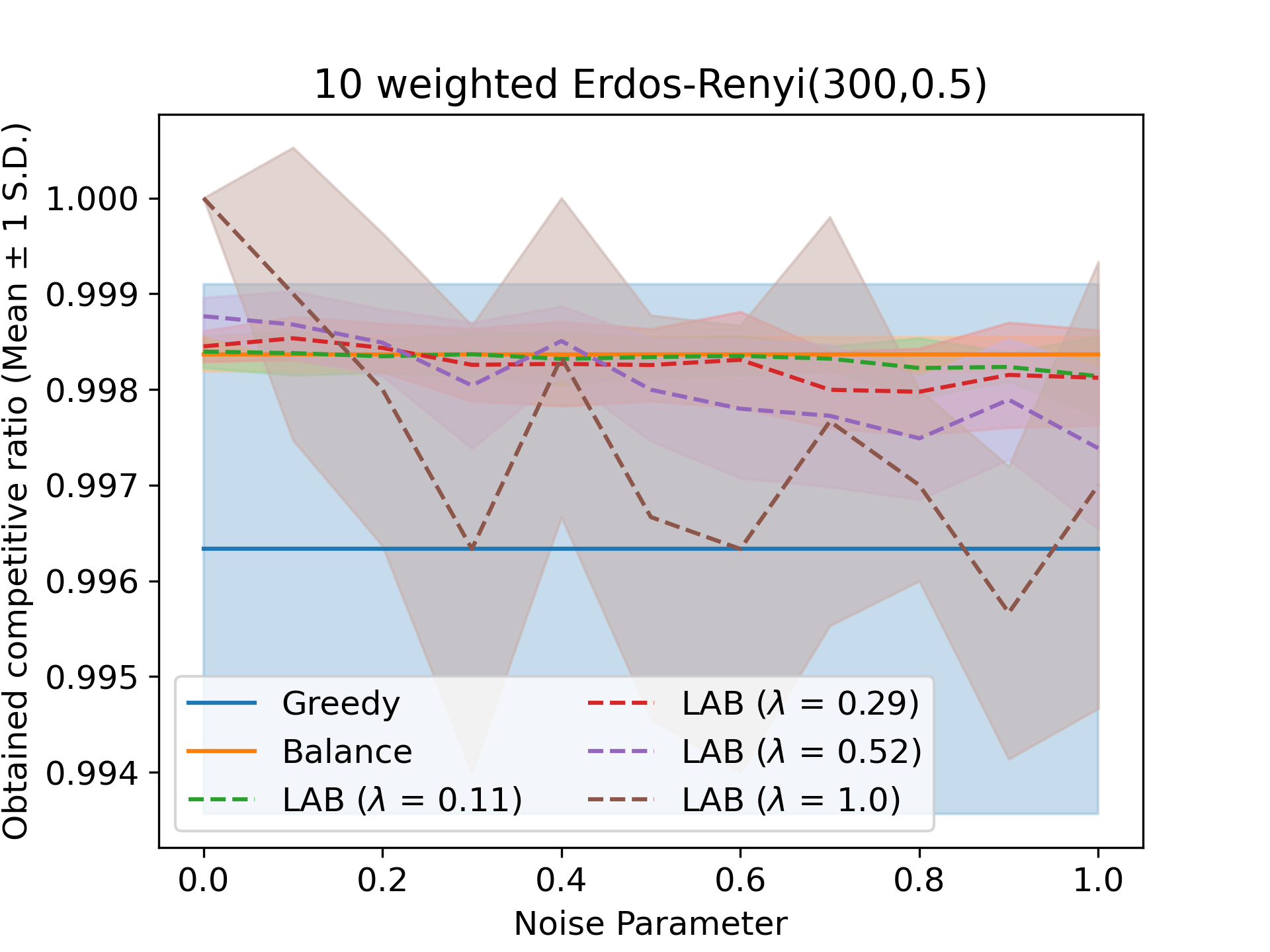}
    \end{subfigure}
    \caption{Empirical results for unweighted Erd\H{o}s-R\'enyi graph instances with $n \in \{100, 200, 300\}$ and $p \in \{0.1, 0.2, 0.5\}$}
    \label{fig:exp-ER-weighted}
\end{figure}

\begin{figure}[htb]
    \centering
    \begin{subfigure}{0.32\textwidth}
        \centering
        \includegraphics[width=\textwidth]{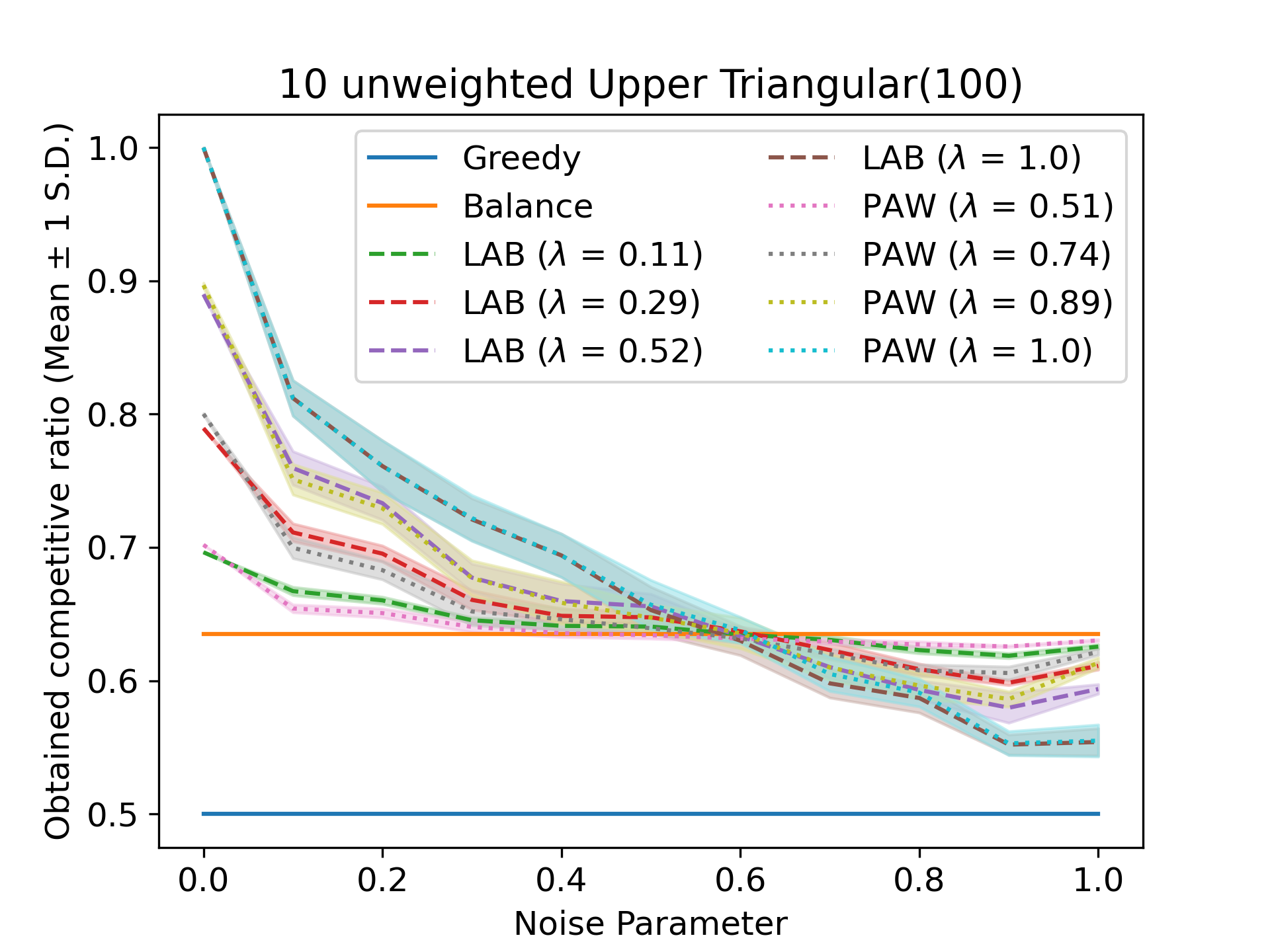}
    \end{subfigure}
    \begin{subfigure}{0.32\textwidth}
        \centering
        \includegraphics[width=\textwidth]{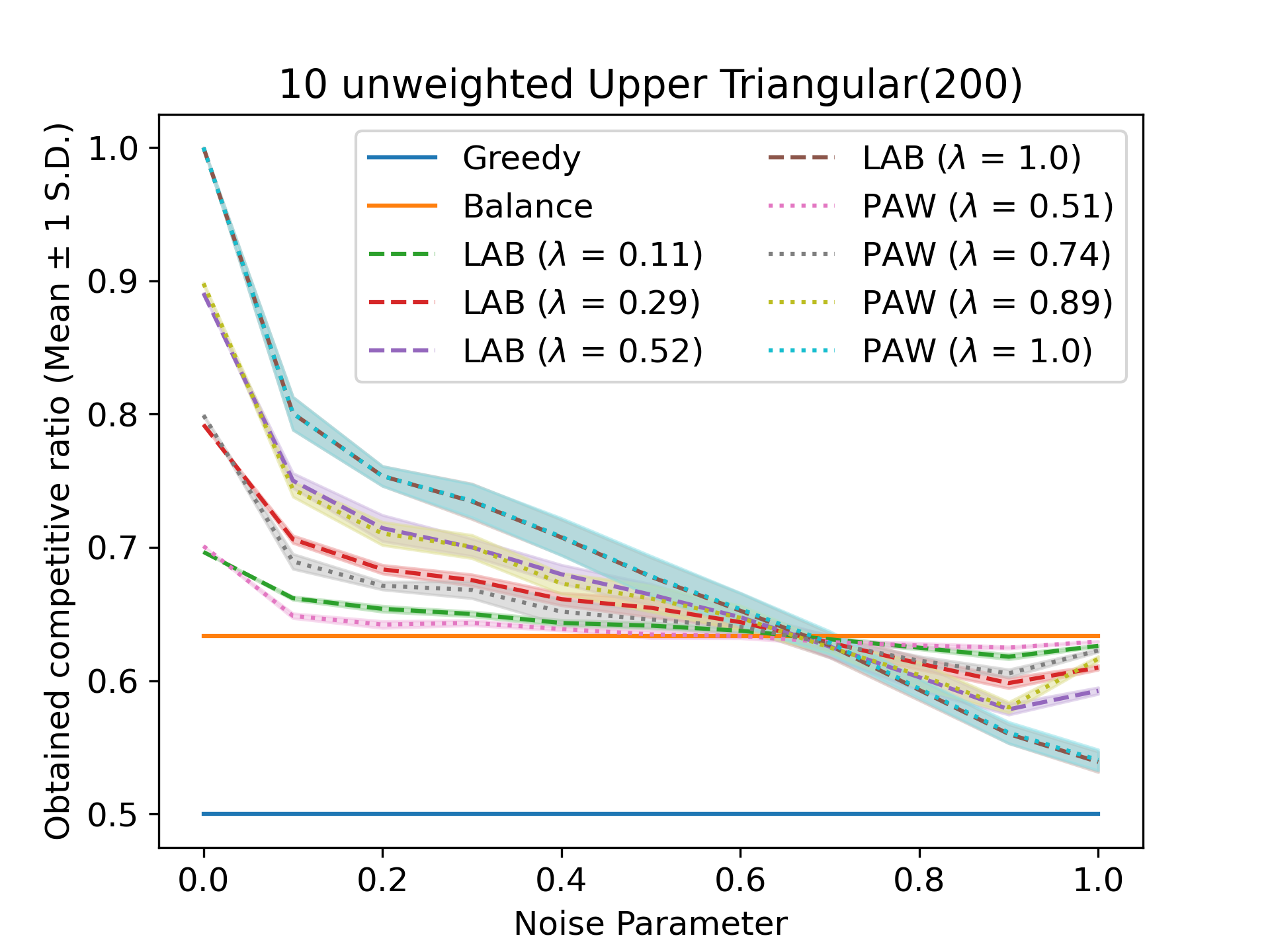}
    \end{subfigure}
    \begin{subfigure}{0.32\textwidth}
        \centering
        \includegraphics[width=\textwidth]{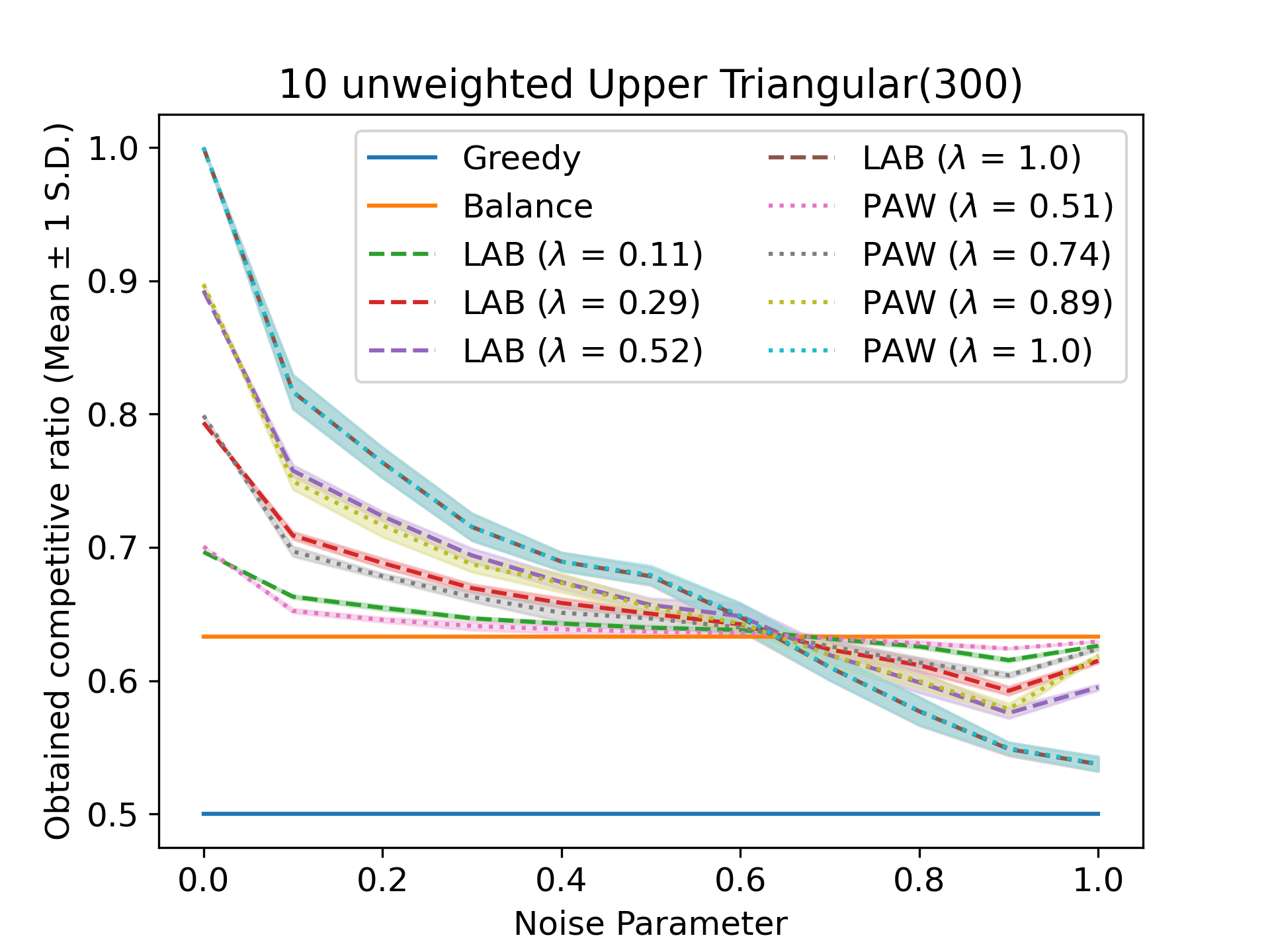}
    \end{subfigure}
    \caption{Empirical results for unweighted Upper Triangular graph instances with $n \in \{100, 200, 300\}$.}
    \label{fig:exp-UT-unweighted}
\end{figure}

\begin{figure}[htb]
    \centering
    \begin{subfigure}{0.32\textwidth}
        \centering
        \includegraphics[width=\textwidth]{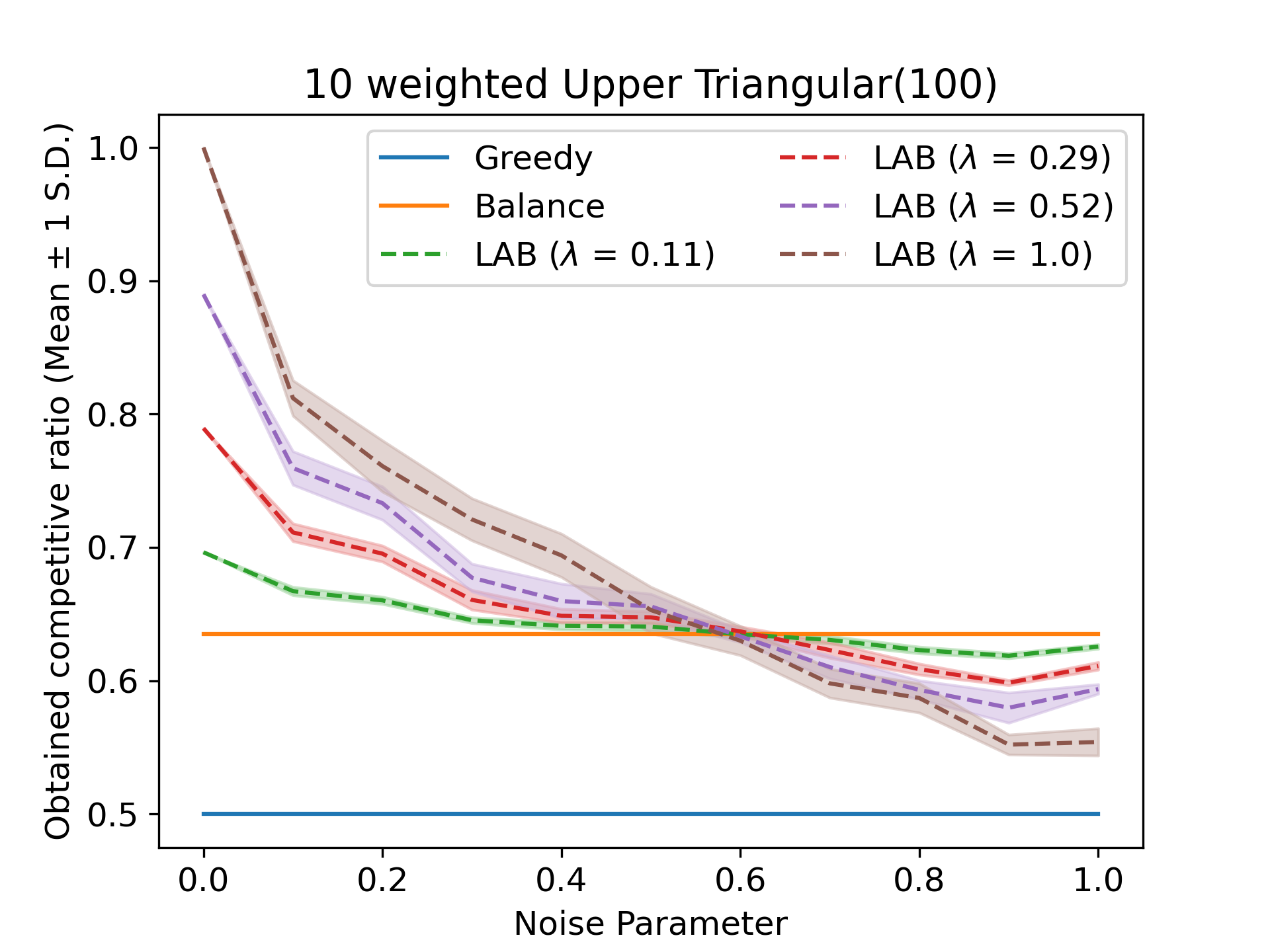}
    \end{subfigure}
    \begin{subfigure}{0.32\textwidth}
        \centering
        \includegraphics[width=\textwidth]{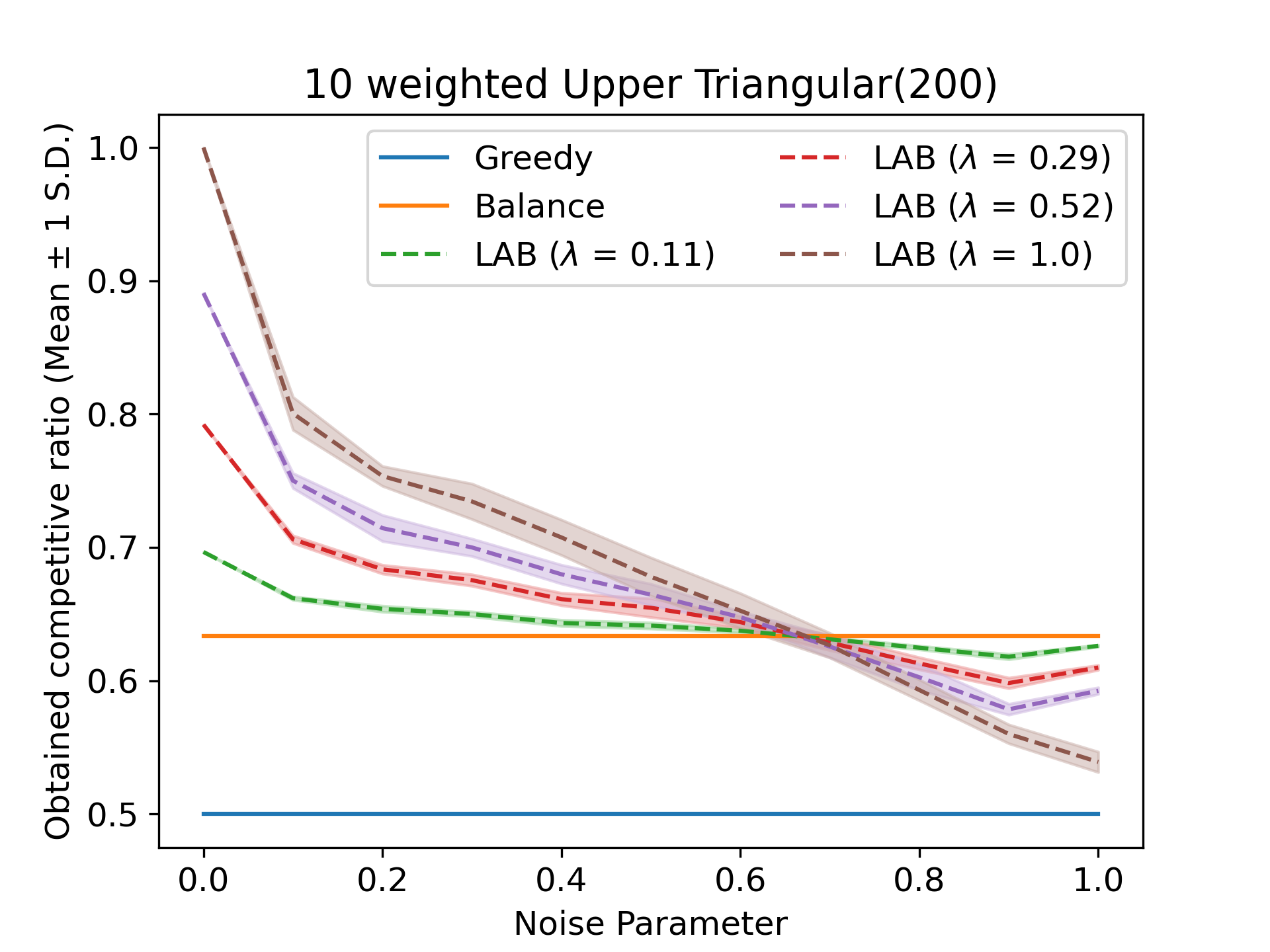}
    \end{subfigure}
    \begin{subfigure}{0.32\textwidth}
        \centering
        \includegraphics[width=\textwidth]{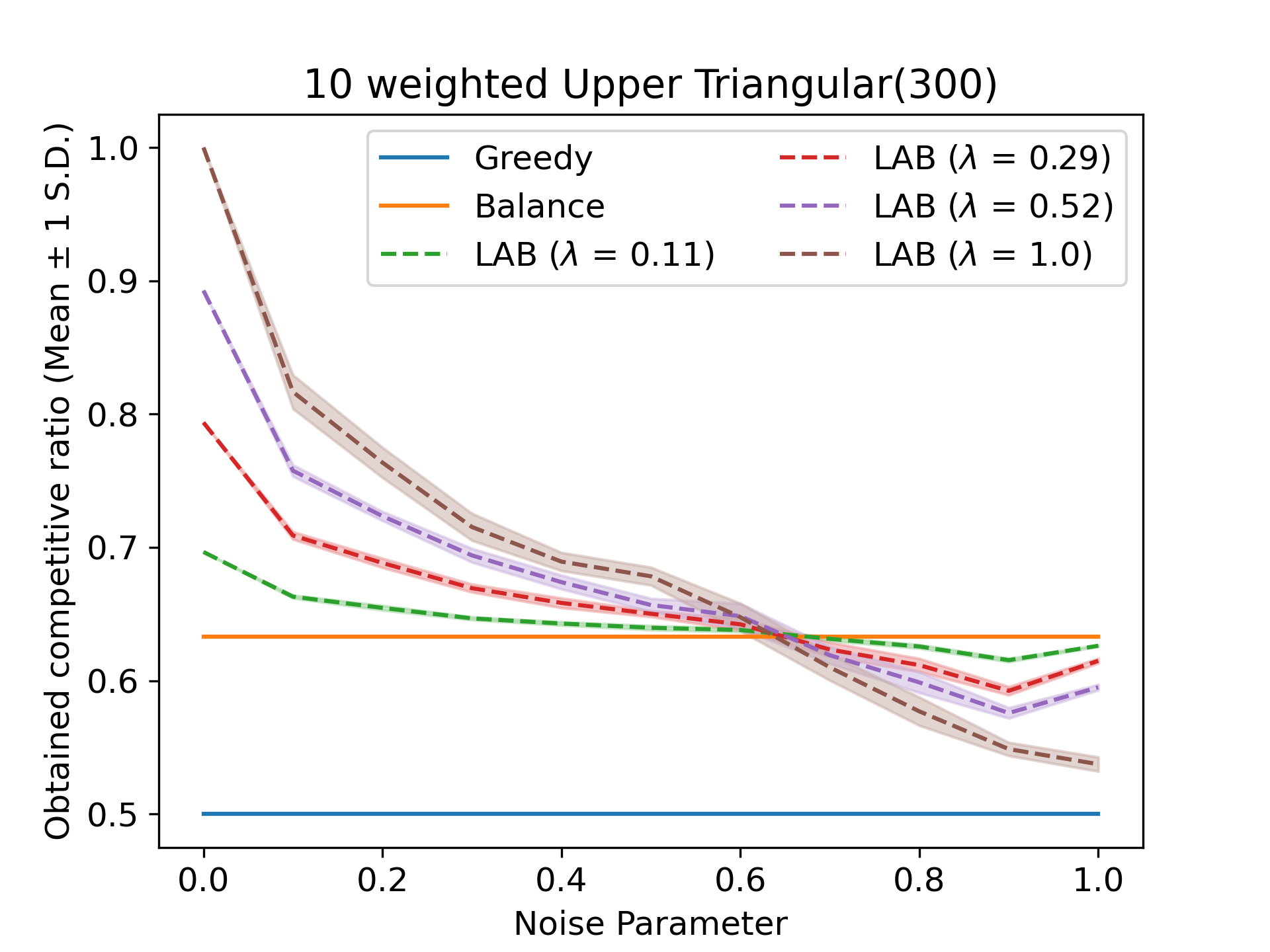}
    \end{subfigure}
    \caption{Empirical results for weighted Upper Triangular graph instances with $n \in \{100, 200, 300\}$.}
    \label{fig:exp-UT-weighted}
\end{figure}

\begin{figure}[htb]
    \centering
    \begin{subfigure}{0.32\textwidth}
        \centering
        \includegraphics[width=\textwidth]{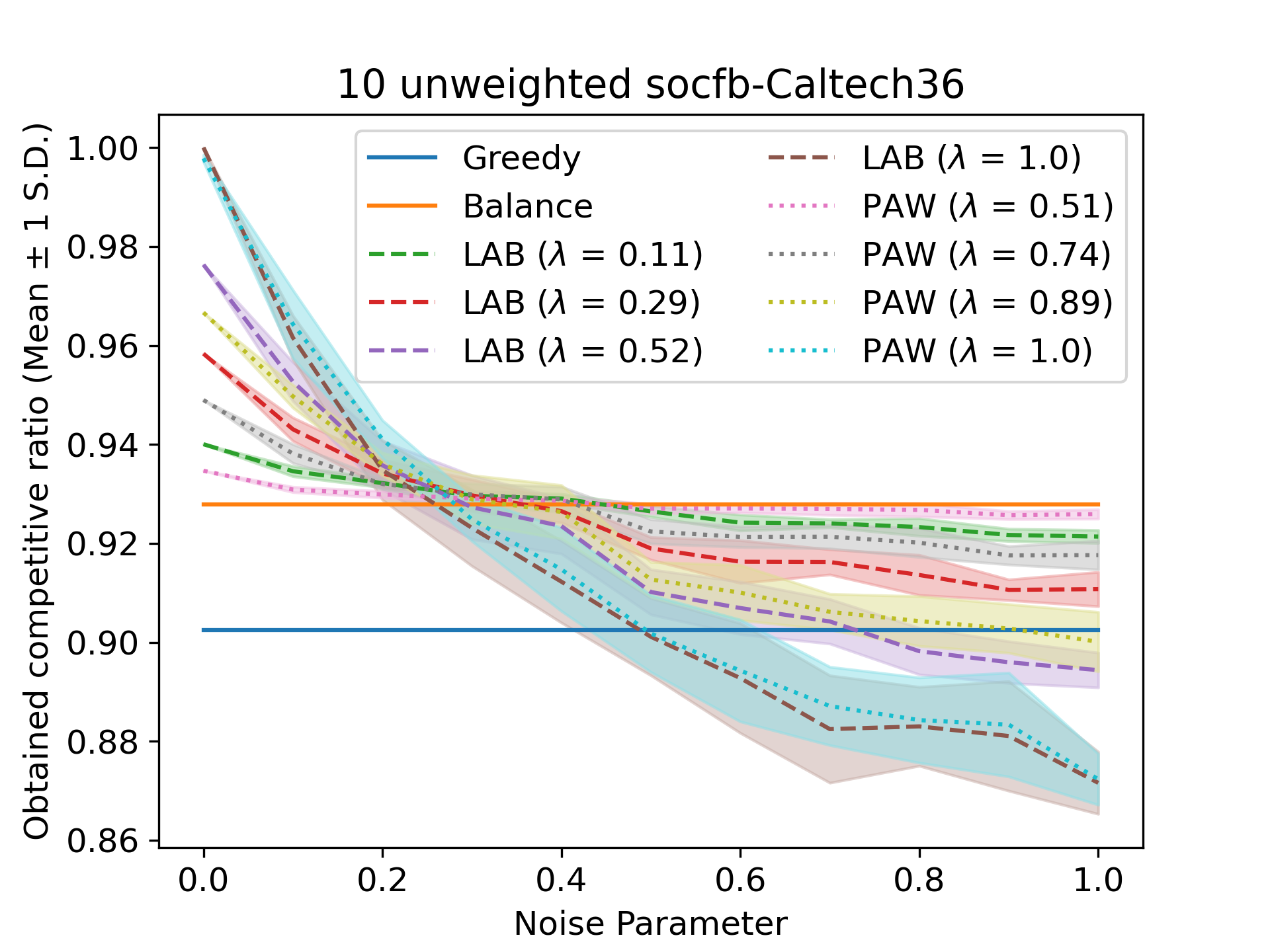}
    \end{subfigure}
    \begin{subfigure}{0.32\textwidth}
        \centering
        \includegraphics[width=\textwidth]{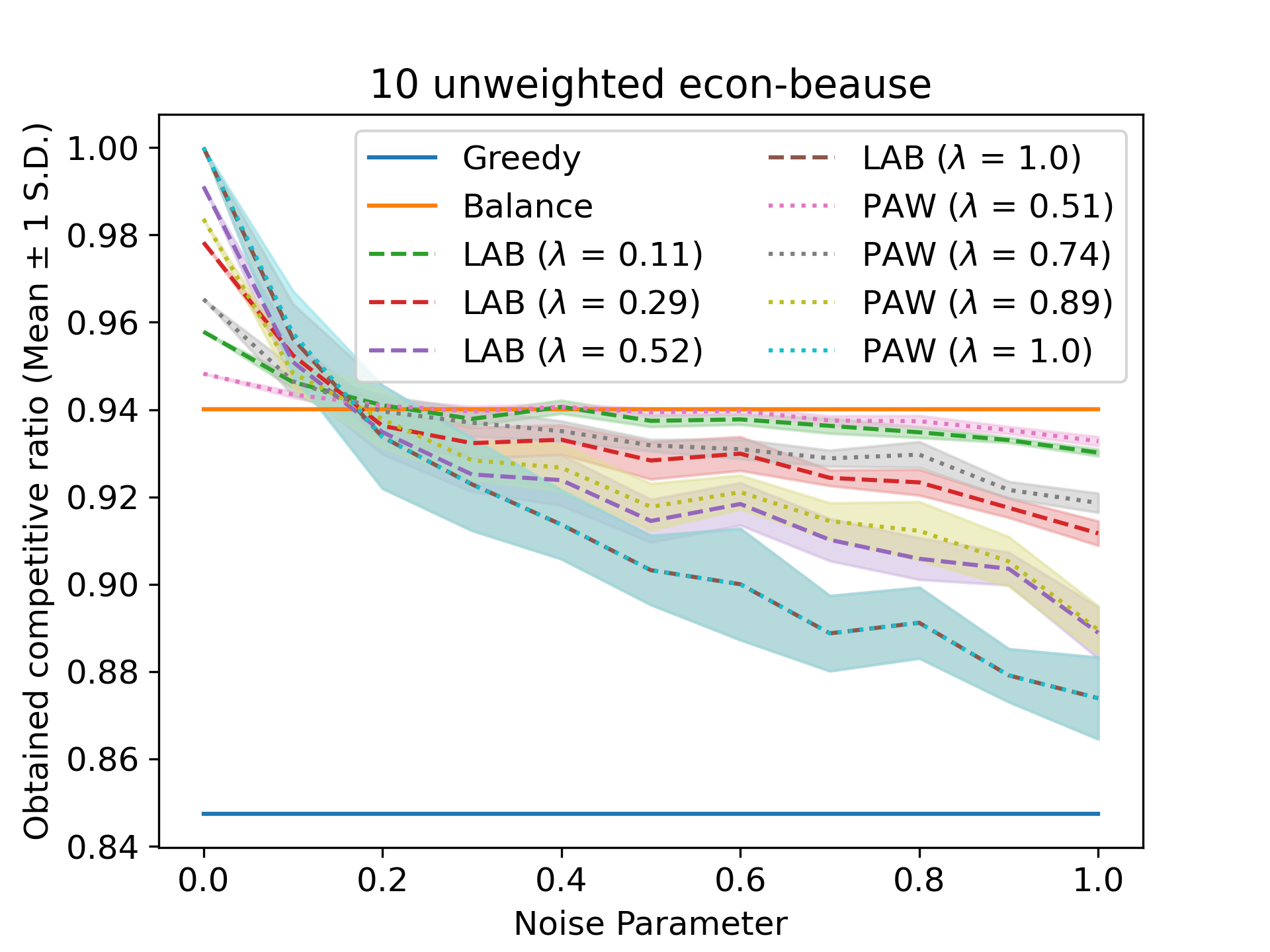}
    \end{subfigure}
    \begin{subfigure}{0.32\textwidth}
        \centering
        \includegraphics[width=\textwidth]{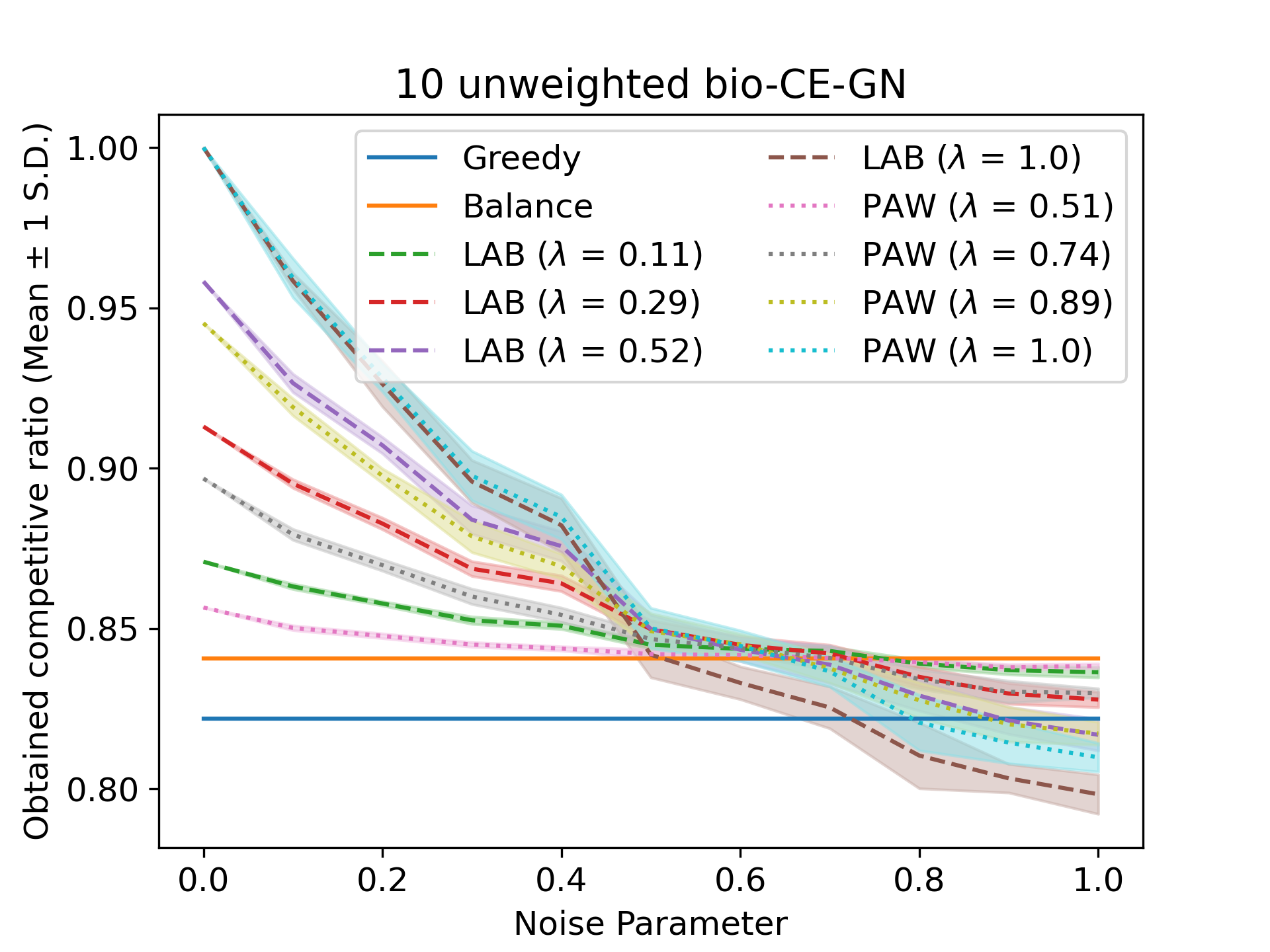}
    \end{subfigure}
    \\
    \begin{subfigure}{0.32\textwidth}
        \centering
        \includegraphics[width=\textwidth]{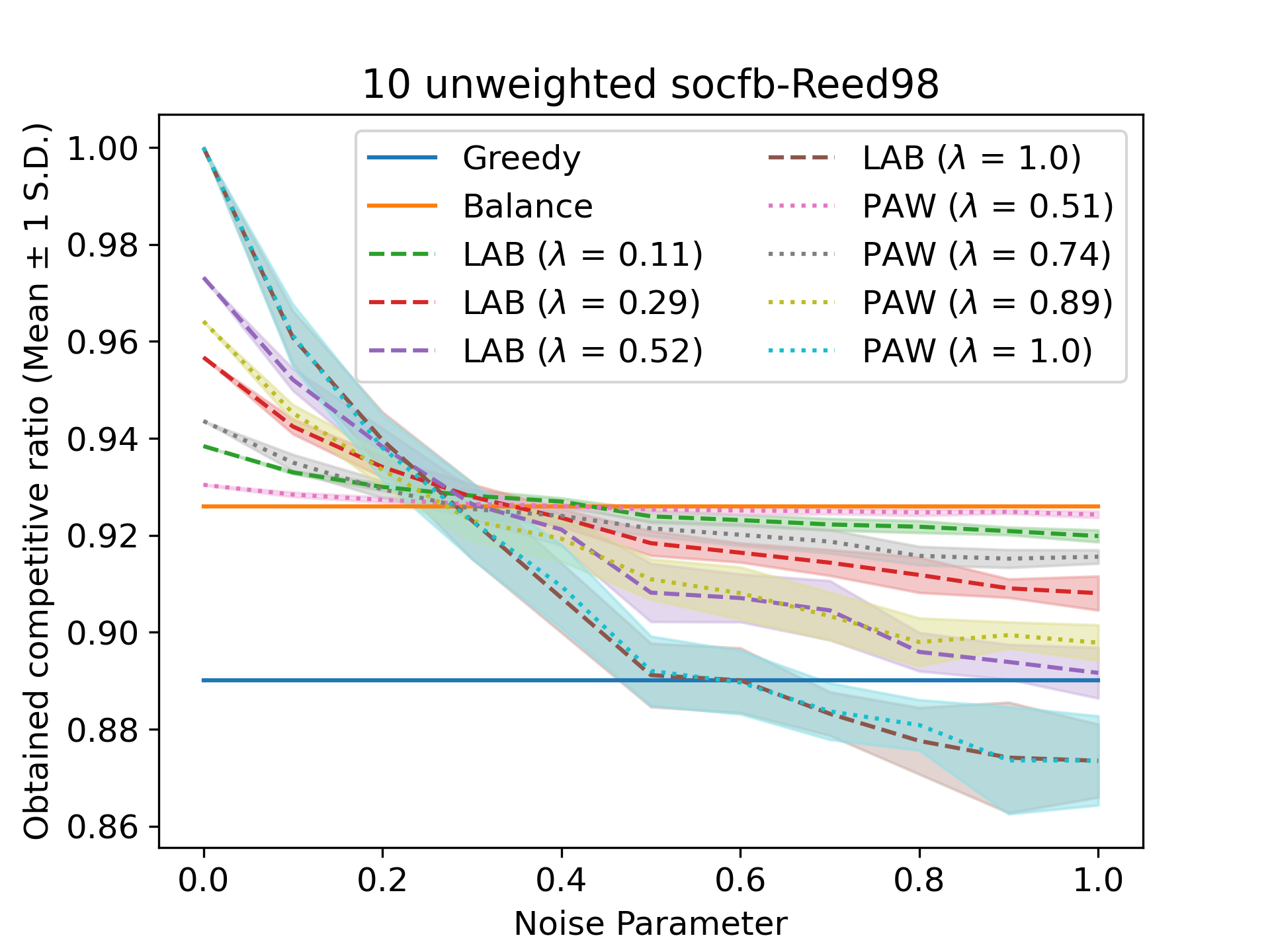}
    \end{subfigure}
    \begin{subfigure}{0.32\textwidth}
        \centering
        \includegraphics[width=\textwidth]{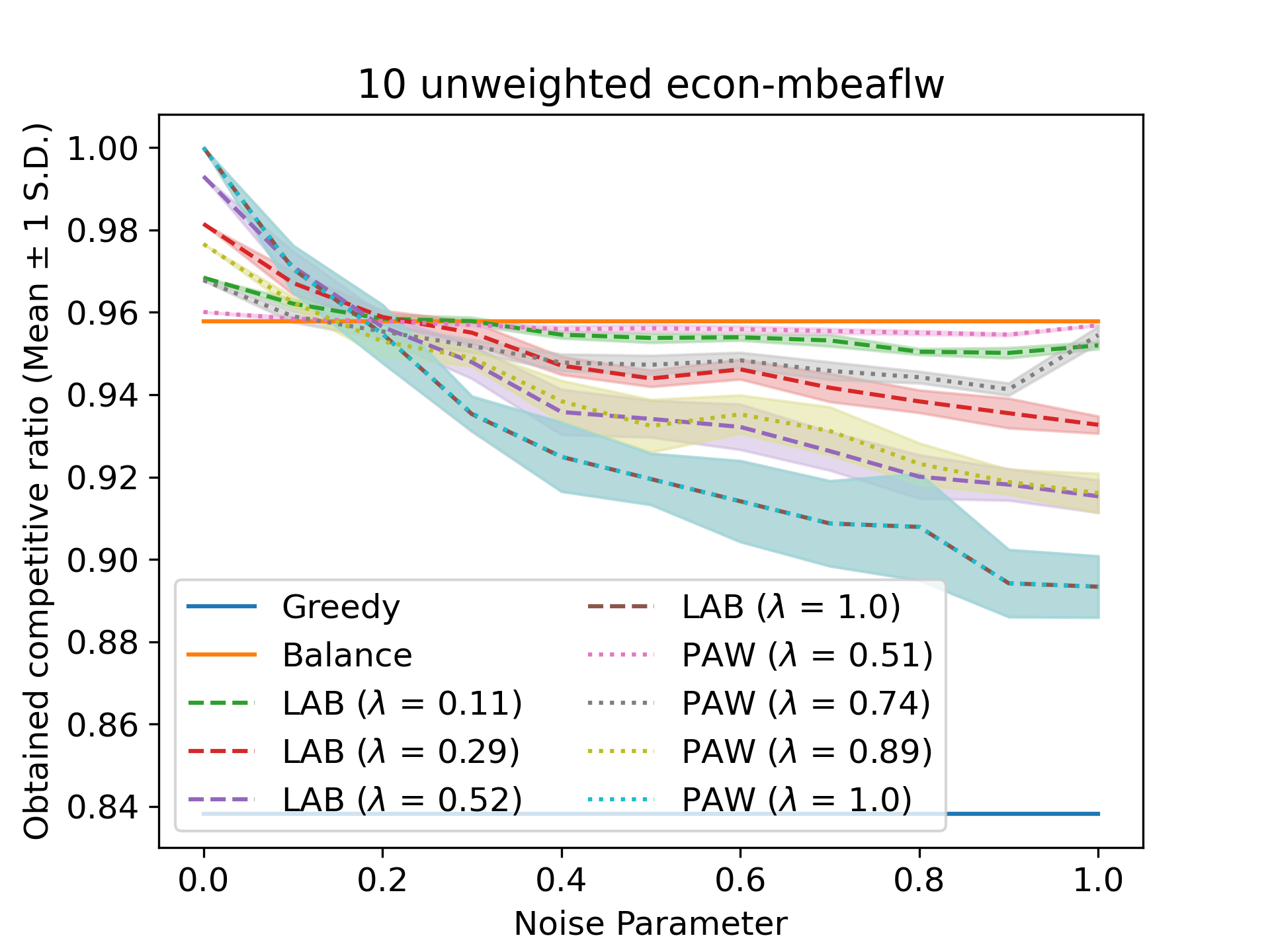}
    \end{subfigure}
    \begin{subfigure}{0.32\textwidth}
        \centering
        \includegraphics[width=\textwidth]{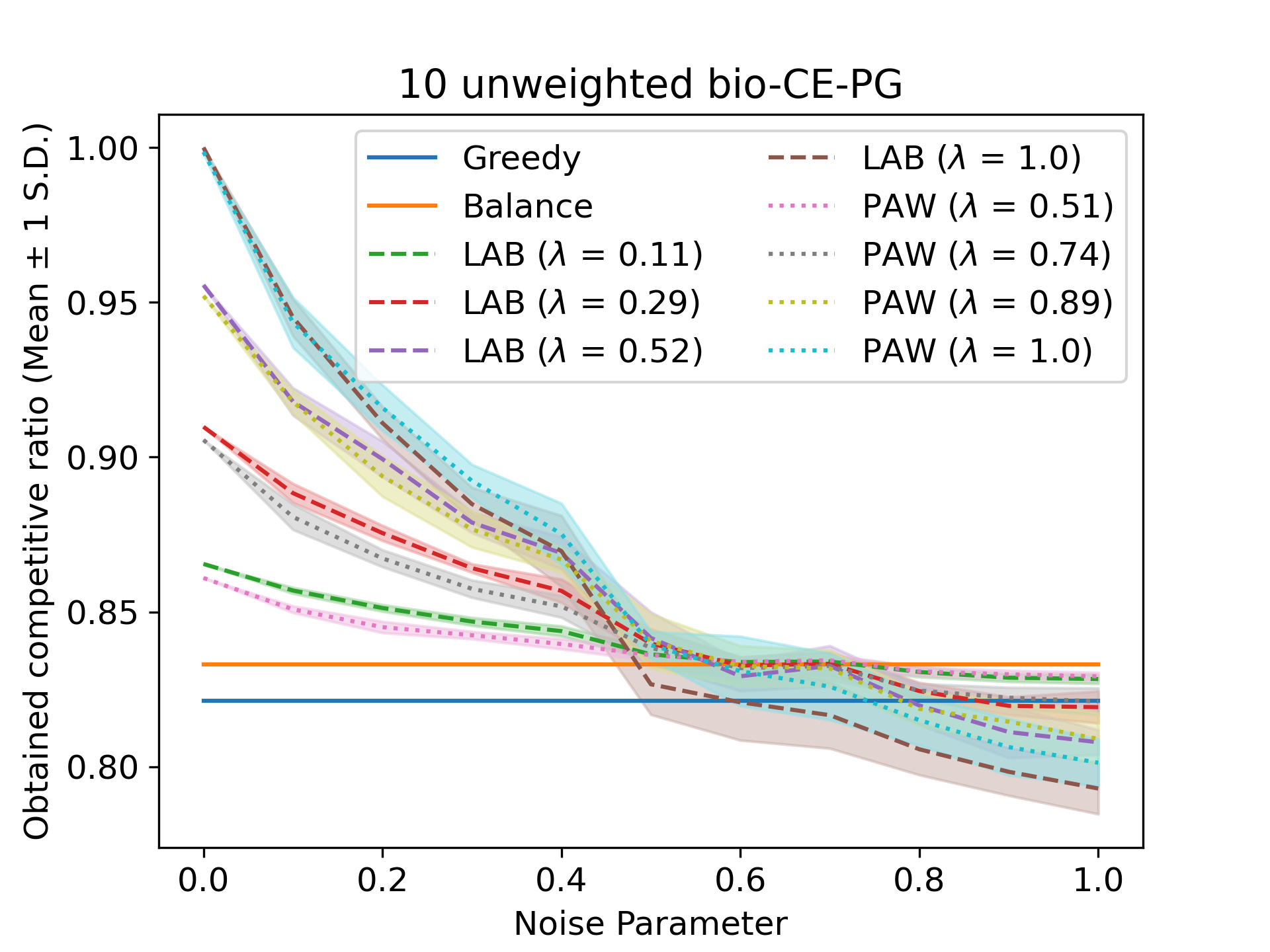}
    \end{subfigure}
    \caption{Empirical results for unweighted real-world graph instances.}
    \label{fig:exp-real-unweighted}
\end{figure}

\begin{figure}[htb]
    \centering
    \begin{subfigure}{0.32\textwidth}
        \centering
        \includegraphics[width=\textwidth]{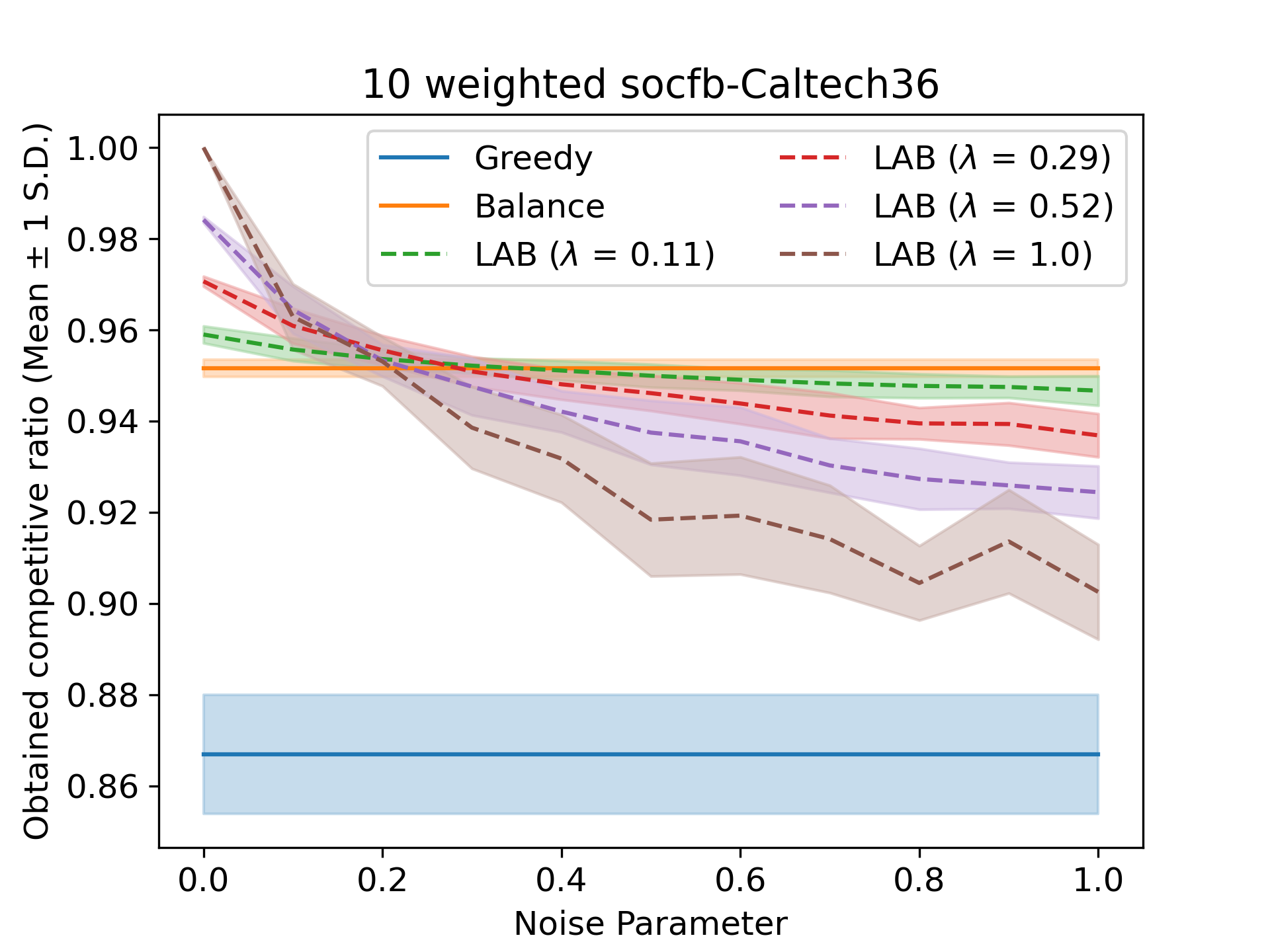}
    \end{subfigure}
    \begin{subfigure}{0.32\textwidth}
        \centering
        \includegraphics[width=\textwidth]{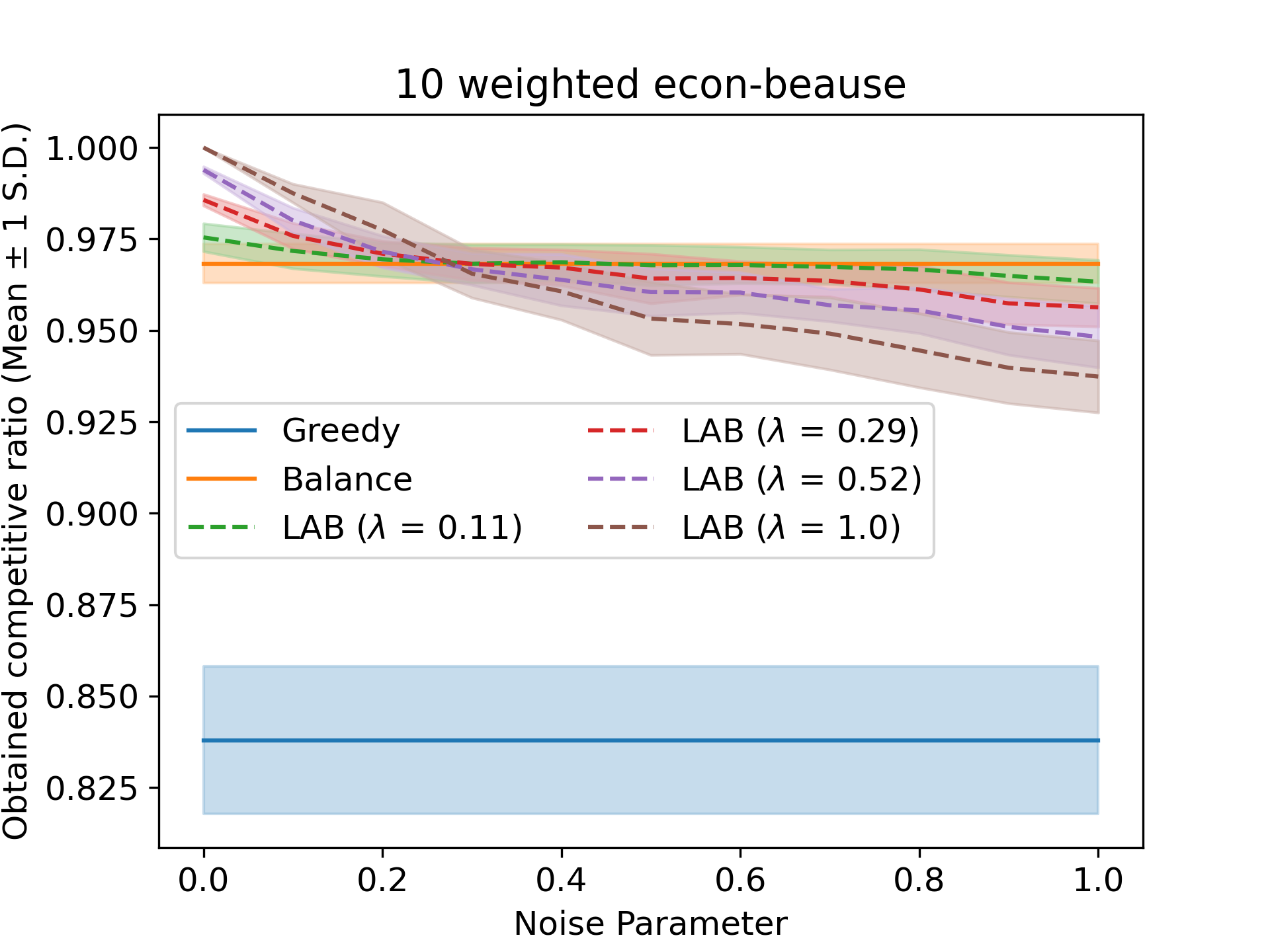}
    \end{subfigure}
    \begin{subfigure}{0.32\textwidth}
        \centering
        \includegraphics[width=\textwidth]{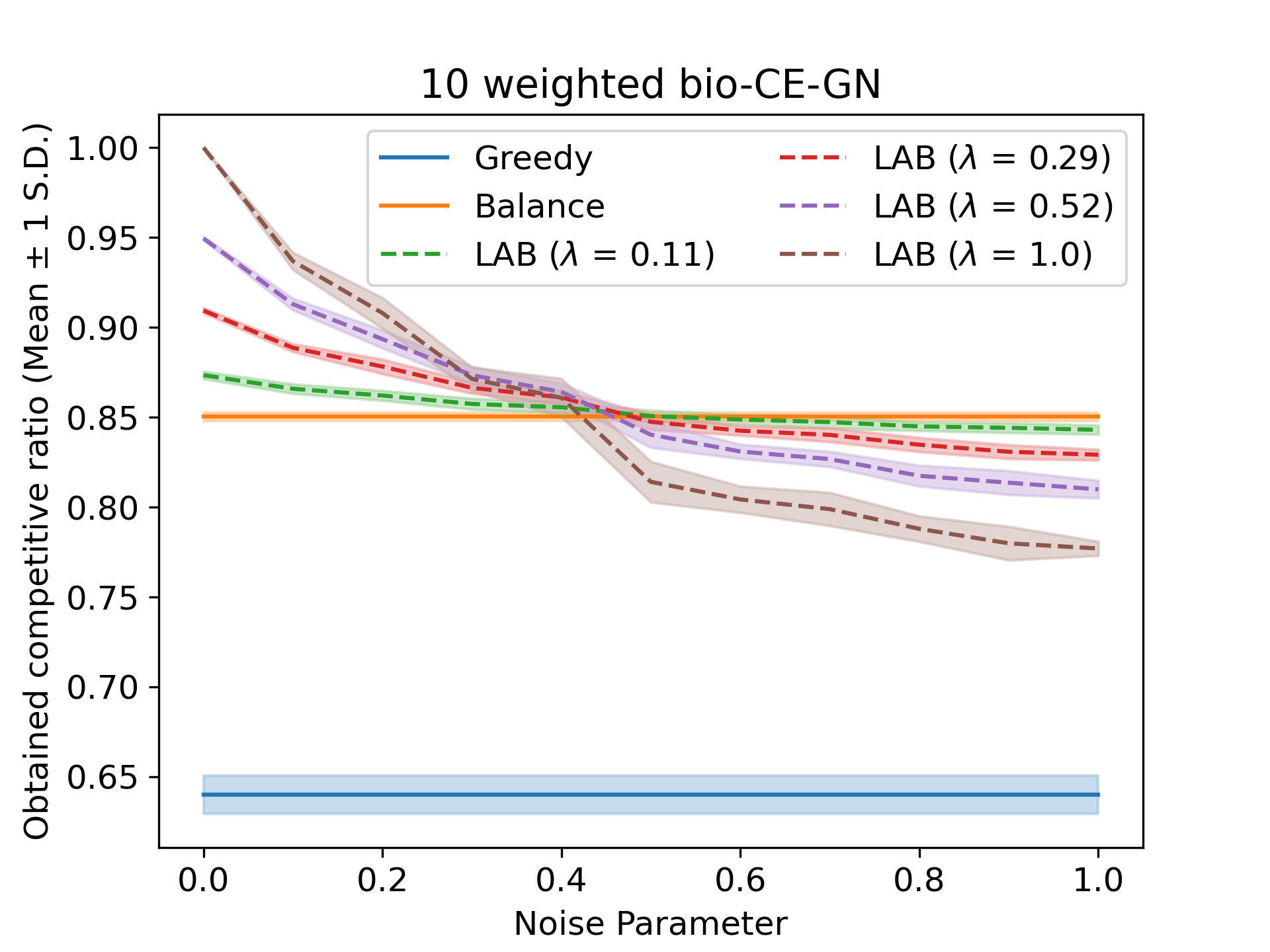}
    \end{subfigure}
    \\
    \begin{subfigure}{0.32\textwidth}
        \centering
        \includegraphics[width=\textwidth]{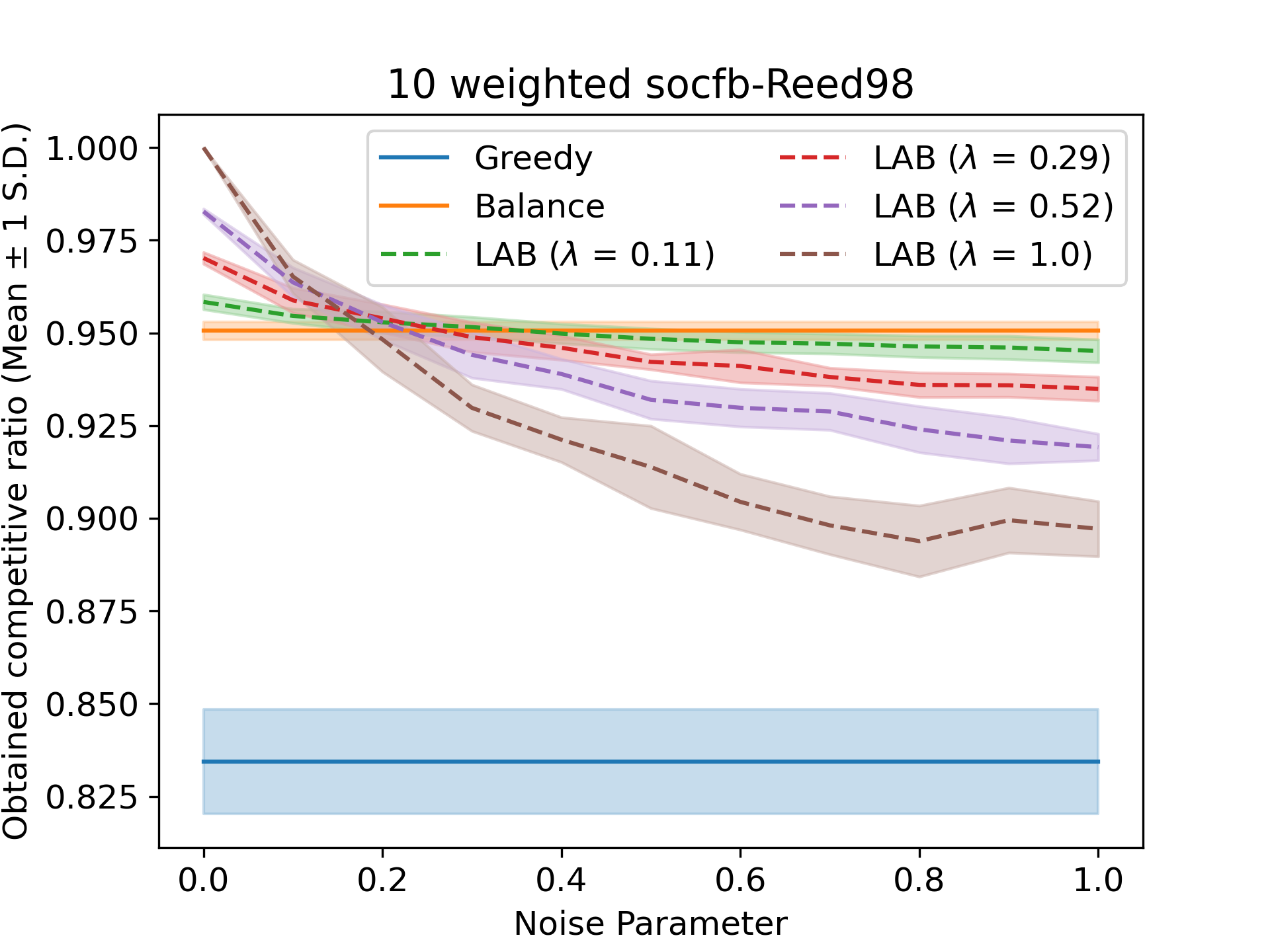}
    \end{subfigure}
    \begin{subfigure}{0.32\textwidth}
        \centering
        \includegraphics[width=\textwidth]{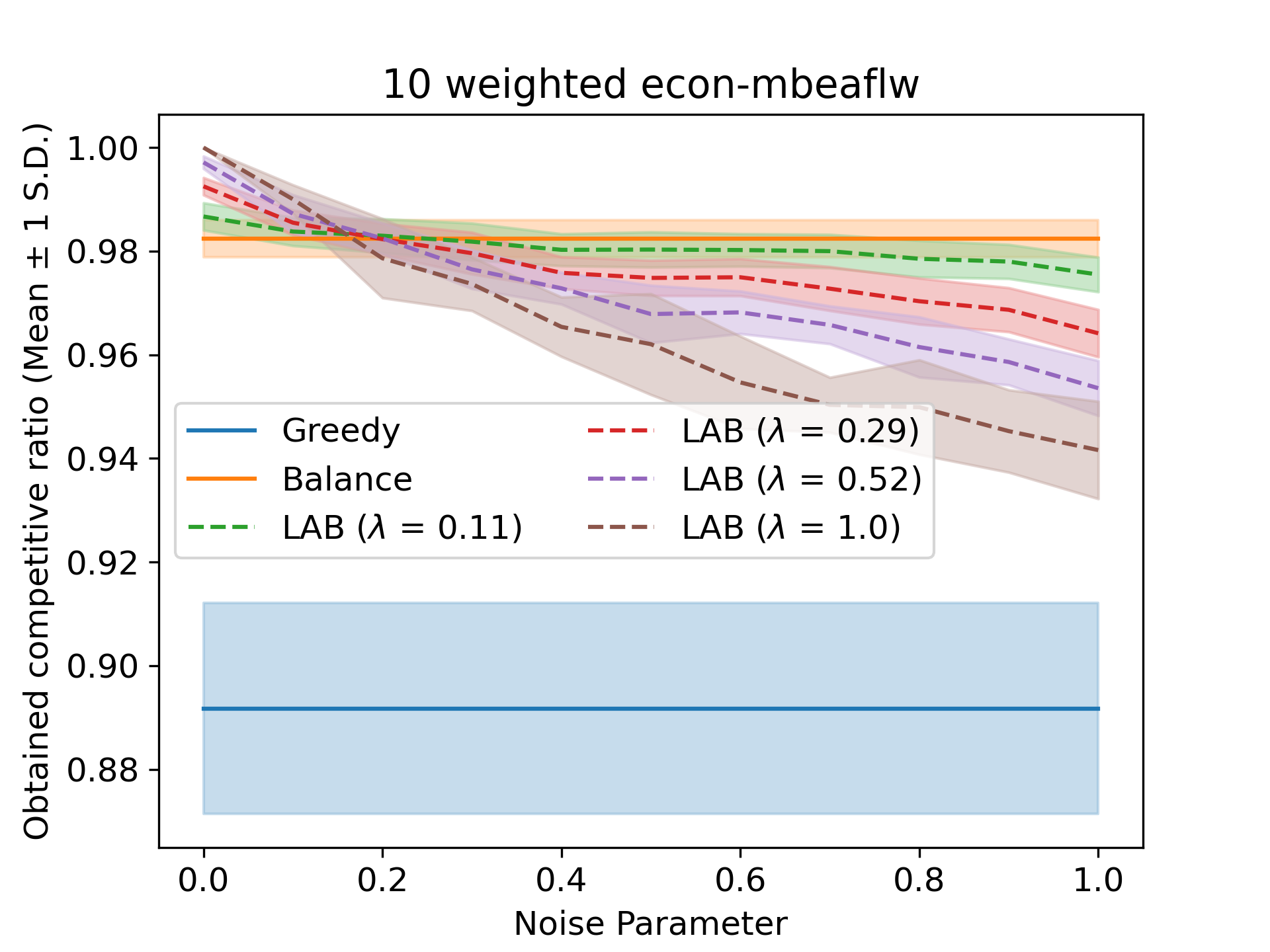}
    \end{subfigure}
    \begin{subfigure}{0.32\textwidth}
        \centering
        \includegraphics[width=\textwidth]{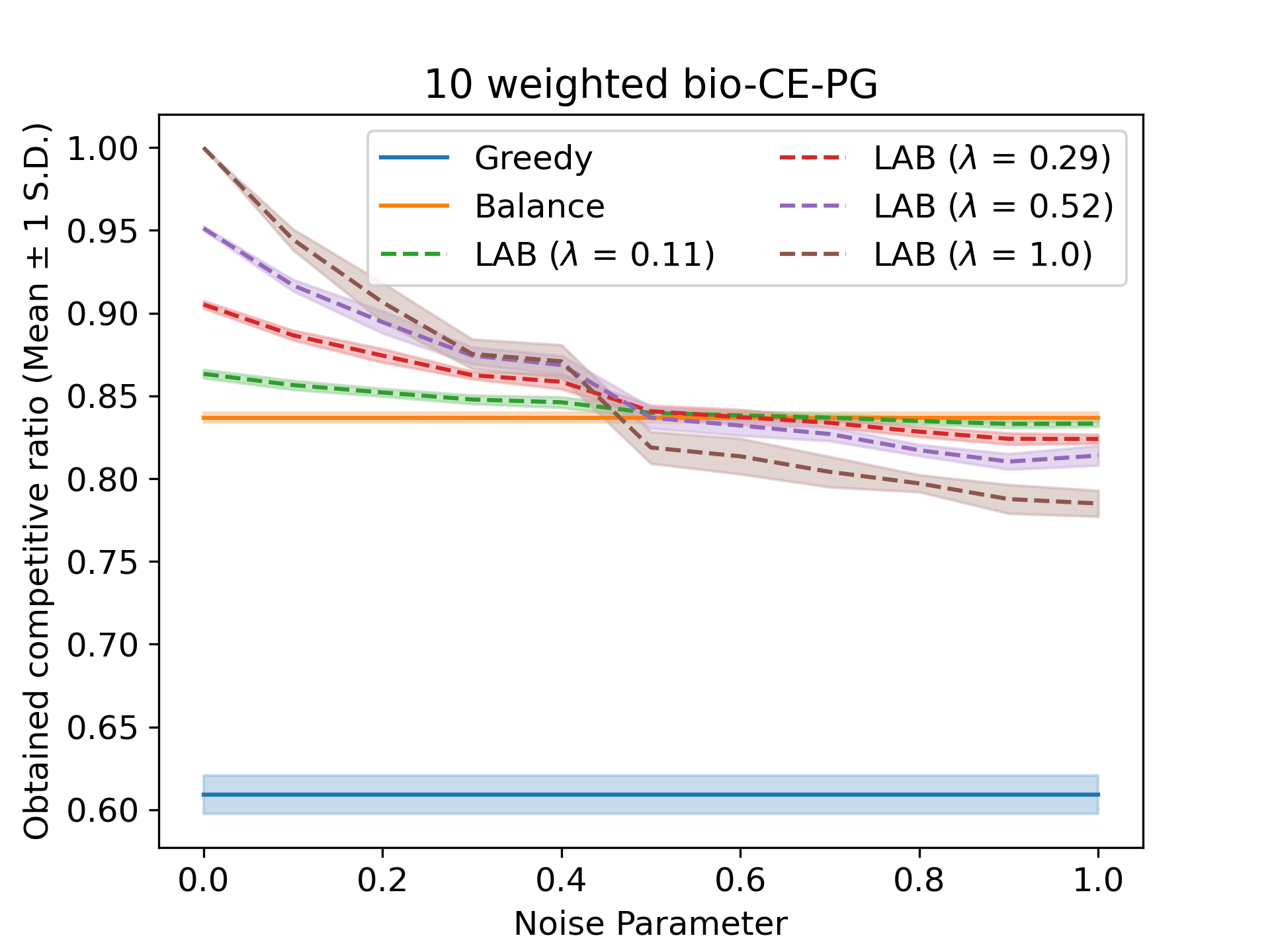}
    \end{subfigure}
    \caption{Empirical results for unweighted real-world graph instances.}
    \label{fig:exp-real-weighted}
\end{figure}

\end{document}